\documentclass[a4paper,onecolumn,11pt,accepted=2023-06-25]{quantumarticle}

\pdfoutput=1

\usepackage[shortlabels]{enumitem}

\usepackage{palatino} 
\usepackage{braket}
\usepackage{amsfonts}
\usepackage{amssymb}
\usepackage{amsmath}
\usepackage{mathrsfs} 
\usepackage{bbm}
\usepackage{latexsym}
\usepackage{amsthm}
\usepackage[usenames]{color}
\usepackage{etoolbox}
\usepackage{diagbox}
\usepackage{subcaption}
\usepackage{verbatim}
\usepackage{ragged2e, graphicx, pifont}
\usepackage{indentfirst}
\usepackage{tikz-cd}
\usetikzlibrary{decorations.pathmorphing}
\usetikzlibrary{arrows.meta}
\usepackage{mdframed}
\usepackage{mathtools} 
\usepackage{tabularx}
\usepackage[font=small,labelfont=bf]{caption}

\usepackage{accents}
\DeclareRobustCommand{\bigdot}[1]{{\accentset{\mbox{\large\bfseries .}}{#1}}}
\DeclareRobustCommand{\bigddot}[1]{{\accentset{\mbox{\large\bfseries .\hspace{-0.05ex}.}}{#1}}}

\usepackage[
backend=bibtex,
style=alphabetic,
giveninits=true, 
maxbibnames=99,
minalphanames=3,
date=year
]{biblatex}
\usepackage{hyperref}
\usepackage{doi} 
\usepackage{algorithm} 
\usepackage[noend]{algpseudocode}

\hypersetup{
colorlinks = true,
citecolor= blue,
linkcolor= blue
}

\renewbibmacro{in:}{}

\AtEveryBibitem{\clearfield{issue}} 

\newbibmacro{string+doiurl}[1]{%
	\iffieldundef{doi}{%
		\iffieldundef{url}{
			#1%
		}{%
			\href{\thefield{url}}{#1}%
		}%
	}{%
		\href{http://doi.org/\thefield{doi}}{#1}%
	}%
}
\DeclareFieldFormat{title}{\usebibmacro{string+doiurl}{\mkbibemph{#1}}}
\DeclareFieldFormat[article,thesis,incollection,inproceedings,manual]{title}%
{\usebibmacro{string+doiurl}
	{#1} 
}
\floatname{algorithm}{Protocol}

\newcommand\Algphase[1]{
\vspace*{-.5\baselineskip}\Statex\hspace*{\dimexpr-\algorithmicindent-2pt\relax}
\Statex\hspace*{-\algorithmicindent}\textbf{#1}
\vspace*{-.8\baselineskip}\Statex\hspace*{\dimexpr-\algorithmicindent-2pt\relax}
}

\newcommand{\alglinelabel}{%
  \addtocounter{ALG@line}{-1}
  \refstepcounter{ALG@line}
  \label
}

\makeatletter
\renewcommand\paragraph{\@startsection{paragraph}{4}{\z@}%
                                      {\parskip}
                                      {-1em}%
                                      {\normalfont\normalsize\bfseries}}
\makeatother

\hypersetup{pdfpagemode=UseNone}

\newtheorem{theorem}{Theorem}
\newtheorem{lemma}[theorem]{Lemma}

\newtheorem{definition}{Definition}
\newtheorem{altdef}{Alternative Definition}

\newtheorem{remark}{Remark}
\newtheorem{conj}[theorem]{Conjecture}
\newtheorem{fact}[theorem]{Fact}

\newcommand{\norm}[1]{\ensuremath{\left\lVert #1 \right\rVert}}

\newcommand{\eps}{\varepsilon}
\newcommand{\inner}[2]{\langle #1, #2 \rangle}
\newcommand{\matel}[3]{\langle #1 | #2 | #3\rangle}

\newcommand{\clA}{\mathcal{A}}
\newcommand{\clB}{\mathcal{B}}
\newcommand{\clC}{\mathcal{C}}

\newcommand{\clE}{\mathcal{E}}
\newcommand{\clF}{\mathcal{F}}

\newcommand{\clH}{\mathcal{H}}

\newcommand{\clO}{\mathcal{O}}

\newcommand{\clR}{\mathcal{R}}
\newcommand{\clS}{\mathcal{S}}
\newcommand{\clT}{\mathcal{T}}
\newcommand{\clU}{\mathcal{U}}
\newcommand{\clV}{\mathcal{V}}

\newcommand{\clX}{\mathcal{X}}
\newcommand{\clY}{\mathcal{Y}}
\newcommand{\clZ}{\mathcal{Z}}

\newcommand{\sfB}{\mathsf{B}}
\newcommand{\sfD}{\mathsf{D}}
\newcommand{\sfF}{\mathsf{F}}
\newcommand{\sfH}{\mathsf{H}}
\newcommand{\sfI}{\mathsf{I}}
\newcommand{\sfP}{\mathsf{P}}
\newcommand{\sfQ}{\mathsf{Q}}
\newcommand{\sfS}{\mathsf{S}}
\newcommand{\sfV}{\mathsf{V}}

\newcommand{\Id}{\mathbbm{1}}

\newcommand{\bbE}{\mathop{\mathbb{E}}}
\newcommand{\mfH}{\mathfrak{H}}

\newcommand{\A}{\mathrm{A}}
\newcommand{\B}{\mathrm{B}}

\newcommand{\E}{\mathrm{E}}

\newcommand{\tA}{\widetilde{A}}
\newcommand{\tB}{\widetilde{B}}
\newcommand{\tE}{\widetilde{E}}
\newcommand{\tK}{\widetilde{K}}
\newcommand{\tM}{\widetilde{M}}
\newcommand{\tX}{\widetilde{X}}
\newcommand{\tY}{\widetilde{Y}}
\newcommand{\tZ}{\widetilde{Z}}

\newcommand{\hA}{\widehat{A}}
\newcommand{\hB}{\widehat{B}}
\newcommand{\oC}{\overline{C}}

\newcommand{\vph}{\varphi}

\newcommand{\ketbra}[2]{\ket{#1}\!\!\bra{#2}}
\newcommand{\state}[1]{\ketbra{#1}{#1}}

\newcommand{\cmark}{\ding{51}\,}
\newcommand{\xmark}{\ding{55}\,}

\newcommand{\MS}{\mathrm{MS}}
\newcommand{\MSB}{\mathrm{MSB}}
\newcommand{\MSE}{\mathrm{MSE}}
\newcommand{\syn}{\mathtt{syn}}

\newcommand{\supp}{\operatorname{supp}}
\newcommand{\Tr}{\operatorname{Tr}}
\newcommand{\mval}{m}
\newcommand{\hBrho}{\rho^1}
\newcommand{\thBrho}{\widetilde{\rho}^1}
\newcommand{\Frho}{\rho^2}

\newcommand{\lcom}{\lambda_\mathrm{com}}
\newcommand{\lCI}{\lambda_\mathrm{CI}}
\newcommand{\lEC}{\lambda_\mathrm{EC}}

\newcommand{\conv}{\chi}
\newcommand{\prot}{\Pi}
\newcommand{\simu}{\Sigma}
\newcommand{\freal}{\clF^\mathrm{real}}
\newcommand{\fideal}{\clF^\mathrm{ideal}}
\newcommand{\fR}[2]{\prot^{#1}\freal_{#2}}
\newcommand{\fI}[1]{\fideal_{#1}\simu^{#1}}

\newcommand{\inset}[1]{\begin{mdframed}[backgroundcolor=black!10,linecolor=white]#1\end{mdframed}}

\definecolor{dgreen}{rgb}{0.0, 0.6, 0}

\allowdisplaybreaks

\title{Composably secure device-independent encryption with certified deletion}
\author{Srijita Kundu}
\affiliation{Institute for Quantum Computing and Department
of Combinatorics and Optimization, University of Waterloo, Waterloo, Ontario N2L 3G1, Canada.}
\email{srijita.kundu@uwaterloo.ca}
\orcid{0000-0002-8630-0113}

\author{Ernest Y.-Z. Tan}
\affiliation{Institute for Theoretical Physics, ETH Z\"{u}rich, Switzerland.}
\affiliation{Institute for Quantum Computing and Department
of Physics and Astronomy, University of Waterloo, Waterloo, Ontario N2L 3G1, Canada.}
\email{yzetan@uwaterloo.ca}
\orcid{0000-0003-4872-158X}

\begin{document}
\maketitle
\begin{abstract}
We study the task of encryption with certified deletion (ECD) introduced by Broadbent and Islam \cite{BI19}, but in a device-independent setting: we show that it is possible to achieve this task even when the honest parties do not trust their quantum devices. Moreover, we define security for the ECD task in a composable manner and show that our ECD protocol satisfies conditions that lead to composable security. Our protocol is based on device-independent quantum key distribution (DIQKD), and in particular the parallel DIQKD protocol based on the magic square non-local game, given by Jain, Miller and Shi \cite{JMS17}. To achieve certified deletion, we use a property of the magic square game observed by Fu and Miller \cite{FM17}, namely that a two-round variant of the game can be used to certify deletion of a single random bit. In order to achieve certified deletion security for arbitrarily long messages from this property, we prove a parallel repetition theorem for two-round non-local games, which may be of independent interest.
\end{abstract}

\section{Introduction}
Consider the following scenario: Alice wants to send a message to Bob that is secret from any third party. She may do this by sending Bob a ciphertext which contains the message encrypted with a secret key, such that when the key is revealed to Bob he may learn the message. Now suppose after sending the ciphertext Alice decides that she does not want Bob to learn the message after all, but she cannot prevent the secret key from eventually being revealed to him. So Alice wants to encrypt the message in such a way that she can ask Bob for a \emph{deletion certificate} if she changes her mind. If Bob sends a valid deletion certificate, Alice can be convinced that Bob has indeed deleted his ciphertext and cannot hereafter learn the message even if the secret key is revealed to him. In this scenario Alice is not actually forcing Bob to delete the ciphertext, but she is making sure that he cannot simultaneously convince her that he has deleted the ciphertext, and also learn the message.

An encryption scheme for the above scenario is called \emph{encryption with certified deletion} (ECD) and was introduced by Broadbent and Islam \cite{BI19}. It is easy to see that ECD cannot be achieved with a classical ciphertext: since classical information can always be copied, any deletion certificate Bob sends to Alice can only convince her that he has deleted one copy of it -- he may have kept another copy to decrypt from, when he learns the key. However, quantum states cannot in general be copied, and are disturbed by measurements. So if Bob has a quantum ciphertext that he cannot copy, and needs to perform a measurement on it to produce a deletion certificate, the state may be disturbed to such an extent that it is no longer possible to recover the message from it, even with the key.

The no-cloning property and the fact that measurements disturb quantum states have been useful for various cryptographic tasks, such as quantum key distribution (QKD) \cite{BB84} and unforgeable quantum money \cite{Wie83}. The concept of \emph{revocable timed-release encryption} --- a task which has some similarities to encryption with certified deletion --- was studied by Unruh \cite{Unr14}, who showed it can be achieved with quantum encodings. Another related task of tamper-evident delegated quantum storage  --- here Alice wants to store data that she encrypts using a short key on a remote server, so that she can retrieve it later and also detect if the server has tampered with it --- was studied by van der Vecht, Coiteaux-Roy and \v{S}kori\'{c} \cite{VCS20}. L\"{u}tkenhaus, Marwah and Touchette \cite{LMT20} studied a different form of delegated storage, where Alice commits to a single random bit that Bob can learn at some fixed time, or she can erase, using a temporarily trusted third party. Finally, the ECD task itself, as mentioned before, was introduced by Broadbent and Islam, who achieved it using Wiesner's encoding scheme \cite{Wie83}.

All of the works mentioned above are in the device-dependent setting, where the honest parties trust either the quantum states that are being used in the protocol, or the measurement devices, or both. However, in general a sufficiently powerful dishonest party may make the quantum state preparation and measurement devices used in a protocol behave however they want. As it turns out, with some mild assumptions it is possible to achieve certain cryptographic tasks even in this scenario. There is a long line of works studying the device-independent security of QKD \cite{PAB+09, ADF+18, JMS17}. Device-independent protocols for two-party cryptographic tasks such as coin flipping \cite{ACK+14}, bit commitment \cite{SCA+11, AMPS16} and XOR oblivious transfer \cite{KST22} have also been shown. Fu and Miller \cite{FM17} studied the task of sharing between two parties a single random bit, which can be certifiably deleted, in the device-independent setting.

A desirable property of cryptographic protocols is that they should be \emph{composable}, meaning that if a protocol is used as part of a larger protocol to achieve some more complex task, then security of the larger protocol should follow from the security of its constituent protocols. While it is possible to achieve composable security of 
various cryptographic tasks such as QKD
\cite{BHL+05, PR14}, this is in general not so easy to achieve for many examples, such as the others mentioned above.

\subsection{Our contributions}\label{sec:results}
Informally stated, our main contributions in this work are:
\begin{enumerate}
\item We define the ECD task and its security in a composable manner.
\item We give a quantum protocol (with information-theoretic security) for the ECD task that satisfies certain properties of correctness, completeness and secrecy, even when the honest parties do not trust their own quantum devices.
\item We show how to prove that a protocol that satisfies the above properties achieves the ECD task in a composably secure manner.
\end{enumerate}
The reason we do not combine items 2 and 3 above to make the claim that we give a protocol that achieves the ECD task in a composably secure manner is because the notion of device-independence itself has not been precisely formalized in a composable manner yet. So item 3 in the device-independent setting is conditional on a conjecture that we shall soon explain (though note that even without this conjecture, our proof already shows that the protocol is indeed composably secure in the standard device-dependent setting, 
i.e.~if one imposes the condition that the honest parties are performing trusted measurements). In contrast, our proof that our protocol satisfies the security properties in item 2 holds under standard device-independent conditions, without additional conjectures.

Our composable security definition uses the framework of \emph{Abstract Cryptography} introduced by Maurer and Renner \cite{MR11}. In the Abstract Cryptography framework, a \emph{resource} is an abstract system with an interface available to each party involved, to and from which they can supply some inputs and receive some outputs. A \emph{protocol} uses some resources (meaning it interacts with the interfaces of such resources) in order to construct new resources. The protocol is said to construct the new resource in a composably secure manner if it is not possible to tell the ideal resource apart from the protocol acting on the resources it uses, under certain conditions. As such, a composable security definition for a cryptographic task would be the description of a reasonable (in the sense of being potentially achievable by actual protocols) ideal functionality or resource corresponding to that task, and a composable security proof for a protocol for this task would show that the constructed resource and ideal resource are indistinguishable.

We model the notion of a device-independent resource in the Abstract Cryptography framework as
a resource which supplies some black boxes representing quantum states to the honest parties, and the honest parties may press some buttons on these boxes to obtain some classical outputs. However, the resource allows the boxes themselves to be chosen by a dishonest third party Eve, and they produce the outputs by implementing whatever states and measurements Eve wants. 

Strictly speaking, to avoid the \emph{memory attack} in the device-independent setting described in \cite{BCK13}, some additional constraints need to be placed on the registers that the measurements act on. Namely, one has to impose the condition that the measurements cannot access any registers storing private information from previous (potentially unrelated) protocols. Such a condition is implicitly imposed, albeit not always obvious,
in the standard (device-dependent) framework for Abstract Cryptography~\cite{MR11,PR14}, where the measurements are assumed to be fully characterized (and thus the registers which the measurements act on can be specified to be independent of previous protocols). However, the question of precisely formalizing this condition in the device-independent setting has not been completely resolved, and is currently a topic of active research. For the purposes of this work, we consider the technical treatment of this subject to be beyond our scope, and for ease of presentation we shall proceed under the assumption that it will be possible to find an appropriate such formulation in the device-independent setting.
That is, we shall assume the following somewhat informal conjecture and prove our main theorem conditional on it.
\begin{conj}\label{conj:DI}
The quantum ``boxes'' typically considered in the device-independent setting can be formalized in the Abstract Cryptography framework, in a manner that allows one to impose the conditions we have outlined above (and which we elaborate on in Section~\ref{sec:achiev}).
\end{conj} 

We can also introduce a formalization of the notion that Alice cannot prevent the decryption key from leaking to Bob in the ECD setting. We model this as follows: Alice has access to a \emph{trusted temporarily private randomness source} --- meaning it supplies random variables with any requested distribution, but after some fixed time it will make public any randomness it provided. 
Furthermore we impose the constraint that after some point in time (and before decryption takes place), Alice no longer has any communication channels to Bob, and thus her only way for Bob to learn the effective ``decryption key'' is through the public broadcast made by this temporarily private randomness source --- note that this broadcast still occurs even if Alice wants to revoke Bob's ability to learn the message, thereby formalizing the notion that the ``decryption key'' is eventually leaked. 
For technical reasons regarding the anchoring-based proof, we also need Alice to have a small supply of local randomness that remains private (it does not need to be announced even for decryption).
Overall, our protocol constructs the ECD resource using only these randomness sources, untrusted quantum boxes, and an authenticated classical channel (which only lasts until the aforementioned time) --- all of which are fairly weak.

Aside from device-independence and composability, our ECD security definition and protocol construction has the following advantage over the ones in \cite{BI19}: we consider the possibility of Bob being honest and consider security against a third party eavesdropper Eve when this is the case. This helps motivate the encryption aspect of the ECD task: if Alice did not need to conceal the message from Eve, she could have waited until she was sure whether she trusts Bob or not, and then sent the message as plaintext. Security against Eve is not considered in \cite{BI19}, and indeed in their protocol Eve may be able to learn the message whenever Bob does.

We note however that making the protocol secure against Eve comes at the cost of making it interactive. In the \cite{BI19} protocol, Bob receives the ciphertext in one round from Alice, whereas in our protocol, Alice sends a message to Bob and Bob replies back with a message before Alice can send the ciphertext (or abort the protocol). Furthermore, to prove this security property we do need to impose a somewhat non-standard condition regarding honest Bob's boxes; however, we highlight that this condition is only used to prove security against Eve, and is not required in any situation where Bob is dishonest (see Sec.~\ref{sec:achiev}--\ref{sect:sec-disc} for further details). 

A more minor difference is that we think it makes more intuitive sense for the ECD task to include a further message from Alice to Bob in which she tells Bob whether she wants him to delete his ciphertext or not (and potentially provides additional information he needs in order to do so), 
and hence we present our protocol as including this communication from Alice to Bob. However, the protocol and security proof can easily be modified to consider a version in which Alice sends the information required for deletion from the start, and Bob alone makes the decision of whether to delete the ciphertext (this version would essentially match the \cite{BI19} protocol). 

\subsection{Our techniques}
\subsubsection{Constructing the DI ECD protocol}
All device-independent security proofs are based on non-local games. One approach towards constructing such proofs is to use the property known as \emph{self-testing} or \emph{rigidity} displayed by certain non-local games. Specifically, suppose we play a non-local game with boxes implementing some unknown state and measurements, and in fact even the dimension of the systems are unspecified. If these state and measurements regardless achieve a winning probability for the game that is close to its optimal winning probability, then self-testing tells us that the state and measurements are close to the ideal state and measurements for that game, up to trivial isometries. For DIQKD, this means in particular that the measurement outputs of the devices given the inputs are random, i.e., they cannot be predicted by a third party even if they have access to the inputs used. This lets us use the outputs of the devices to produce a secret key.

\paragraph{Parallel DIQKD protocol.} We make use of the parallel DIQKD protocol given by Jain, Miller and Shi \cite{JMS17}, and its subsequent simplification given by Vidick \cite{Vid17}, based on the magic square non-local game. In the magic square game, henceforth deonoted by $\MS$, Alice and Bob respectively receive trits $x$ and $y$, and they are required to output 3-bit strings $a$ and $b$, which respectively have even parity and odd parity, and satisfy $a[y] = b[x]$. The classical winning probability of MS is 8/9, whereas the quantum winning probability is 1. The \cite{JMS17} protocol works as follows: Alice and Bob have boxes which can play $l$ many instances of MS. Using trusted private randomness, Alice and Bob generate i.i.d. inputs $x,y$ for each of their boxes and obtain outputs $a,b$ (which are not necessarily i.i.d.). The inputs $x,y$ are then publicly communicated. Alice and Bob select a small subset of instances on which to communicate their outputs and test if the MS winning condition is satisfied (up to some error tolerance) on those instances. If this test passes, then they go ahead and select their common bits $a[y] = b[x]$ from all the instances --- they can do this since Alice has $a$, Bob has $b$, and they both have $x,y$ --- as their raw secret key (some \emph{privacy amplification} of the raw key is required in order to get the final key). Otherwise, the protocol aborts.

If the MS winning condition is satisfied on the test instances with high probability, then self-testing says that the states and measurements are close to the ideal ones for MS; but this property is not directly used in the security proof. In the version of the security proof given by \cite{Vid17}, instead a guessing game variant MSE of MS, involving three players Alice, Bob and Eve, is considered. MSE is the same as MS on Alice and Bob's parts, and additionally, Eve also gets Alice and Bob's inputs and has to guess Alice and Bob's common bit. It can be shown that MSE cannot be won with probability 1 by all three players, and in particular if Alice and Bob's winning condition is satisfied, then Eve cannot guess their common bit with high probability. Now making use of a parallel repetition theorem for the MSE game, which first requires a small transformation called \emph{anchoring} in order to make the parallel repetition proof work, we can say that Eve's guessing probability for the shared bit in $l$ many instances of MSE, is exponentially small in $l$. Since Alice and Bob's winning condition is satisfied on a random subset of instances, we can say it is satisfied on all instances with high probability by making use of a sampling lemma. Hence Eve's guessing probability for the raw key is exponentially small in $l$. Now using the operational interpretation of min-entropy, this means that the min-entropy of the raw key conditioned on Eve's quantum system is linear in $l$, and we can use privacy amplification to get a final key that looks almost uniformly random to Eve.

\paragraph{Using DIQKD for ECD.} It is easy to see how to use DIQKD to achieve the encryption aspect of ECD --- Alice and Bob can perform the DIQKD protocol to share a secret key, and then Alice can encrypt the message by one-time padding it with the key, and send it to Bob. This certainly achieves security against Eve if Bob is honest. However, in the ECD scenario, unlike in the QKD scenario, Bob may not be honest, and hence the QKD security proof may not apply, since it requires him to honestly follow the protocol. Moreover, it is not clear how to achieve the certified deletion aspect of ECD this way. Instead, we do the following for our ECD protocol: we make Alice obtain the inputs $x,y$ for MS from the trusted temporarily private randomness source, and obtain the raw key by herself using her boxes and the inputs, but she does not reveal the inputs to Bob (who hence does not have the raw key). She then one-time pads the message with the final key obtained from the raw key and sends the resulting ciphertext to Bob. Bob cannot decrypt the message at this time since he does not have the raw key, but he can get the key from his boxes and decrypt as soon as the temporarily private randomness source reveals $x,y$ to him. Hence in order to achieve certified deletion security, Alice needs to make Bob do some operation on his boxes which destroys his ability to learn the raw key even if he gets $x,y$.

We also note that for technical reasons, Alice actually needs to one-time pad the message with an extra uniformly random string $u$ that she gets from the randomness source, in addition to the final key. This makes no difference to Bob's ability to decrypt the message when all the randomness is revealed by the source or in certified deletion (since $u$ is revealed at the end), but it does potentially make a difference at intermediate stages in the protocol. Such an extra one-time pad is also used by \cite{BI19} in their protocol.

\paragraph{Achieving certified deletion security.} Fu and Miller \cite{FM17} made the following observation about the magic square game: suppose Alice does the measurement corresponding to $x$ and Bob does the measurement corresponding to $y'$ on the MS shared state, then if Bob is later given $x$ he can guess $a[y']$ perfectly as $b[x]$ from his output. But if Bob has indeed performed the $y'$ measurement, then he cannot guess the value of $a[y]$ for some $y\neq y'$, even given $x$. In fact this property holds in a device-independent manner, i.e., if Alice and Bob have boxes implementing some unknown state and measurements which are compatible with MS, and Alice and Bob input $x$ and $y'$ into their boxes and get outputs that satisfy the MS winning condition with probability close to 1, then Bob cannot later perfectly guess $a[y]$ given $x,y$. Now consider a 2-round variant of MS, which we shall call MSB: in the first round, Alice and Bob are given $x,y'$ and are required to output $a,b'$ that satisfy the MS winning condition; in the second round, Bob is given $x,y$ such that $y \neq y'$ and he is required to produce a bit equal to $a[y]$. We note that Bob can use his first round input, his first round measurement outcome, and his half of the post-measured shared state in order to produce the second round output. The \cite{FM17} observation implies that the winning probability of MSB is less than 1.

Using the same anchoring trick as in MSE, we can prove a parallel repetition theorem for the 2-round MSB game. Now to achieve certified deletion, Alice gets i.i.d. $y'\neq y$ for all $l$ instances from the randomness source, and if she wants Bob to delete his ciphertext, she sends Bob $y'$ and asks for $b'$ as a deletion certificate. If the $b'$ sent by Bob satisfies $a[y']=b'[x]$ (up to some error tolerance) then Alice accepts his deletion certificate. Due to the parallel repetition theorem for MSB, if Bob's deletion certificate has been accepted, then his guessing probability for $a[y]$, i.e., the raw key, given $x,y$, is exponentially small in $l$. Due to privacy amplification, the final key looks uniformly random to Bob, and thus the message is secret from him. We note that certified deletion security should be against Bob and Eve combined rather than just Bob, as a dishonest Bob could collude with Eve in order to try and guess the message. This is fine for our security proof approach, as we can consider Bob and Eve combined as a single party for the purposes of the MSB game.

\begin{remark}\label{remark:parallel}
Many security proofs for DIQKD work in the sequential rather than parallel setting. In the sequential setting, Alice and Bob provide inputs to and get outputs from each instance of their boxes sequentially, which limits the kinds of correlations that are possible between the whole string of their inputs and outputs. While this is easy to justify for DIQKD when Alice and Bob are both honest, justification is harder for the ECD scenario where Bob may be dishonest and need not use his boxes sequentially, so that more general correlations between his inputs and outputs are possible. Hence a parallel rather than sequential security proof is crucial for us.
\end{remark}

\subsubsection{Proving composable security}
To prove composable security for our protocol, we need to show that a \emph{distinguisher} cannot tell the protocol apart from the ideal ECD functionality, when it is constrained to interact with the ideal functionality via a \emph{simulator} acting on the dishonest parties' interfaces. That is, for any possible behaviour of the dishonest parties in the real protocol, we need to construct a simulator such that the above is true. This needs to be done for all possible combinations of honest/dishonest parties involved, and here we only describe the idea for the case when Bob is dishonest and the dishonest third party Eve is present.

Our simulator construction is inspired by the composable security proofs of QKD \cite{BHL+05,PR14}: it internally simulates the real protocol using whatever outputs it gets from the ideal functionality, so that the distinguisher is able to get states on the dishonest parties' side similar to what it would have gotten in the real protocol. However, the ideal functionality is only supposed to reveal the message $m$ to Bob at the end (if either Alice did not ask for a deletion certificate, or Alice asked for a deletion certificate and Bob did not produce a valid one), and since the simulator needs to simulate the real protocol before this time, it has to instead release a dummy ciphertext that does not depend on the message. Hence we require that the states on the dishonest parties' side corresponding to $m$ and the dummy ciphertext in the protocol be indistinguishable, if the message has not been revealed. This is related to the security notions of \emph{ciphertext indistinguishability} and \emph{certified deletion security} considered in \cite{BI19}. But these properties hold only as long as the protocol does not actually reveal $m$. If $m$ is actually revealed at the end, the simulator needs to fool the distinguisher into believing it originally released the ciphertext corresponding to $m$. This is where the extra one-time pad $u$ we use comes in handy: the simulator can
edit the value on the one-time pad register to a value compatible with the true message $m$.

Overall, our security proof is fairly ``modular'': our simulator construction for dishonest Bob and Eve works for any protocol in which the extra $u$ OTP is used and which satisfies the ciphertext indistinguishability and certified deletion security properties (jointly called \emph{secrecy}). For other combinations of honest/dishonest parties, the proof works for any protocol that satisfies notions of \emph{completeness} and \emph{correctness}, even for devices with some small noise. Completeness here means that if all parties are honest then the protocol aborts with small probability, and Bob's deletion certificate is accepted by Alice with high probability; correctness means that an honest Bob can recover the correct message from the quantum ciphertext with high probability.

\subsubsection{Proving parallel repetition for 2-round games}
As far as we are aware, our proof of the  parallel repetition theorem for the MSB game is the first parallel repetition result for 2-round games, which may be of independent interest. First we clarify what we mean by a 2-round game: in the literature, boxes that play multiple instances of a game, whether sequentially or in parallel, are sometimes referred to as multi-round boxes, and certainly the nomenclature makes sense in the sequential setting. However, the two rounds for us are not two instances of the same game --- they both constitute a single game and in particular, the outputs of the second round are required to satisfy a winning condition that depend on the inputs and outputs of the first round. Alice and Bob share a single entangled state at the beginning of the game, and the second round outputs are obtained by performing a measurement that can depend on the first round inputs and outputs in addition to the second round inputs, on the post-measured state from the first round. This is what we refer to as a 2-round game; it can be viewed as a specific type of interactive game.

We actually prove a parallel repetition theorem for a wider class of 2-round games than just the anchored MSB game; namely, what we call \emph{product-anchored games}. This captures elements of both product games and anchored games, whose parallel repetition has been studied for 1-round games \cite{JPY14, BVY15, BVY17, JK20} (although we consider only a specific form of anchoring which is true of the MSB game --- anchored distributions can be more general), and our proof is inspired by techniques from proving parallel repetition for both product and anchored 1-round games. We call a 2-round game product-anchored iff the first round inputs $x,y$ are from a product distribution, and in the second round, only Bob gets an input $z$ which takes a special value $\perp$ with constant probability such that the distribution of $x,y$ conditioned on $z=\perp$ is the same as their marginal distribution, and otherwise $z=(x,y')$ (where $y'$ may be arbitrarily correlated with $x,y$). The first and second round outputs are $(a,b)$ and $b'$ respectively.\footnote{In this notation we switch around the roles of $y, y', b, b'$ as compared to our definition of $\MSB$. We do this in order to make our notation more compatible with standard parallel repetition theorems. As this definition refers to a wider class of games than just $\MSB$, we hope this will not cause any confusion.}

We use the information theoretic framework for parallel repetition established by \cite{Raz95, Hol09}: we consider a strategy $\clS$ for $l$ instances of the game $G$, condition on the event $\clE$ of the winning condition being satisfied on some $C\subseteq[l]$ instances, and show that if $\Pr[\clE]$ is not already small, then we can find another coordinate in $i\in\oC $ where the winning probability conditioned on $\clE$ is bounded away from 1. For 1-round games (where there is no $z_i$), this is done in the following way: Alice and Bob's overall state in $\clS$ conditioned on $\clE$ is considered; this state depends on Alice and Bob's inputs --- suppose it is $\ket{\vph}_{x_iy_i}$ when Alice and Bob's inputs in the $i$-th coordinate are $(x_i,y_i)$. 
We then argue that there exists some coordinate $i$ and unitaries $\{U_{x_i}\}_{x_i}, \{V_{y_i}\}_{y_i}$ acting on Alice and Bob's registers respectively, such that the operator $U_{x_i}\otimes V_{y_i}$ brings some shared initial state close to the state $\ket{\vph}_{x_iy_i}$. (In the product case, this shared initial state would be a superposition of the states $\ket{\vph}_{x_iy_i}$, weighted according to the distributions of $x_i$ and $y_i$.)
Hence, unless the winning probability in the $i$-th coordinate is bounded away from 1, Alice and Bob can play a single instance of $G$ by sharing this initial state, performing $U_{x_i}, V_{y_i}$ on it on inputs $(x_i,y_i)$, and giving the measurement outcome corresponding to the $i$-th coordinate on the resulting state; the winning probability of this strategy would then be higher than the optimal winning probability of $G$ --- a contradiction.

For 2-round games, the state conditioned on success depends on all three inputs $x_iy_iz_i$, and Alice and Bob obviously cannot perform unitaries $U_{x_i}$ and $V_{y_iz_i}$ in order to produce their first round outputs, since Bob has not received $z_i$ yet. However, we observe that Alice and Bob don't actually need the full $\ket{\vph}_{x_iy_iz_i}$ state in order to produce their first round outputs --- they only need a state whose $A_iB_i$ registers, containing their first round outputs, are close to those of $\ket{\vph}_{x_iy_iz_i}$. We observe that $\ket{\vph}_{x_iy_i\perp}$ is indeed such a state. In fact, in the unconditioned state, given $x_iy_i$, all of Alice's registers as well as all of $B$ are independent of $z_i$, as the second round unitary depending on $z_i$ does not act on these registers (the second round unitary may use $B$ as a control register, but that does not affect the reduced state of $B$). Conditioning on the high probability event $\clE$ does not disturb the state too much, and by chain rule of mutual information, we can argue that there exists an $i$ such that Alice's registers and $B$ in $\vph_{x_iy_iz_i}$ are close to those in $\vph_{x_iy_i}$ (i.e., averaged over $z_i$). Since $z_i=\perp$ with constant probability, this means that these registers are indeed close in $\vph_{x_iy_iz_i}$ and $\vph_{x_iy_i\perp}$.

Conditioned on $z_i=\perp$, the situation in the first round is identical to the product case; we can argue the same way as in the product parallel repetition proof by \cite{JPY14} that there exist unitaries $\{U_{x_i}\}_{x_i}, \{V_{y_i}\}_{y_i}$ such that $U_{x_i}\otimes V_{y_i}$ takes $\ket{\vph}_{\perp}$ close to $\ket{\vph}_{x_iy_i\perp}$. Now we use the fact that Alice's registers and $B$ are close in $\vph_{x_iy_iz_i}$ and $\vph_{x_iy_i\perp}$ once again to argue that there exist unitaries $\{W_{x_iy_iz_i}\}_{x_iy_iz_i}$ acting on Bob's registers except $B$ that take $\ket{\vph}_{x_iy_i\perp}$ to $\ket{\vph}_{x_iy_iz_i}$. We notice that $W_{x_iy_iz_i}$ is in fact just $W_{y_iz_i}$, because either $z_i$ contains $x_i$ or it can just be the identity, which means Bob can use $W_{y_iz_i}$ as his second round unitary. Moreover, these unitaries commute with the measurement operator on the $A_iB_i$ registers, hence $W_{y_iz_i}$, acting on the post-measured $\ket{\vph}_{x_iy_i\perp}$ also takes it to the post-measured $\ket{\vph}_{x_iy_iz_i}$. Thus the distribution Bob would get by measuring $B'_i$ after applying $W_{y_iz_i}$ on his post $A_iB_i$ measurement state is close to the correct distribution of $B'_i$ conditioned on any values $(a_i,b_i)$ obtained in the first round measurement. This gives a strategy $\clS'$ for a single instance of $G$, where $U_{x_i}, V_{y_i}$ are the first round unitaries and $W_{y_iz_i}$ is Bob's second round unitary.

\subsection{Organization of the paper}
In Section \ref{sect:ideal-f} we formally describe the ideal ECD functionality we consider, and state our precise result regarding it.
In Section \ref{sect:prelim} we provide definitions and known results for the quantities used in our proofs. In Section \ref{sect:ms}, we describe the variants of the magic square non-local games and state the parallel repetition theorems for them that we use. In Section \ref{sect:protocol} we give our real ECD protocol and prove various intermediate results that help establish its composable security, which is done in Section \ref{sect:comp-sec}. Finally, in Section \ref{sect:parrep} we provide the proofs for the parallel repetition theorems stated in Section \ref{sect:ms}.

\section{Composable security definition for ECD}\label{sect:ideal-f}

In this section, we define the precise functionality we aim to achieve, in terms of the Abstract Cryptography framework. We then state our main result regarding achieving this functionality.

\subsection{Abstract cryptography}\label{sect:abs-crypt}

\newcommand{\boxwidth}{2}
\newcommand{\simwidth}{2}
\newcommand{\arrowwidth}{1.2}
\newcommand{\offset}{-\boxwidth/2}
\tikzstyle{wrapbox} = [text width=2cm, text centered]

We first briefly state the concepts and definitions we require from the Abstract Cryptography framework. Note that this is not intended to be a complete introduction to all aspects of the framework, as that would be fairly extensive and outside the scope of this work --- if more in-depth or pedagogical explanations are required, refer to e.g.~\cite{MR11,PR14,VPdR19}.

\subsubsection{Resources, converters, and distinguishers}

In this framework, a \emph{resource} is an abstract system with an interface available to each party, to and from which they can supply some inputs and receive some outputs. 
Qualitatively, the idea of this framework is to model how protocols convert some ``real'' resources available to the various parties into more ``ideal'' functionalities for some cryptographic tasks. 
To formalize this notion, let us lay out the basic setup. We first assume that there is a fixed set $Q$ of parties that can choose to be either honest or dishonest. As an example, in QKD this set $Q$ can be taken to consist of Eve only, while Alice and Bob are always honest. As another example, for two-party ``distrustful'' cryptographic tasks such as oblivious transfer or bit commitment, $Q$ consists of both Alice and Bob, who can each potentially be dishonest. For our certified deletion protocol, we will take this set $Q$ to consist of Bob and Eve only; Alice is always honest. 

Having fixed this set $Q$, we introduce the following notation: $
(\clF_P)_{P\subseteq Q}$ 
is a tuple indexed by subsets $P\subseteq Q$, where for each $P\subseteq Q$ (i.e.~each possible subset of the potentially-dishonest parties), $\clF_P$ denotes the resources available when the parties $P$ are behaving dishonestly. 
Notice that this means we often work with {tuples} of resources rather than just a single specific resource; this is intended to capture the notion that these resources may have different functionalities depending on whether each of the parties is behaving honestly or dishonestly. (In principle one might want to consider only tuples $(\clF_P)_{P\subseteq Q}$ where the resource $\clF_P$ is the same regardless of the subset $P$ of dishonest parties. However, this sometimes turns out to be overly restrictive, hence for generality we allow ourselves the flexibility to consider tuples where each resource $\clF_P$ is different for different subsets $P$.)

With this in mind, let $\left(\freal_P\right)_{P\subseteq Q}$ denote the resource tuple describing the ``real'' functionalities available to the various parties. Analogously, let $\left(\fideal_P\right)_{P\subseteq Q}$ denote the resource tuple describing the ``ideal'' functionalities we would like to achieve. In order to do so, we would informally like our protocol to ``interact with'' $\left(\freal_P\right)_{P\subseteq Q}$ in such a way that it is transformed to $\left(\fideal_P\right)_{P\subseteq Q}$. To turn this into a full security definition, we now discuss the notion of \emph{converters}, which can interact with a resource to produce a new resource. 

A converter is an abstract system with an ``inner'' and ``outer'' interface, with the inner interface being connected to the resource interfaces, and the outer interface becoming the new interface of the resulting resource. If $P$ is a subset of the parties and we have a converter $\conv^{P}$ that connects to their interfaces in a resource $\clF$, we shall denote this as $\conv^{P}\clF$
or $\clF\conv^{P}$ (the ordering has no significance except for readability).
Each converter describes how that party interacts with its interfaces in $\clF$, producing a new set of inputs and outputs ``externally'' (i.e.~at the outer interface). If we have (for instance) a protocol with converters $\prot^{\A}$ and $\prot^{\B}$ for parties $\A$ and $\B$, for brevity we shall use $\prot^{\A\B}$ to denote the converter obtained by attaching both the converters $\prot^{\A}$ and $\prot^{\B}$.
An important basic example of converters arises from protocols themselves, which we shall now describe more thoroughly (also refer to Figure~\ref{fig:protexample} for a schematic depiction). 

Explicitly, we shall model a protocol as a tuple $\mathscr{P} = (\prot^{\A},\prot^{\B},\dots)$ of converters, one for each party. 
When trying to convert a real resource tuple $\left(\freal_P\right)_{P\subseteq Q}$ to an ideal resource tuple $\left(\fideal_P\right)_{P\subseteq Q}$, these converters would have inner interfaces that connect to $\freal_{\{\}}$ (i.e.~the real resource for the case when all parties choose to behave honestly), while their outer interfaces are required to have the same structure as those in $\fideal_{\{\}}$. (For the purposes of formalizing protocols we will not need to consider the resources for the cases where some parties behave dishonestly; we shall return to discussing such behaviour when introducing the security definitions shortly below.)
We give a schematic depiction of this in Figure~\ref{fig:protexample}.


\tikzset{pics/.cd,
idarrowsleft/.style={code={
\draw[blue,thick,->] (0,2) -- node[midway,yshift=0.3cm] {
} (\arrowwidth,2);
\draw[blue,thick,<-] (0,0) -- node[midway,yshift=0.3cm] {
} (\arrowwidth,0);
}}}

\tikzset{pics/.cd,
idarrowsright/.style={code={
\draw[blue,thick,<-] (0,0) -- node[midway,yshift=0.3cm] {
} (\arrowwidth,0);
\draw[blue,thick,->] (0,-2) -- node[midway,yshift=0.3cm] {
} (\arrowwidth,-2);
}}}

\tikzset{pics/.cd,
realarrowsleft/.style={code={
\foreach \i in {1,...,3} {\draw[thick,<-] (0,-3.5+2*\i) -- (\arrowwidth,-3.5+2*\i);}
\foreach \i in {1,...,3} {\draw[thick,->] (0,-4.5+2*\i) -- (\arrowwidth,-4.5+2*\i);}
}}}

\tikzset{pics/.cd,
realarrowsright/.style={code={
\foreach \i in {1,...,3} {\draw[thick,<-] (0,-4+2*\i) -- (\arrowwidth,-4+2*\i);}
\foreach \i in {2,...,3} {\draw[thick,->] (0,-5+2*\i) -- (\arrowwidth,-5+2*\i);}
}}}

\begin{figure}
\centering
\begin{minipage}{.45\textwidth}
\begin{tikzpicture}[scale=0.7]
\draw[rounded corners] (\offset,-3) rectangle (\offset+\boxwidth,3) node[midway,yshift=0cm] {
$\freal_{\{\}}$
} ;

\pic at (\offset-\arrowwidth,0) [scale=0.7]{realarrowsleft};
\pic at (\offset+\boxwidth,0) [scale=0.7]{realarrowsright};

\draw[rounded corners] (\offset-\arrowwidth-\boxwidth,-3) rectangle (\offset-\arrowwidth,3) node[midway,yshift=0cm] {
$\prot^{\A}$
} ;
\pic at (\offset-2*\arrowwidth-\boxwidth,0) [scale=0.7,color=blue]{idarrowsleft};

\draw[rounded corners] (\offset+\arrowwidth+\boxwidth,-3) rectangle (\offset+\arrowwidth+2*\boxwidth,3) node[midway,yshift=0cm] {
$\prot^{\B}$
} ;
\pic at (\offset+\arrowwidth+2*\boxwidth,0) [scale=0.7,color=blue]{idarrowsright};

\end{tikzpicture}
\end{minipage}
\begin{minipage}{.05\textwidth}
~\\[6mm]
\centering
\end{minipage}
\begin{minipage}{.35\textwidth}
\centering
\begin{tikzpicture}[scale=0.7]
\draw[rounded corners] (\offset,-3) rectangle (\offset+\boxwidth,3) node[midway,yshift=0cm] {
$\fideal_{\{\}}$ 
};

\pic at (\offset-\arrowwidth,0) [scale=0.7,color=blue]{idarrowsleft};

\pic at (\offset+\boxwidth,0) [scale=0.7,color=blue]{idarrowsright};

\end{tikzpicture}
\end{minipage}
\caption{Schematic depiction of protocols interacting with a real resource, in a scenario with only two parties Alice and Bob (labelled as $\A,\B$). In the diagram on the left, $\freal_{\{\}}$ denotes the real resources available when all parties are honest (i.e.~the set of dishonest parties $P$ is the empty set $\{\}$). The ``inner'' interfaces of the protocol converters $\prot^{\A},\prot^{\B}$ attach to the interfaces of $\freal_{\{\}}$, converting it to a resource with interfaces given by the ``outer'' interfaces of $\prot^{\A},\prot^{\B}$ (shown in blue here). The resulting resource can be variously denoted as $\prot^{\A}\freal_{\{\}}\prot^{\B}$, $\prot^{\A\B}\freal_{\{\}}$ or $\freal_{\{\}}\prot^{\A\B}$, as discussed in the main text. In the diagram on the right, $\fideal_{\{\}}$ denotes the desired ideal resource when all parties are honest. We require the outer interfaces of the protocol converters $\prot^{\A},\prot^{\B}$ to have the same structure as those of $\fideal_{\{\}}$, i.e.~the interfaces shown in blue have the same structure in both the left and right diagrams. We describe more precisely in the rest of this section (see e.g.~Definition~\ref{def:comp}) how to formalize a notion that the resources in the left and right diagrams are ``close'' to each other.
}
\label{fig:protexample}
\end{figure}

Finally, we shall require the concept of \emph{distinguishers}.
Given two resources $\clF$ and $\clF'$, a distinguisher is a system that interacts with the interfaces of these resources, and then produces a single bit $G$ (which can be interpreted as a guess of which resource it is interacting with). For a given distinguisher, let $\sfP_{G|\clF}$ be the probability distribution it produces on $G$ when supplied with $\clF$, and analogously for $\clF'$. Its \emph{distinguishing advantage} $\lambda$ between these two resources is defined to be 
\begin{equation*}
\lambda = \left|\sfP_{G|\clF}(0) - \sfP_{G|\clF'}(0)\right| = \frac{1}{2} \norm{\sfP_{G|\clF} - \sfP_{G|\clF'}}_1.
\end{equation*}
The main way we shall make use of this concept is as a method to formalize how ``close'' two resources are to each other --- if the distinguishing advantage between the two resources is small for any possible distinguisher, it seems reasonable to say these resources are ``close''. This notion of closeness also has important operational implications, as we shall explain later in Section~\ref{sec:operational}.

\subsubsection{Security definitions}

With these concepts, we can now discuss the security definitions in this framework. 
Suppose that as mentioned above, we have a resource tuple $\left(\freal_P\right)_{P\subseteq Q}$ describing the real functionalities available to the various parties, and a protocol $\mathscr{P}$ with inner interfaces connecting to $\freal_{\{\}}$, with the informal goal of constructing a more idealized resource tuple $\left(\fideal_P\right)_{P\subseteq Q}$. 
We would like to formalize the notion of having achieved a sufficiently ``good'' conversion from $\left(\freal_P\right)_{P\subseteq Q}$ to $\left(\fideal_P\right)_{P\subseteq Q}$. 
To do so, we use the following approach: for each possible subset $P\subseteq Q$ of the potentially-dishonest parties, consider the scenario where all parties in $P$ choose to behave dishonestly. In such a scenario, it makes sense to consider the resource $\prot^{\overline{P}}\freal_P$ (here, $\overline{P}$ denotes the set complement of $P$ with respect to the set of \emph{all} parties, not just with respect to $Q$)\footnote{For instance in our certified deletion protocol between Alice, Bob and Eve (labelled as $\A,\B,\E$), 
when considering e.g.~$P=\{\E\}$ the notation $\overline{P}$ means the set $\{\A,\B\}$ (despite the fact that only we only consider Bob and Eve to be potentially dishonest, i.e.~$Q=\{\B,\E\}$).} --- this is the resource obtained when the parties in $\overline{P}$ interact with the real resource $\freal_P$ using the protocol converters $\prot^{\overline{P}}$ ``as intended'', while the parties in $P$ do not implement the protocol converters since they are behaving dishonestly. (We give a schematic depiction of this in the left diagram in Figure~\ref{fig:simexample}.)

As a first attempt, we could try saying that this resource $\prot^{\overline{P}}\freal_P$ should be required to be ``close to'' the resource $\fideal_P$, i.e.~the ideal resource for the case where all parties in $P$ behave dishonestly. Furthermore, a natural notion of ``closeness'' between resources is to say that the maximum possible distinguishing advantage is small, as discussed above.
Unfortunately, this idea does not quite work by itself, because the interfaces for the dishonest parties $P$ in the resource $\prot^{\overline{P}}\freal_P$ are simply those of the real resource $\freal_P$ (since no converters have been attached to those interfaces), which could generally be very different from those of the ideal resource $\fideal_P$. The solution for overcoming this issue turns out to be allowing the attachment of an additional converter $\simu^{P}$ (referred to as a \emph{simulator}) to the interfaces of the dishonest parties $P$ in $\fideal_P$, and only requiring $\prot^{\overline{P}}\freal_P$ to be ``close to'' $\fideal_P\simu^P$ rather than $\fideal_P$. (We depict this in Figure~\ref{fig:simexample}.) At first sight, this appears to be a rather contrived approach, since there seems to be no purpose to this simulator other than as an artificial method to make the interfaces of $\fideal_P$ the same as those of $\prot^{\overline{P}}\freal_P$. However, we shall return to this point in the next section (after formalizing these security notions as Definition~\ref{def:comp} below), and explain how this apparently contrived concept in fact leads to powerful operational implications.

\begin{figure}
\centering
\begin{minipage}{.35\textwidth}
\begin{tikzpicture}[scale=0.7]
\draw[rounded corners] (\offset,-3) rectangle (\offset+\boxwidth,3) node[midway,yshift=0cm] {
$\freal_{\B}$
} ;

\pic at (\offset-\arrowwidth,0) [scale=0.7]{realarrowsleft};
\pic at (\offset+\boxwidth,0) [scale=0.7]{realarrowsright};

\draw[rounded corners] (\offset-\arrowwidth-\boxwidth,-3) rectangle (\offset-\arrowwidth,3) node[midway,yshift=0cm] {
$\prot^{\A}$
} ;

\pic at (\offset-2*\arrowwidth-\boxwidth,0) [scale=0.7,color=blue]{idarrowsleft};

\end{tikzpicture}
\end{minipage}
\begin{minipage}{.1\textwidth}
~\\[6mm]
\centering
{\large $\approx_{\scriptscriptstyle \lambda} \;$}
\end{minipage}
\begin{minipage}{.35\textwidth}
\centering
\begin{tikzpicture}[scale=0.7]
\draw[rounded corners] (\offset,-3) rectangle (\offset+\boxwidth,3) node[midway,yshift=0cm] {
$\fideal_{\B}$ 
};

\pic at (\offset-\arrowwidth,0) [scale=0.7,color=blue]{idarrowsleft};

\pic at (\offset+\boxwidth,0) [scale=0.7,color=blue]{idarrowsright};

\tikzset{pics/.cd,
simulatorpic/.style={code={
\draw[rounded corners] (0,-3) rectangle (0+\simwidth,3) node[midway,yshift=0cm] {
$\simu^{\B}$ 
} ;
\pic at (\boxwidth,0) [scale=0.7]{realarrowsright};
\draw[rounded corners,draw=none] 
(2.9,-3) rectangle (2.9+\simwidth,3) 
node[wrapbox,midway,yshift=1.25cm] {
}
node[wrapbox,midway,yshift=-1.4cm] {
} 
;
}}}

\pic at (\offset+\arrowwidth+\boxwidth,0) [scale=0.7]{simulatorpic};

\end{tikzpicture}
\end{minipage}
\caption{Schematic depiction of one case in Definition~\ref{def:comp}, in a scenario with only two parties Alice and Bob (labelled as $\A,\B$). 
Shown here is the case where $P=\{\B\}$, i.e.~only Bob is being dishonest. On the left, we have the real resource $\freal_{\B}$ with the protocol $\prot^{\A}$ attached (this captures the notion that Alice is performing her protocol honestly). On the right, we have the ideal resource $\fideal_{\B}$ with a simulator $\simu^{\B}$ attached (the role of the simulator is described in the main text). 
Note that both the left and right sides have outer interfaces with the same structure.
Definition~\ref{def:comp} requires both sides to be ``$\lambda$-close'' in the sense that the distinguishing advantage between them is at most $\lambda$ for any distinguisher. Similarly, Definition~\ref{def:comp} also requires ``$\lambda$-closeness'' for any other subset $P$ of the potentially dishonest parties and the corresponding analogously defined resources. (In fact the case $P=\{\}$, i.e.~none of the parties are dishonest, is basically depicted in Figure~\ref{fig:protexample} --- Definition~\ref{def:comp} requires the left and right resources in that figure to also be ``$\lambda$-close''.)}
\label{fig:simexample}
\end{figure}

We now present the full security definition we use, which formalizes the above discussions and clarifies the quantifier ordering (also see Figure~\ref{fig:simexample} for an example of one case in the definition):
\begin{definition}\label{def:comp}
For a scenario in which there is some set $Q$ of potentially dishonest parties, we say that \emph{$\mathscr{P}$ constructs $\left(\fideal_P\right)_{P\subseteq Q}$ from $\left(\freal_P\right)_{P\subseteq Q}$ within distance $\lambda$} iff the following holds: for every $P \subseteq Q$, there exists a converter $\simu^{P}$ which connects to their interfaces in $\fideal_P$, such that for every distinguisher, the distinguishing advantage between $\prot^{\overline{P}}\freal_P$ and $\fideal_P\simu^P$ is at most $\lambda$.
The converters $\simu^{P}$ shall be referred to as \emph{simulators}. 
\end{definition}

\begin{remark}
We have stated Definition~\ref{def:comp} slightly differently from~\cite{MR11}, in which an individual simulator is required for each dishonest party. 
A security proof satisfying Definition~\ref{def:comp} could be converted into one satisfying the~\cite{MR11} definition by modifying the choice of $\left(\fideal_P\right)_{P\subseteq Q}$ to one that (for every $P$ containing more than one party) explicitly includes quantum channels between the dishonest parties $P$, which would allow for individual simulators that communicate using these quantum channels in order to effectively implement the simulator $\simu^{P}$ in Definition~\ref{def:comp}. From the perspective of~\cite{MR11}, this would basically reflect the inability of a protocol to \emph{guarantee} that the dishonest parties cannot communicate with each other. For subsequent ease of describing the simulators, in this work we shall follow Definition~\ref{def:comp} as stated, instead of the exact definition in~\cite{MR11}.
\end{remark}

\subsubsection{Operational implications}
\label{sec:operational}

The operational implications\footnote{A more abstract composability notion given by this definition is that if several protocols satisfying this definition are composed, the ``error'' $\lambda$ of the resulting larger protocol can be bounded by simply by adding those of the sub-protocols; see~\cite{MR11} for further details.} of 
Definition~\ref{def:comp} can be seen by considering the composition of protocols satisfying this definition. Namely, suppose we have a protocol $\mathscr{P}$ that constructs $\left(\fideal_P\right)_{P\subseteq Q}$ from $\left(\freal_P\right)_{P\subseteq Q}$ within distance $\lambda$ in the sense of Definition~\ref{def:comp}.
Consider any larger protocol $\mathscr{P}'$ that is designed to use the ideal functionality $\left(\fideal_P\right)_{P\subseteq Q}$ as a resource. 
We want to study what happens if this intended functionality $\left(\fideal_P\right)_{P\subseteq Q}$ is replaced by the protocol $\mathscr{P}$ applied to $\left(\freal_P\right)_{P\subseteq Q}$.

To do so, let us pick any $P\subseteq Q$ and focus on the scenario where the parties $P$ are dishonest.
Now take any event in this scenario that could be considered a ``failure'' in the larger protocol (we impose no restrictions on the nature of a failure, except that it be a well-defined event), and suppose the larger protocol also comes with a proof that for any strategy by the dishonest parties, the probability of this failure event is upper-bounded by some $p_0$ when using the ideal resource $\fideal_P$. In that case, one implication of Definition~\ref{def:comp} being satisfied is that the probability of this failure event (in the larger protocol) is still at most $p_0+\lambda$ even when $\left(\fideal_P\right)_{P\subseteq Q}$ is replaced by the protocol $\mathscr{P}$ applied to $\left(\freal_P\right)_{P\subseteq Q}$.
This follows from the following observations. Firstly, since the bound $p_0$ for the functionality $\fideal_P$ holds for any strategy by the dishonest parties, it must in particular hold when they implement the simulator $\simu^P$, i.e.~it holds if they are using $\fideal_P\simu^P$ instead of $\fideal_P$. Secondly, since the distinguishing advantage between $\fideal_P\simu^P$ and $\prot^{\overline{P}}\freal_P$ is at most $\lambda$ according to Definition~\ref{def:comp}, the probability of the failure event cannot differ by more than $\lambda$ between them (otherwise the event would serve as a way to distinguish the two cases with advantage greater than $\lambda$). In other words, the replacement has not increased the maximum probability of the failure event by more than $\lambda$ --- we highlight again that this failure event could have been chosen to be any arbitrary event; hence this is a very powerful operational statement.

From this argument, we can also gain insight into the purpose of the apparently contrived simulators $\simu^P$ in Definition~\ref{def:comp}. Basically, they represent some operations that the dishonest parties \emph{could have performed} using only the ideal resource $\fideal_P$. While these operations are not necessarily the ``optimal'' ones they could do in order to ``attack'' $\fideal_P$, the fact remains that they are valid operations for the dishonest parties, and are hence subject to the bound $p_0$ on the failure probability. Critically, this was enough for the rest of the argument above to carry through, even though the simulator $\simu^P$ may be carrying out some complicated operations that have no correspondence to any sort of reasonable ``attack'' on $\fideal_P$. Such arguments are not uncommon in simulator-based notions of security (even outside the Abstract Cryptography framework), the shared underlying idea being typically that the real behaviour could have been reproduced (or approximated) by taking \emph{only} the ideal behaviour and applying some operations to it, hence it essentially suffices to only consider the ideal behaviour. 
(In the context of our specific protocol, we shall briefly discuss this further in Remark~\ref{remark:sim} of Sec.~\ref{sec:compBE}.)

There is a technicality in the above argument, namely that in order for the reasoning to be valid, the bound $p_0$ (for the larger protocol $\mathscr{P}'$ using $\fideal_P$) must be derived for a class of dishonest-party strategies that includes the simulator $\simu^P$, in order for the bound to hold for $\fideal_P\simu^P$ as well. 
This means that 
if a more ``powerful'' simulator is used in Definition~\ref{def:comp}, then the bound $p_0$ must be proved against a more ``powerful'' class of strategies. In particular, 
for instance the simulators $\simu^P$ we construct in this work assume that the dishonest parties $P$ collaborate to some extent (when there are multiple dishonest parties), which means that to apply the above operational interpretation, the bound $p_0$ for the larger protocol must be valid against collaborating dishonest parties. However, we note that this is more of a consideration for the larger protocol $\mathscr{P}'$, rather than the protocol $\mathscr{P}$ in Definition~\ref{def:comp} itself.

\subsection{Ideal ECD functionality}
We work in a setting with 
three parties: Alice who is always honest, and Bob and Eve who may independently be honest or dishonest. The inputs for Alice and honest Bob into the functionality are:
\begin{enumerate}[(i)]
\item Message $M\in\{0,1\}^n$ from Alice at time $t_2$
\item Deletion decision $D\in\{0,1\}$ from Alice at time $t_3$
\end{enumerate}
and their outputs are:
\begin{enumerate}[(i)]
\item Abort decision $O \in \{\top, \bot\}$ to Alice and Bob at time $t_1$
\item Deletion decision $D$ to Bob at time $\bigdot{t}_3$
\item Deletion flag $F \in \{\text{\cmark, \xmark}\}$ to Alice at time $t_4$
\item $\tM = \begin{cases} M & \text{ if } D\land F=0 \\ 0^n & \text{ if } D\land F=1\end{cases}$ \hspace{0.2cm} to Bob at time $t_5$
\end{enumerate}
where for the purposes of applying the AND function to the binary symbols $\{\text{\xmark, \cmark}\}$, \xmark is interpreted as 0.

The times corresponding to the inputs and outputs must satisfy $t_1 \leq t_2 \leq t_3 \leq \bigdot{t}_3 \leq t_4 < t_5$. In particular, we shall call a functionality an ideal ECD functionality if it produces the above inputs and outputs at any points in time satisfying the above constraints. We have strict inequality only between $t_4$ and $t_5$ because this is necessary in any real protocol for achieving the functionality.

We now describe how the honest inputs and outputs are to be interpreted. Alice and Bob's output $O$ is to detect interference by Eve. If $O=\bot$ then the protocol stops and no further inputs are fed in or outputs are received. Alice's input $M$ is self-explanatory: this is the secret message that she potentially wants Bob to learn. Alice's decision $D$ is her later decision about whether she wants Bob to learn $M$: some time after inputting the message but strictly before the time $t_5$ when the message is supposed to be revealed, Alice inputs $D=1$ if she does not want Bob to learn $M$; otherwise she inputs $D=0$. $D$ is directly output to Bob some time after Alice inputs it. The output $F$ to Alice is only produced if $D=1$, and this indicates whether Bob has produced a valid deletion certificate (although the deletion certificate itself is not part of the ideal functionality). If Bob is honest then he always produces a valid certificate if Alice asks him to, and $F$ is always \cmark. Finally, the output $\tM$ to Bob is a function of $M$, $D$ and $F$: if $D=0$, i.e., Alice wanted him to learn the message, or $F=$ \xmark, i.e., he did not produce a valid deletion certificate, then $\tM = M$; otherwise it is the dummy string $0^n$.

Now we come to the inputs and outputs of dishonest parties, which are the following:
\begin{enumerate}[(i)]
\item Abort decisions $O^\B$ and $O^\E \in \{\top, \bot\}$ from Bob and Eve at times $t_1', t_1''$ respectively
\item Deletion decision $D\in\{0,1\}$ to Eve at time $\bigddot{t}_3$
\item Deletion flag $F\in \{\text{\cmark, \xmark}\}$ from Bob at time $t_4'$.
\end{enumerate}
The times corresponding to these inputs and outputs must satisfy the following ordering with respect to the previously specified times: $t_1', t_1'' \leq t_1$, $t_3 \leq \bigddot{t}_3$ and $\bigdot{t}_3 \leq t_4' \leq t_4$. We remark that we are indifferent about the relative ordering of $t_1', t_1''$ and $\bigdot{t}_3, \bigddot{t}_3$.

Eve's input $O^\E$ is similar to what she has in the ideal key distribution functionality that is achieved by quantum key distribution (see e.g.~\cite{PR14}). She has the ability to interfere in a way that makes the honest parties abort the protocol. If Bob is honest then Eve's choice of $O^\E$ directly gets output to Alice and Bob as $O$ and the protocol stops if $O^\E=\bot$. However, if $O^\E = \top$, then the protocol continues and Eve gets nothing other than $D$ as further outputs, and in particular she is not able to learn the message. We include $D$ as an output for Eve because we cannot prevent her from learning this in our actual protocol. Dishonest Bob also has an input $O^\B$ that he can use to make the protocol abort: Alice and Bob's output $O$ is $\bot$ if \emph{either one of} Bob and Eve inputs $\bot$. We include this input for Bob because in the real protocol we cannot prevent Bob from deliberately sabotaging whatever test Alice and Bob are supposed to perform in order to detect interference from Eve, so that the output is $O=\bot$.\footnote{Certainly from the point of view of the ECD functionality such an action by Bob seems pointless, but we cannot exclude this possibility for any composable security proof, since we do not know Bob's motivations in some hypothetical larger protocol in which the real ECD protocol is used, and it may be useful for Bob there if the protocol outputs $\bot$.} Finally, Bob's input $F$ indicates whether he has decided to produce a valid deletion certificate and hence lose his ability to learn the message or not, and this is directly output to Alice. Honest Bob does not have this functionality, as he simply always deletes his information if requested by Alice.

The final $\text{ECD}_n$ functionality, parametrized by the message length $n$, is depicted in Figure \ref{fig:F_ideal} in the four possible combinations of honest and dishonest Bob and Eve. A security proof of a protocol for $\text{ECD}_n$ with security parameter $\lambda$ will consist of showing that the protocol constructs the functionality depicted in Figure \ref{fig:F_ideal} within distance $\lambda$ as per Definition~\ref{def:comp}.

\newpage
\begin{figure}[!h]
\centering
\begin{subfigure}[b]{0.48\textwidth}
\centering
\begin{tikzpicture}[scale=0.9]
\draw[rounded corners] (-1,-3.5) rectangle (0.6,3.5);
\draw[thick,<-] (-3.3,1.8) -- node[yshift=0.3cm] {\small $O=\top$} node[xshift=-1.3cm] {\color{dgreen} $t_1$} (-1,1.8);
\draw[thick,->] (0.6,1.8) -- node[yshift=0.3cm] {\small $O=\top$} node[xshift=1.3cm] {\color{dgreen} $t_1$} (2.9,1.8);
\draw[thick,->] (-3.3,0.8) -- node[yshift=0.3cm] {\small $M \in \{0,1\}^n$} node[xshift=-1.3cm] {\color{dgreen} $t_2$} (-1,0.8);
\draw[thick,->] (-3.3,-0.2) -- node[yshift=0.3cm] {\small $D \in \{0,1\}$} node[xshift=-1.3cm] {\color{dgreen} $t_3$} (-1,-0.2);
\draw[thick,->] (0.6,-0.2) -- node[yshift=0.3cm] {\small $D$} node[xshift=1.3cm] {\color{dgreen} $\bigdot{t}_3$} (2.9,-0.2);
\draw[thick,<-] (-3.3,-1.2) -- node[yshift=0.3cm] {\small $F=$ \cmark} node[xshift=-1.3cm] {\color{dgreen} $t_4$} (-1,-1.2);
\draw[thick,->] (0.6,-2.8) -- node[yshift=0.6cm] {\footnotesize $\tM = \begin{cases} M \text{ if } D=0 \\ 0^n \text{ if } D=1\end{cases}$} node[xshift=2cm] {\color{dgreen}$t_5$} (4.5,-2.8); 
\end{tikzpicture}
\caption{Alice, honest Bob and honest Eve}
\end{subfigure}
\hfill
\begin{subfigure}[b]{0.48\textwidth}
\centering
\begin{tikzpicture}[scale=0.9]
\draw[rounded corners] (-1,-3.5) rectangle (0.6,3.5);
\draw[thick,->] (-0.5,4.5) -- node[xshift=-1.2cm] {\small $O^\E \in \{\bot,\top\}$} node[yshift=0.8cm] {\color{dgreen} $t''_1$} (-0.5,3.5);
\draw[thick,<-] (0.1,4.5) -- node[xshift=0.3cm] {\small $D$} node[yshift=0.8cm] {\color{dgreen} $\bigddot{t}_3$} (0.1,3.5);
\draw[thick,<-] (-3.3,1.8) -- node[yshift=0.3cm] {\small $O=O^\E$} node[xshift=-1.3cm] {\color{dgreen} $t_1$} (-1,1.8);
\draw[thick,->] (0.6,1.8) -- node[yshift=0.3cm] {\small $O=O^\E$} node[xshift=1.3cm] {\color{dgreen} $t_1$} (2.9,1.8);
\draw[thick,->] (-3.3,0.8) -- node[yshift=0.3cm] {\small $M \in \{0,1\}^n$} node[xshift=-1.3cm] {\color{dgreen} $t_2$} (-1,0.8);
\draw[thick,->] (-3.3,-0.2) -- node[yshift=0.3cm] {\small $D \in \{0,1\}$} node[xshift=-1.3cm] {\color{dgreen} $t_3$} (-1,-0.2);
\draw[thick,->] (0.6,-0.2) -- node[yshift=0.3cm] {\small $D$} node[xshift=1.3cm] {\color{dgreen} $\bigdot{t}_3$} (2.9,-0.2);
\draw[thick,<-] (-3.3,-1.2) -- node[yshift=0.3cm] {\small $F=$ \cmark} node[xshift=-1.3cm] {\color{dgreen} $t_4$} (-1,-1.2);
\draw[thick,->] (0.6,-2.8) -- node[yshift=0.6cm] {\footnotesize $\tM = \begin{cases} M \text{ if } D=0 \\ 0^n \text{ if } D=1\end{cases}$} node[xshift=2cm] {\color{dgreen}$t_5$} (4.5,-2.8); 
\end{tikzpicture}
\caption{Alice, honest Bob and dishonest Eve}
\end{subfigure}
\vskip\baselineskip
\begin{subfigure}[b]{0.48\textwidth}
\centering
\begin{tikzpicture}[scale=0.9]
\draw[rounded corners] (-1,-3.5) rectangle (0.6,3.5);
\draw[thick,<-] (0.6,2.8) -- node[yshift=0.3cm] {\small $\; O^\B \in \{\bot,\top\}$} node[xshift=1.3cm] {\color{dgreen} $t_1'$} (2.9,2.8);
\draw[thick,<-] (-3.3,1.8) -- node[yshift=0.3cm] {\small $O=O^\B$} node[xshift=-1.3cm] {\color{dgreen} $t_1$} (-1,1.8);
\draw[thick,->] (0.6,1.8) -- node[yshift=0.3cm] {\small $O=O^\B$} node[xshift=1.3cm] {\color{dgreen} $t_1$} (2.9,1.8);
\draw[thick,->] (-3.3,0.8) -- node[yshift=0.3cm] {\small $M \in \{0,1\}^n$} node[xshift=-1.3cm] {\color{dgreen} $t_2$} (-1,0.8);
\draw[thick,->] (-3.3,-0.2) -- node[yshift=0.3cm] {\small $D \in \{0,1\}$} node[xshift=-1.3cm] {\color{dgreen} $t_3$} (-1,-0.2);
\draw[thick,->] (0.6,-0.2) -- node[yshift=0.3cm] {\small $D$} node[xshift=1.3cm] {\color{dgreen} $\bigdot{t}_3$} (2.9,-0.2);
\draw[thick,<-] (-3.3,-1.2) -- node[yshift=0.3cm] {\small $F$} node[xshift=-1.3cm] {\color{dgreen} $t_4$} (-1,-1.2);
\draw[thick,<-] (0.6,-1.2) -- node[yshift=0.3cm] {\small $F \in \{\text{\xmark, \cmark}\}$} node[xshift=1.3cm] {\color{dgreen} $t'_4$} (2.9,-1.2);
\draw[thick,->] (0.6,-2.8) -- node[yshift=0.6cm] {\footnotesize $\quad \tM = \begin{cases} M \text{ if } D\land F=0 \\ 0^n \text{ if } D\land F=1\end{cases}$} node[xshift=2cm] {\color{dgreen}$t_5$} (4.5,-2.8);
\end{tikzpicture}
\caption{Alice, dishonest Bob and honest Eve}
\end{subfigure}
\hfill
\begin{subfigure}[b]{0.485\textwidth}
\centering
\begin{tikzpicture}[scale=0.9]
\draw[rounded corners] (-1,-3.5) rectangle (0.6,3.5);
\draw[thick,->] (-0.5,4.5) -- node[xshift=-1.2cm] {\small $O^\E \in \{\bot,\top\}$} node[yshift=0.8cm] {\color{dgreen} $t''_1$} (-0.5,3.5);
\draw[thick,<-] (0.1,4.5) -- node[xshift=0.3cm] {\small $D$} node[yshift=0.8cm] {\color{dgreen} $\bigddot{t}_3$} (0.1,3.5);
\draw[thick,<-] (0.6,2.8) -- node[yshift=0.3cm] {\small $\; O^\B \in \{\bot,\top\}$} node[xshift=1.3cm] {\color{dgreen} $t_1'$} (2.9,2.8);
\draw[thick,<-] (-3.3,1.8) -- node[xshift=-0.1cm,yshift=0.3cm] {\small $O=O^\B\land O^\E$} node[xshift=-1.3cm] {\color{dgreen} $t_1$} (-1,1.8);
\draw[thick,->] (0.6,1.8) -- node[xshift=0.1cm,yshift=0.3cm] {\small $O=O^\B\land O^\E$} node[xshift=1.3cm] {\color{dgreen} $t_1$} (2.9,1.8);
\draw[thick,->] (-3.3,0.8) -- node[yshift=0.3cm] {\small $M \in \{0,1\}^n$} node[xshift=-1.3cm] {\color{dgreen} $t_2$} (-1,0.8);
\draw[thick,->] (-3.3,-0.2) -- node[yshift=0.3cm] {\small $D \in \{0,1\}$} node[xshift=-1.3cm] {\color{dgreen} $t_3$} (-1,-0.2);
\draw[thick,->] (0.6,-0.2) -- node[yshift=0.3cm] {\small $D$} node[xshift=1.3cm] {\color{dgreen} $\bigdot{t}_3$} (2.9,-0.2);
\draw[thick,<-] (-3.3,-1.2) -- node[yshift=0.3cm] {\small $F$} node[xshift=-1.3cm] {\color{dgreen} $t_4$} (-1,-1.2);
\draw[thick,<-] (0.6,-1.2) -- node[yshift=0.3cm] {\small $F \in \{\text{\xmark,\cmark}\}$} node[xshift=1.3cm] {\color{dgreen} $t'_4$} (2.9,-1.2);
\draw[thick,->] (0.6,-2.8) -- node[yshift=0.6cm] {\footnotesize $\quad \tM = \begin{cases} M \text{ if } D\land F=0 \\ 0^n \text{ if } D\land F=1\end{cases}$} node[xshift=2cm] {\color{dgreen}$t_5$} (4.5,-2.8);
\end{tikzpicture}
\caption{Alice, dishonest Bob and dishonest Eve}
\end{subfigure}
\caption{Ideal $\text{ECD}_n$ functionality in four cases. The times at which various events occur satisfy $t_1', t_1'' \leq t_1 \leq t_2 \leq t_3 \leq {\bigdot{t}}_3, \bigddot{t}_3 \leq t'_4 \leq t_4 < t_5$. All inputs and outputs after $O$ are only provided if $O=\top$; the $F$ input and output are only provided if $D=1$. $\bot$ and \xmark are interpreted as 0 for the AND function.}
\label{fig:F_ideal}
\end{figure}

\newpage
\subsection{Achievability result}\label{sec:achiev}
Before stating the result about our protocol constructing the ideal ECD functionality, we clarify what resources are used by the protocol, and what assumptions are needed on said resources.
\paragraph{Resources used.}
\begin{enumerate}[(i)]
\vspace{-0.2cm} \item Boxes $(\clB^1,\clB^2)$ of the form $(\clB^1_1\ldots \clB^1_l,\clB^2_1\ldots \clB^2_l)$, where each $(\clB^1_i, \clB^2_i)$ is compatible with one instance of the magic square game $\MS$. These boxes are ``untrusted'', in the sense that Eve\footnote{This describes the resource behaviour in the cases where Eve is dishonest. In the case where only Bob is dishonest, we shall consider Bob to be the party choosing the box behaviour.} supplies them to Alice and Bob after having chosen the states in the boxes and the measurements they perform.
\item A trusted temporarily private randomness source $\clR$, which if used by Alice by time $t_1$, makes public all the randomness it supplied at some time $t_4 < t_5' \leq t_5$.
\item A trusted private randomness source $\clR^P$ with Alice (which does not make public the randomness it supplies to Alice).
\item An authenticated classical channel $\clC$ between Alice and Bob, which faithfully transmits all classical messages sent between them, but also supplies copies of the messages to Eve; the channel only needs to be active until time $t_4$.
\end{enumerate}

\paragraph{Assumptions about quantum boxes.} We make the following standard assumptions about the boxes $(\clB^1, \clB^2)$ for device-independent settings:
\begin{enumerate}[(i)]
\item The boxes cannot access any registers storing private information from previous protocols, and this restriction continues to hold if the boxes are re-used in future protocols.
\item When held by an honest party, the boxes do not broadcast the inputs supplied to them and outputs obtained.
\item There is a tripartite tensor-product structure across the state in the $\clB^1$ boxes, the state in the $\clB^2$ boxes, and Eve's side-information.\footnote{With the technical exception of the case where both Bob and Eve are dishonest, since they may choose to collaborate in that case; however, this is not too important since even this case can in principle be mapped to a tensor-product structure with quantum communication between Bob and Eve.} 
\item Whenever Bob is dishonest, he can ``open'' his boxes and directly perform arbitrary operations or measurements on the quantum state they contained. (This matches the allowed dishonest behaviours in other DI protocols for two-party cryptographic tasks~\cite{ACK+14, SCA+11, AMPS16, KST22, FM17}.)
\end{enumerate}
The first assumption here is to address the memory attack of~\cite{BCK13}, as mentioned in the introduction (an alternative possibility would be to require that this is the first time the devices are used, and the devices are destroyed afterwards, though this is rather impractical). The second assumption is a basic one that is rather necessary to do any cryptography at all, and the third is a standard nonlocal-game scenario that could be enforced in principle by spatial separation between the parties. 

Additionally, however, we also impose the following assumption:
\begin{enumerate}[(i)]
\setcounter{enumi}{4}
\item Whenever Bob is honest (but the boxes are supplied by a potentially dishonest Eve), he is able to enforce that his $l$ boxes $\clB^2_1\ldots \clB^2_l$ satisfy no-signalling conditions between them --- this could, for instance, be supported by physical ``shielding'' measures between the boxes that prevent the input to each box from leaking to other boxes. This implies that for any subset $T \subseteq [l]$, there is a well-defined notion of Bob supplying inputs to only the boxes specified by $T$ and immediately receiving the corresponding outputs, with the rest of the inputs to be supplied later (such a procedure is required in our protocol). Furthermore, this no-signalling constraint also implies that the overall output distribution across all the rounds will be unaffected by which subset $T$ of inputs Bob supplies first.
\label{subset-outputs}
\end{enumerate}
This condition is rather stronger than the typical requirements in the device-independent setting (for non-IID situations, at least). However, it appears to be necessary to ensure security against dishonest Eve in our protocol, though it could be omitted if we instead aim for a more limited functionality similar to~\cite{BI19} that does not consider security against an eavesdropper, or potentially it could be worked around using some ideas developed in~\cite{KT22}. We defer further discussion of this point to Section \ref{sect:sec-disc} below.

When Eve and Bob are honest, we assume that the boxes $(\clB^1_1\ldots \clB^1_l,\clB^2_1\ldots \clB^2_l)$ play $\MS^l$ with an i.i.d. strategy, although they may do so $\eps/2$-noisily. That is, each box independently wins $\MS$ with probability $1-\eps/2$ instead of 1, for some $\eps>0$.

We shall use $(\clB^1_1\ldots \clB^1_l,\clB^2_1\ldots \clB^2_l)_\eps$ to refer to boxes with the above properties.

\begin{theorem}\label{thm:main-achieve}
Assuming Conjecture \ref{conj:DI}, there exists a universal constant $\eps_0 \in (0,1)$ such that for any $\eps \in(0,\eps_0]$, $\lambda \in (0,1]$ and $n \in \mathbb{N}$,
there is a protocol that constructs the $\mathrm{ECD}_n$ functionality depicted in Figure \ref{fig:F_ideal}, within distance $\lambda$, using only the resources $\clR$, $\clC$ and $(\clB^1_1\ldots \clB^1_l, \clB^2_1\ldots \clB^2_l)_\eps$ for some $l=l(\lambda,\eps,n)$.
\end{theorem}
We reiterate here from Section \ref{sec:results} that in the above theorem, Conjecture \ref{conj:DI} is not required to show intermediate security properties (see Lemmas~\ref{lem:com}--\ref{lem:ct-indist}), but only to show that these properties lead to composable security. Furthermore, even without that conjecture, our proof is sufficient to show that the protocol is composably secure in the standard device-dependent setting at least.

\begin{remark}
The resources we use put some constraints on the timings achievable in the $\mathrm{ECD}_n$ functionality. For example, if $\clR$ makes the randomness used by Alice public at time $t_5'$, then $t_4$ and $t_5$ must satisfy $t_4 < t'_5 \leq t_5$. Similarly, the delay between $t_3$ and $\bigdot{t}_3$, $t_4'$ and $t_4$ depend on the time taken to transmit information between Alice and Bob using $\clC$. 
\end{remark}

\subsection{Implications of security definitions and assumptions}\label{sect:sec-disc}

A number of points regarding our definition of the ideal functionality, the composable security framework, the resources used, and what they mean for security, need further clarification. We do this below.

\begin{itemize}
\item While the message $D$ from Alice indicates whether she would like Bob to delete his ciphertext, the final decision of whether Bob actually does so is entirely up to him. Because of this, one could consider a variant protocol and ideal functionality definition where this message $D$ is omitted. However, it is easy to see that our security proof also proves the security of this variant (one way to see this would be to note that it is basically equivalent to the distinguisher always selecting $D=1$, and since our security proof needs to cover any distinguisher strategy, in particular it also includes this possibility).
\item In line with the Abstract Cryptography framework, we have defined the ideal functionality above for all four combinations of honest/dishonest Bob and Eve. In particular, we remark that in the case where both Bob and Eve are dishonest, this framework does not explicitly specify whether they collaborate with each other --- the framework is \textit{a priori} neutral about this possibility. However, in our proof, the simulator we construct to satisfy Definition~\ref{def:comp} for this case does involve some collaboration between them. This does not affect the abstract definition, but we reiterate that it has implications for the operational interpretation as described in Section~\ref{sect:abs-crypt}. Specifically, to ``safely'' use our protocol in place of the ideal ECD functionality in a larger protocol, the bounds on failure probabilities in the larger protocol when Bob and Eve are dishonest must at least hold for the case where Bob and Eve collaborate enough to implement the simulator we constructed. 
\item Our protocol is clearly unable to guarantee the \emph{absence} of communication between Bob and Eve when both are dishonest. Hence while we did not explicitly specify a communication channel between them in the ideal functionality, we do implicitly allow for the possibility. (As previously mentioned, our security proof proceeds without explicitly including such a channel in the ideal functionality, because we have stated Definition~\ref{def:comp} using joint simulators for the dishonest parties.)
\item We imposed a condition that honest Bob's boxes satisfy no-signalling conditions between them. This is required for our security proof in the case where Bob is honest and Eve is dishonest, because it relies on the existence of a single distribution $\sfP_{AB|XY}$ that describes the Alice-Bob input-output behaviour regardless of the order in which Bob supplies his inputs (essentially because in our proof here we use a sampling lemma that requires the underlying distribution to be independent of the sampled subset).\footnote{Still, in~\cite{KT22} we recently proposed a protocol with a different approach for the random sampling, and it seems possible in principle that the approach there could be adapted to the ECD protocol in order to remove the no-signalling condition on honest Bob's boxes, though that would be beyond the scope of this work.} (An analogous condition is not required for Alice because she supplies all her inputs at once in our protocol.) However, if we were to instead restrict ourselves to a more restricted two-party functionality that omits Eve from the scenario considered here (i.e.~similar to~\cite{BI19}, which does not model a third-party eavesdropper), then this condition would be unnecessary, due to the discussion in the next point.
\item In contrast, whenever Bob is dishonest, we do not require an analogous condition --- this is consistent with the fact that we allow dishonest Bob to ``open'' his boxes and perform arbitrary operations on the states within. This is not a problem for our security proof against dishonest Bob, because the distribution $\sfP_{AB|XY}$ in that context is not quite a fixed input-output distribution for ``abstract'' boxes, but rather describes the conditional distribution of the strings $AB$ when dishonest Bob produces $B$ using some operation chosen depending on a string $Y$ which he learns over several steps (although Alice is indeed obtaining $A$ by just supplying $X$ as input to her boxes). In particular, in that proof it is acceptable for there to be a different distribution $\sfP_{AB|XY}$ for each possible order in which Bob learns the string $Y$. (From a physical standpoint, the order in which Bob learns $Y$ does impose some time-ordering constraints on the parts of $B$ that he produces before learning all of $Y$; however, it seems difficult to make use of such constraints in the proof. Hence we simply use the fact that for each possible order in which Bob learns $Y$, there is a resulting distribution $\sfP_{AB|XY}$ which falls within the larger set of correlations captured in the parallel-repetition framework.)
\item Ideally, we would like to have a protocol for ECD that does not require private randomness on Alice's side. This is because if Alice had access to randomness that remains private unless she chooses to communicate it, she could simply one-time pad her message with a random string to generate the ciphertext, then refuse to reveal the random string to Bob later if she decides she does not trust him. However, our current protocol does involve some private randomness, due to a technical point in our security proof: we prove security by means of an anchored parallel repetition theorem, where only part of Alice's input is revealed to Bob or Eve. This seems to be a limitation of parallel proof techniques that have been given so far --- parallel QKD security proofs also have this requirement \cite{JMS17, Vid17}, and for essentially this reason, there have been no parallel security proofs of device-independent randomness expansion. Still, the advantage of our protocol over the trivial protocol with private randomness we just described is that in the trivial protocol, Alice needs to communicate with Bob at the moment when she decides whether the message should be revealed or not, whereas this is not required in our protocol: Alice can simply rely on the temporarily private randomness source to reveal all the information that Bob needs for decryption, and the randomness that Alice got from the private randomness source is not needed for decryption. In other words, we can suppose Alice does not even have a communication channel with Bob after a certain point, in which case the trivial protocol does not work. (Overall, this issue is somewhat due to the formalization we chose here for ``leaking the decryption key'' in the composable framework --- in contrast, in the game-based framework of~\cite{BI19}, one simply declares what is to be counted as part of the decryption key.)
\end{itemize}

\subsection{Alternative security definition}
\label{sec:altdef}

In order to provide an easier comparison to previous works, we now describe an alternative set of security definitions that could be considered, in a framework that is more similar to that considered in~\cite{BI19}. 
However, as the setup we consider is an extended version of that considered in~\cite{BI19}, some changes as compared to the definitions in that work are necessary --- we elaborate on those changes after presenting these modified definitions.
(These security definitions are somewhat similar to some intermediate properties in our composable security analysis, but our purpose in stating them here is mainly to ease comparison with~\cite{BI19}. Some features of this definition are similar to~\cite{GMP22}, which also featured interactive encryption schemes.) As we are aiming for information-theoretic security, we present various conditions below in terms of trace distance (rather than distinguishability with respect to a computationally bounded adversary).

\begin{altdef}
A device-independent ECD protocol involves two parties Alice and Bob and a potential third-party ``eavesdropper'' Eve, out of which Alice always behaves honestly while Bob and Eve might behave dishonestly. It consists of an \emph{encryption} phase, followed by an optional \emph{certified deletion} phase, and finally a \emph{decryption} phase. These phases have the following structure, where we let $Q_B$ denote a register that Bob holds and updates arbitrarily over the course of each phase (there is no loss of generality in using a single register for this purpose, because we do not bound its dimension), and analogously we have a register $Q_E$ for Eve. (As for Alice, she is always honest in this scenario, hence in the following we only state the registers she is supposed to hold at the end of each phase.)
\begin{itemize}
\item In the encryption phase, Alice and Bob begin with boxes of the form described in Section~\ref{sec:achiev}, and perform an interactive procedure (over a public authenticated channel) that outputs a classical value $O \in \{\top, \bot\}$ (informally, indicating whether ``eavesdropping'' was detected). If $O=\bot$, Alice and Bob do not proceed further beyond this point. If $O=\top$, Alice then chooses a message $m$, and Alice and Bob perform some further interactive steps,
during which Alice generates a decryption key $R$ and an additional string $P$ for potential use in the certified deletion phase, while Bob generates some quantum ciphertext state (depending on the message $m$). At the end, each party has the following registers: Alice holds the classical registers $ORP$, Bob holds his register $Q_B$ containing the ciphertext state and possibly some other quantum side-information, and Eve holds some quantum side-information in $Q_E$. 
\item In the certified deletion phase (if it occurs), Alice computes some message using the string $P$ and sends it to Bob. Alice and Bob then perform an interactive procedure (in which Bob can use the register $Q_B$ retained from the encryption phase), at the end of which Alice uses the string $P$ to produce a classical value $F \in \{\text{\cmark, \xmark}\}$ (indicating whether she accepted the deletion certificate).
At the end, each party has the following registers: 
Alice holds the classical registers $ORF$, and Bob and Eve hold some quantum side-information in their respective registers $Q_B$ and $Q_E$.
\item In the decryption phase, Alice sends the decryption key $R$ to Bob using the authenticated public channel. Bob then performs some local operations using $R$ and his register $Q_B$, producing some classical value $\tM$ (intended to be a guess for the message $m$). At the end, each party has the following registers: Alice holds the classical registers $OF$, 
Bob's register $Q_B$ contains the classical value $\tM$ and possibly some other quantum side-information, and Eve holds some quantum side-information in $Q_E$.
\end{itemize}
We require that the ECD protocol satisfies notions of completeness, correctness, and secrecy, defined as follows (where the parameters $\lcom,\lEC,\lCI$ have values in $[0,1]$):
\begin{itemize}
\item The protocol is $\lcom$-complete if, whenever Bob and Eve are behaving honestly, we have
\begin{align}\label{eq:defncomBE}
\Pr[O=\bot] \leq \lcom,
\end{align}
and whenever Bob is honest and the certified deletion phase occurs, we have 
\begin{align}\label{eq:defncomB}
\Pr[O=\top \land F=\text{\xmark}]
\leq \lcom.
\end{align} 
\item The protocol is $\lEC$-correct if for any choice of message $m$, whenever Bob is behaving honestly, we have 
\begin{align}\label{eq:defncorr}
\Pr[O=\top \land \tM\neq m]
\leq \lEC.
\end{align} 
\item The protocol is $\lCI$-secret if for any choice of message $m$, all of the following conditions hold. Firstly, for any behaviour by Bob and Eve, at all times between the point where $O$ is produced and the point where Alice releases the decryption key $R$, we have (letting $\omega_{Q_B Q_E|\top}(m)$ denote the state produced at that point in time when the chosen message value is $m$, conditioned on $O=\top$)
\begin{align}\label{eq:CI-BEbeforedec}
\Pr[O=\top] \norm{\omega_{Q_B Q_E|\top}(m) - \omega_{Q_B Q_E|\top}(0^n)}_1 \leq 2 \lCI, 
\end{align} 
and furthermore, if the certified deletion phase occurs, then at all times after $F$ is produced 
(including after Alice releases the decryption key $R$), 
we have (letting $\omega_{Q_B Q_E|\top\text{\cmark}}(m)$ denote the state produced at that point in time when the chosen message value is $m$, conditioned on $O=\top$ and $F=\text{\cmark}$)
\begin{align}\label{eq:certdel-BE}
\Pr[O=\top \land F=\text{\cmark}] \norm{\omega_{Q_B Q_E|\top\text{\cmark}}(m) - \omega_{Q_B Q_E|\top\text{\cmark}}(0^n)}_1 \leq 2 \lCI.
\end{align} 
Secondly, whenever Bob is behaving honestly (while Eve may not be), at all times after $O$ is produced, we have (where $\omega_{Q_E|\top}(m)$ is defined analogously to above)
\begin{align}\label{eq:CI-E}
\Pr[O=\top] \norm{\omega_{Q_E|\top}(m) - \omega_{Q_E|\top}(0^n)}_1 \leq 2 \lCI. 
\end{align} 
\end{itemize}
\end{altdef}

As compared to the security definitions in~\cite{BI19}, the above definitions may appear somewhat elaborate. However, this is necessary in some form due to the additional features we intend to incorporate here as compared to that work (such as security against an eavesdropper Eve, as well as the DI security guarantee), which in particular have required us to include in our protocol some additional interaction rounds and a flag $O$ to test for eavesdropping, as compared to~\cite{BI19}. 
Since their protocol involves comparatively few rounds of interaction, they were able to present their definitions using an explicitly denoted individual channel for each round; however, such a presentation would quickly become unwieldy for our protocol. Hence in the above definition we have simply described the processes in terms of ``interactive procedures'', which are to be implicitly understood as some sequence of channels performed by the relevant parties, and we have introduced the registers $Q_B,Q_E$ to account for dishonest Bob and/or Eve arbitrarily updating their registers over that process. (A similar approach has been used in other recent work on encryption protocols with interaction; see for instance Definition~5.8 in~\cite{GMP22}.)
To aid understanding, we now qualitatively describe the intuition behind each of the requirements in the above definition, and compare them to~\cite{BI19}.

Firstly, the completeness conditions~\eqref{eq:defncomBE}--\eqref{eq:defncomB} simply reflect the notion that when various parties are behaving honestly, the corresponding flags should take the ``accept'' value with high probability. More specifically,~\eqref{eq:defncomBE} states that when there is no eavesdropping, the flag $O$ takes value $\top$ with high probability (this has no counterpart in~\cite{BI19}, which does not check for eavesdropping); while~\eqref{eq:defncomB} basically ensures that when Bob validly deletes his information, the flag $F$ takes value $\text{\cmark}$ with high probability (this corresponds to Eq.~(37) in~\cite{BI19}, though here we 
slightly modify the condition to account for the $O$ flag).

Next, the correctness condition~\eqref{eq:defncorr} is also straightforward: it simply requires that when Bob behaves honestly, then regardless of how Eve acts, the probability that Bob decrypts to the wrong message and the flag $O$ is set to the ``accept'' value is small. This is the same as the definition of correctness provided in Section~5.2 of~\cite{BI19}. 

Finally, the secrecy condition is the one which appears most different from~\cite{BI19}. Qualitatively,~\eqref{eq:CI-BEbeforedec} and~\eqref{eq:certdel-BE} are meant to correspond respectively to the notions of \emph{ciphertext indistinguishability} and \emph{certified deletion security} in~\cite{BI19}: each of these notions informally describes some idea of an adversary trying to distinguish between an encryption of the actual message versus an encryption of the all-zero string $0^n$.\footnote{More specifically, ciphertext indistinguishability covers the situation where the adversary does not yet have access to the decryption key, while certified deletion security covers the situation where the certified deletion procedure has taken place and the adversary is then supplied with the decryption key.} In~\cite{BI19}, those definitions could be presented using a small number of explicitly denoted channels applied by the adversary, due to the small number of communication rounds in their protocol. However, our protocol involves too many rounds for this to be a practical description. We hence instead impose conditions in terms of the trace distance between the states produced in the two cases, and require these conditions to hold over the relevant timeframes. Due to the operational interpretation of trace distance as distinguishing advantage, this 
corresponds to a bound on the probability of an adversary successfully distinguishing the two cases (up to some technicalities about conditioning on the $OF$ flag values),
hence having implications that are qualitatively similar to the~\cite{BI19} definitions of ciphertext indistinguishability and certified deletion security.  
Lastly,~\eqref{eq:CI-E} is meant to ensure an analogous notion of indistinguishability from Eve's perspective even if she tries to eavesdrop on the devices, which is an aspect that was not considered in~\cite{BI19} and hence has no counterpart in that work.

In the remainder of this work, we will mostly focus on proving security in the abstract cryptography framework; however, we shall also include discussion of this alternative definition at several points (see Remark~\ref{remark:altdef} and Section~\ref{sec:altdefproof}) to aid understanding and comparison to previous work.
For instance, when we present our full ECD protocol later (Protocol~\ref{prot:ECD}), the descriptions are phrased in terms of the abstract cryptography resources explained in Sections~\ref{sect:abs-crypt}--\ref{sec:achiev}, but we also include below the protocol description (in Remark~\ref{remark:altdef}) an explanation of how the registers in the protocol correspond to the registers described in this alternative definition. For now, we just highlight and briefly discuss a particularly significant point: in the alternative definition presented above, there is a decryption key $R$ that is revealed in the decryption phase. In the context of the abstract cryptography framework, this register $R$ is instead formalized as the randomness Alice obtains from the temporarily private randomness source described in Section~\ref{sec:achiev}, which is revealed to all parties at some fixed later time --- this is how we have chosen to formalize (in that framework) the notion that this ``decryption key'' value might be revealed to Bob and/or Eve after the certified deletion phase.

\section{Preliminaries}\label{sect:prelim}

\subsection{Probability theory}
We shall denote the probability distribution of a random variable $X$ on some set $\clX$ by $\sfP_X$. For any event $\clE$ on $\clX$, the distribution of $X$ conditioned on $\clE$ will be denoted by $\sfP_{X|\clE}$. For joint random variables $XY$ with distribution $\sfP_{XY}$, $\sfP_X$ is the marginal distribution of $X$ and $\sfP_{X|Y=y}$ is the conditional distribution of $X$ given $Y=y$; when it is clear from context which variable's value is being conditioned on, we shall often shorten the latter to $\sfP_{X|y}$. We shall use $\sfP_{XY}\sfP_{Z|X}$ to refer to the distribution
\[ (\sfP_{XY}\sfP_{Z|X})(x,y,z) = \sfP_{XY}(x,y)\cdot\sfP_{Z|X=x}(z).\]
Occasionally we shall use notation of the form $\sfP_{XY}\sfP_{Z|x^*}$. This denotes the distribution
\[ (\sfP_{XY}\sfP_{Z|x^*})(x,y,z) = \sfP_{XY}(x,y)\cdot\sfP_{Z|X=x^*}(z),\]
which potentially takes non-zero value when $x\neq x^*$. For two distributions $\sfP_X$ and $\sfP_{X'}$ on the same set $\clX$, the $\ell_1$ distance between them is defined as
\[ \Vert\sfP_X - \sfP_{X'}\Vert_1 = \sum_{x\in\clX}|\sfP_X(x) - \sfP_{X'}(x)|.\]

\begin{fact}\label{fc:l1-dec}
For joint distributions $\sfP_{XY}$ and $\sfP_{X'Y'}$ on the same sets,
\[ \Vert\sfP_X -  \sfP_{X'}\Vert_1 \leq \Vert\sfP_{XY} - \sfP_{X'Y'}\Vert_1.\]
\end{fact}
\begin{fact}\label{fc:l1-dist}
For two distributions $\sfP_X$ and $\sfP_{X'}$ on the same set and an event $\clE$ on the set,
\[ |\sfP_X(\clE) - \sfP_{X'}(\clE)| \leq \frac{1}{2}\Vert\sfP_X - \sfP_{X'}\Vert_1.\]
\end{fact}
\begin{fact}\label{fc:cond-prob}
Suppose probability distributions $\sfP_X, \sfP_{X'}$ satisfy $\Vert \sfP_X - \sfP_{X'}\Vert_1 \leq \eps$, and an event $\clE$ satisfies $\sfP_X(\clE) \geq \alpha$, where $\alpha > \eps$. Then,
\[ \Vert\sfP_{X|\clE} - \sfP_{X'|\clE}\Vert_1 \leq \frac{2\eps}{\alpha}.\]
\end{fact}

The following result is a consequence of the well-known Serfling bound.
\begin{fact}[\cite{TL17}]\label{fc:serfling}
Let $Z=Z_1\ldots Z_l$ be $l$ binary random variables with an arbitrary joint distribution, and let $T$ be a random subset of size $\gamma l$ for $0 \leq \gamma \leq 1$, picked uniformly among all such subsets of $[l]$ and independently of $Z$. Then,
\[ \Pr\left[\left(\sum_{i\in T}Z_i \geq (1-\eps)\gamma l\right) \land \left(\sum_{i\in[l]\setminus T}Z_i < (1-2\eps)(1-\gamma)l\right)\right] \leq 2^{-2\eps^2\gamma l}.\]
\end{fact}

\subsection{Quantum information}
The $\ell_1$ distance between two quantum states $\rho$ and $\sigma$ is given by
\[ \Vert\rho-\sigma\Vert_1 = \Tr\sqrt{(\rho-\sigma)^\dagger(\rho-\sigma)} = \Tr|\rho-\sigma|.\]
The fidelity between two quantum states is given by
\[ \sfF(\rho,\sigma) = \Vert\sqrt{\rho}\sqrt{\sigma}\Vert_1.\]
The Bures distance based on fidelity is given by
\[ \sfB(\rho,\sigma) = \sqrt{1-\sfF(\rho,\sigma)}.\]

$\ell_1$ distance, fidelity and Bures distance are related in the following way.
\begin{fact}[Fuchs-van de Graaf inequality]\label{fc:fvdg}
For any pair of quantum states $\rho$ and $\sigma$,
\[ 2(1-\sfF(\rho,\sigma)) \leq \Vert\rho-\sigma\Vert_1\leq 2\sqrt{1-\sfF(\rho,\sigma)^2}.\]
Consequently,
\[ 2\sfB(\rho,\sigma)^2 \leq \norm{\rho-\sigma}_1 \leq 2\sqrt{2}\cdot\sfB(\rho,\sigma).\]
For two pure states $\ket{\psi}$ and $\ket{\phi}$, we have
\[ \norm{\,\state{\psi} - \state{\phi}\,}_1 = \sqrt{1 - \sfF\left(\state{\psi},\state{\phi}\right)^2} = \sqrt{1-|\inner{\psi}{\phi}|^2}.\]
\end{fact}
\begin{fact}[Uhlmann's theorem]\label{fc:uhlmann}
Suppose $\rho$ and $\sigma$ are mixed states on register $X$ which are purified to $\ket{\rho}$ and $\ket{\sigma}$ on registers $XY$, then it holds that
\[ \sfF(\rho, \sigma) = \max_U|\matel{\rho}{\Id_X\otimes U}{\sigma}|\]
where the maximization is over unitaries acting only on register $Y$. Due to the Fuchs-van de Graaf inequality, this implies that there exists a unitary $U$ such that
\[ \norm{(\Id_X\otimes U)\state{\rho}(\Id_X\otimes U^\dagger) - \state{\sigma}}_1 \leq 2\sqrt{\norm{\rho-\sigma}_1}.\]
\end{fact}
\begin{fact}\label{fc:chan-l1}
For a quantum channel $\clE$ and states $\rho$ and $\sigma$,
\[ \Vert\clE(\rho) - \clE(\sigma)\Vert_1 \leq \Vert\rho-\sigma\Vert_1 \quad \quad \sfF(\clE(\rho),\clE(\sigma)) \geq \sfF(\rho,\sigma).\]
\end{fact}

The entropy of a quantum state $\rho$ on a register $Z$ is given by
\[ \sfH(\rho) = -\Tr(\rho\log \rho).\]
We shall also denote this by $\sfH(Z)_\rho$. For a state $\rho_{YZ}$ on registers $YZ$, the entropy of $Y$ conditioned on $Z$ is given by
\[ \sfH(Y|Z)_\rho = \sfH(YZ)_\rho - \sfH(Z)_\rho\]
where $\sfH(Z)_\rho$ is calculated w.r.t. the reduced state $\rho_Z$. The conditional min-entropy of $Y$ given $Z$ is defined as
\[ \sfH_\infty(Y|Z)_\rho = \inf\{\lambda: \exists \sigma_Z \text{ s.t. } \rho_{YZ} \preceq 2^{-\lambda}\Id_Y\otimes\sigma_Z\}.\]
The conditional Hartley entropy of $Y$ given $Z$ is defined as
\[ \sfH_0(Y|Z)_\rho = \log\left(\sup_{\sigma_Z}\Tr(\supp(\rho_{YZ})(\Id_Y\otimes\sigma_Z))\right)\]
where $\supp(\rho_{YZ})$ is the projector on to the support of $\rho_{YZ}$. For a classical distribution $\sfP_{YZ}$, this reduces to 
\[ \sfH_0(Y|Z)_{\sfP_{YZ}} = \log\left(\sup_z\left|\{y: \sfP_{YZ}(y,z) > 0\}\right|\right).\]
For $0\leq \delta \leq 2$, the $\delta$-smoothed versions of the above entropies are defined as
\[ \sfH^\delta_\infty(Y|Z)_\rho = \sup_{\rho':\Vert\rho-\rho'\Vert_1 \leq \delta}\sfH_\infty(Y|Z)_{\rho'} \quad \quad \sfH^\delta_0(Y|Z)_{\sfP_{YZ}} = \inf_{\rho':\norm{\rho'-\rho}_1 \leq \delta}\inf\sfH_0(Y|Z)_{\rho'}.\]
\begin{fact}\label{fc:Hinf-ch-rule}
For any state $\rho_{XYZ}$ and any $0\leq \delta\leq2$,
\[ \sfH^\delta_\infty(Y|XZ)_\rho \geq \sfH^\delta_\infty(Y|Z)_\rho - \log|\clX|.\]
\end{fact}

The relative entropy between two states $\rho$ and $\sigma$ of the same dimensions is given by
\[ \sfD(\rho\Vert \sigma) = \Tr(\rho\log\rho) - \Tr(\rho\log\sigma).\]
\begin{fact}[Pinsker's Inequality]\label{pinsker}
For any two states $\rho$ and $\sigma$,
\[ \Vert\rho-\sigma\Vert_1^2 \leq 2\ln 2\cdot\sfD(\rho\Vert\sigma) \quad \text{ and } \quad \sfB(\rho,\sigma)^2 \leq \ln 2\cdot\sfD(\rho\Vert\sigma).\]
\end{fact}
\noindent The mutual information between $Y$ and $Z$ with respect to a state $\rho$ on $YZ$ can be defined in the following equivalent ways:
\[  \sfI(Y:Z)_\rho = \sfD(\rho_{YZ}\Vert\rho_Y\otimes\rho_Z) = \sfH(Y)_\rho - \sfH(Y|Z)_\rho = \sfH(Z)_\rho - \sfH(Z|Y)_\rho.\]
The conditional mutual information between $Y$ and $Z$ conditioned on $X$ is defined as
\[ \sfI(Y:Z|X)_\rho = \sfH(Y|X)_\rho - \sfH(Y|XZ)_\rho = \sfH(Z|X)_\rho - \sfH(Z|XY)_\rho.\]
Mutual information can be seen to satisfy the chain rule
\[  \sfI(XY:Z)_\rho =  \sfI(X:Z)_\rho +  \sfI(Y:Z|X)_\rho.\]

A state of the form
\[ \rho_{XY} = \sum_x \sfP_X(x)\state{x}_X\otimes\rho_{Y|x}\]
is called a CQ (classical-quantum) state, with $X$ being the classical register and $Y$ being quantum. We shall use $X$ to refer to both the classical register and the classical random variable with the associated distribution. As in the classical case, here we are using $\rho_{Y|x}$ to denote the state of the register $Y$ conditioned on $X=x$, or in other words the state of the register $Y$ when a measurement is done on the $X$ register and the outcome is $x$. Hence $\rho_{XY|x} = \state{x}_X\otimes \rho_{Y|x}$. When the registers are clear from context we shall often write simply $\rho_x$.
\begin{fact}\label{fc:guess-prob}
For a CQ state $\rho_{XY}$ where $X$ is the classical register, $\sfH_\infty(X|Y)_\rho$ is equal to the negative logarithm of the maximum probability of guessing $X$ from the quantum system $\rho_{Y|x}$, i.e.,
\[ \sfH_\infty(X|Y)_\rho = 	-\log\left(\sup_{\{M_x\}_x}\sum_x\sfP_X(x)\Tr(M_x\rho_{Y|x})\right)\]
where the maximization is over the set of POVMs with elements indexed by $x$.
\end{fact}
For CQ states, the expression for relative entropy for $\rho_{XY}$ and $\sigma_{XY}$ given by
\[ \rho_{XY} = \sum_x\sfP_{X}(x)\state{x}_X\otimes\rho_{Y|x} \quad \quad \sigma_{XY} = \sum_x\sfP_{X'}(x)\state{x}_X\otimes\sigma_{Y|x},\]
reduces to
\begin{equation}\label{eq:CQ-S-ch}
\sfS(\rho_{XY}\Vert \sigma_{XY}) = \sfS(\sfP_X\Vert \sfP_{X'}) + \bbE_{\sfP_X} \sfS(\rho_{Y|x}\Vert \sigma_{Y|x}).
\end{equation}
\noindent Accordingly, the conditional mutual information between $Y$ and $Z$ conditioned on a classical register $X$, is simply
\begin{equation}\label{eq:CQ-I-ch}
\sfI(Y:Z|X) = \bbE_{\sfP_X}\sfI(Y:Z)_{\rho_x}.
\end{equation}

\subsection{2-universal hashing}
\begin{definition}
A family $\mfH$ of functions from $\clX$ to $\clZ$ is called a \emph{2-universal family of hash functions} iff
\[ \forall x\neq x' \quad \quad \Pr[h(x) = h(x')] \leq \frac{1}{|\clZ|}\]
where the probability is taken over the choice of $h$ uniformly over $\mfH$.
\end{definition}
2-universal hash function families always exist if $|\clX|$ and $|\clZ|$ are powers of 2, i.e., bit strings of some fixed length (see e.g.~\cite{CW79}). We shall denote a family of 2-universal hash functions from $\{0,1\}^s$ to $\{0,1\}^n$ by $\mfH(s,n)$.\footnote{Strictly speaking, in our applications of this concept we shall need to use bitstrings of variable length (up to some upper bound $l$) as inputs to the hash functions. This can be handled simply by noting that there are $2^{l+1}-1$ bitstrings of length less than or equal to $l$. Hence there is an injective mapping from such bitstrings to bitstrings of length $l+1$, and we can then apply 2-universal hash families designed for the latter (note that the mapping leaves all conditional entropies invariant because it is injective).}

2-universal hash functions are used for privacy amplification in cryptography. For privacy amplification against an adversary with quantum side information, the following lemma is used.
\begin{fact}[Leftover Hashing Lemma, \cite{Ren-th}]\label{lem:hashing}
The CQ state $\rho_{CKH\tE}$, where $C$ is an $n$-bit classical register, $K$ is an $s$-bit classical register, and $H$ is a classical register of dimension $|\mfH(s,n)|$, is defined as
\[ \rho_{CKH\tE} = \sum_{k\in\{0,1\}^s}\sum_{h\in\mfH(s,n)}\frac{1}{|\mfH(s,n)|}\sfP_K(k)\state{h(k),k,h}_{CKH}\otimes\rho_{\tE|k}.\]
Then for any $\eps \in [0,1)$,
\[ \left\Vert\rho_{CH\tE} - \frac{\Id_C}{2^n}\otimes\rho_{H\tE}\right\Vert_1 \leq 2^{-\frac{1}{2}(\sfH^\eps_\infty(K|\tE)-n)} + 2\eps.\]
\end{fact}

\section{The magic square game \& its parallel repetition}\label{sect:ms}
In a 2-player $k$-round game $G$, Alice and Bob share an entangled state at the beginning of the game. In the $j$-th round, they receive inputs $(x^j, y^j)$ and produce outputs $(a^j, b^j)$ respectively. They can do this by performing measurements that depend on the inputs and outputs of all rounds up to the $j$-th, on the post-measured state from the previous round. Each round has an associated predicate $\sfV^j$ which is a function of all inputs and outputs up to the $j$-th round. Alice and Bob win the game iff
\[\bigwedge_{j=1}^k\sfV^j(x^1\ldots x^j, y^1\ldots y^j, a^1\ldots a^j, b^1\ldots b^j)=1.\]

For a $k$-round game $G$, let $G^l$ denote the $l$-fold parallel repetition of it, and let $G^{t/l}$ denote the following game:
\begin{itemize}
\item For $j=1$ to $k$, in the $j$-th round, Alice and Bob receive $x^j_1,\ldots, x^j_l$ and $y^j_1,\ldots, y^j_l$ as inputs.
\item For $j=1$ to $k$, in the $j$-th round, Alice and Bob output $a^j_1,\ldots, a^j_l$ and $y^j_1\ldots, y^j_l$.
\item Alice and Bob win the game iff $(x^1_i\ldots x^k_i,y^1_i\ldots y^k_i, a^1_i\ldots a^k_i,b^1_i\ldots b^k_i)$ win $G$ for at least $t$ many $i$-s.
\end{itemize}
A parallel repetition threshold theorem gives an upper bound on the winning probability of $G^{t/l}$ which is exponentially small in $l$, for sufficiently high values of $t/l$.

In the magic square game, 
\begin{itemize}
\item Alice and Bob receive respective inputs $x \in \{0,1,2\}$ and $y \in \{0,1,2\}$ independently and uniformly at random.
\item Alice outputs $a \in \{0,1\}^3$ such that $a[0]\oplus a[1]\oplus a[2] = 0$ and Bob outputs $b \in \{0,1\}^3$ such that $b[0]\oplus b[1]\oplus b[2] = 1$.
\item Alice and Bob win the game if $a[y] = b[x]$.
\end{itemize}
The classical value of the magic square game is $\omega(\MS) = {8}/{9}$, whereas the quantum value is $\omega^*(\MS)=1$.

We introduce two variants of the magic square game: a 3-player 1-round version where Alice and Bob's distributions are product and are anchored w.r.t. the third player Eve's input, and a 2-player 2-round version where Alice and Bob's first round inputs are product and anchored w.r.t. Bob's second round input. We shall make use of parallel repetition threshold theorems for both these kinds of games.

\subsection{2-player 2-round $\MSB_\alpha$}
$\MSB_\alpha$ is defined as follows:
\begin{itemize}
\item In the first round, Alice and Bob receive inputs $x\in\{0,1,2\}$ and $y'\in\{0,1,2\}$ independently and uniformly at random.
\item Alice outputs $a \in \{0,1\}^3$ such that $a[0]\oplus a[1] \oplus a[2] = 0$, and Bob outputs $b' \in \{0,1\}^3$ such that $b'[0]\oplus b'[1]\oplus b'[2] = 1$.
\item 
In the second round, Bob gets input $z=\perp$ (indicating no input) with probability $\alpha$, and otherwise $z=(x,y)$ where $y$ is uniformly distributed on $\{0,1,2\} \setminus y'$. (Alice has no input.)
\item Bob outputs $c\in\{0,1\}$. (Alice has no output.)
\item Alice and Bob win the game if $a[y'] = b'[x]$, and either Bob gets input $\perp$ or $a[y] = c$.  
\end{itemize}
\begin{lemma}\label{lem:MSB-w}
There exists a constant $0<c^\B<1$ such that $\omega^*(\MSB_\alpha) = 1 - c^\B(1-\alpha)$.
\end{lemma}
In order to prove Lemma~\ref{lem:MSB-w}, we shall make use of a result due to \cite{FM17}. 
\begin{fact}[\cite{FM17}]
\label{fc:FM}
Suppose 
Alice and Bob have a state and measurements 
that can win the $\MS$ game with probability $1-\delta$. Consider any $x,y,y' \in \{0,1,2\}$ such that $y' \neq y$, and suppose Alice and Bob perform the aforementioned measurements for inputs $x,y'$, receiving outputs $a,b'$. Then if Bob is subsequently given $y$ and Alice's input $x$, the probability that he can guess $a[y]$ is at most $\frac{1}{2}+9\sqrt{\delta}$.
\end{fact}

\begin{proof}[Proof of Lemma~\ref{lem:MSB-w}]
Since $z=(x,y)$ with probability $1-\alpha>0$, it suffices to show that the probability of winning the game for $z=(x,y)$ is at most $1-c^\B$ for some $c^\B>0$. If $z=\perp$, the game being played is just the standard magic square game, which can be won with probability 1.

The probability of winning the game if $z=(x,y)$ can be written as $\Pr[(a[y'] = b'[x]) \land (a[y] = c) | z=(x,y)]$, which is upper bounded by
\begin{equation*}
\min \{\Pr[a[y'] = b'[x]| z=(x,y)], \Pr[a[y] = c | z=(x,y)]\}.
\end{equation*}
Denote $\Pr[a[y'] = b'[x]| z=(x,y)]=1-\delta$. Then by Fact~\ref{fc:FM} we have $\Pr[a[y] = c | z=(x,y)] \leq \frac{1}{2} + 9\sqrt{\delta}$, hence the above expression is upper bounded by the maximum value of $\min \{1-\delta, \frac{1}{2} + 9\sqrt{\delta}\}$ over all possible $\delta$. Since $\frac{1}{2} + 9\sqrt{\delta}$ is continuous in $\delta$ and has value less than $1$ at $\delta=0$, this maximum must be less than $1$. In fact the maximum is obtained at the intersection of $1-\delta$ and $\frac{1}{2} + 9\sqrt{\delta}$ for $\delta \in [0,1]$, where the value is $\sim0.997$.
\end{proof}

\begin{theorem}\label{thm:msb-parrep}
For $c^\B_\alpha = c^\B(1-\alpha)$ from Lemma~\ref{lem:MSB-w}, $\delta$ such that $t=(1-c^\B_\alpha + \eta)l \in ((1-c^\B_\alpha)l, l]$, there exists $d^\B>0$ such that 
\begin{align*}
\omega^*(\MSB_\alpha^{t/l}) & \leq 2^{-d^\B\eta^3\alpha^2l}.
\end{align*}
\end{theorem}
We prove more a general version of Theorem~\ref{thm:msb-parrep} in Section \ref{sect:parrep}.

\subsection{3-player $\MSE_\alpha$ game}
$\MSE_\alpha$ is defined as follows:
\begin{itemize}
\item Alice and Bob receive inputs $x\in\{0,1,2\}$ and $y\in\{0,1,2\}$ independently and uniformly at random.
\item Eve receives an input $z=\perp$ (indicating no input) with probability $\alpha$, and $z=(x,y)$ with probability $1-\alpha$.
\item Alice outputs $a \in \{0,1\}^3$ such that $a[0]\oplus a[1]\oplus a[2] = 0$, Bob outputs $b\in\{0,1\}^3$ such that $b[0]\oplus b[1]\oplus b[2] = 1$ and Eve outputs $c\in\{0,1\}$.
\item Alice, Bob and Eve win the game if $a[y]=b[x]$, and either Eve gets input $\perp$ or she outputs $c=a[y]=b[x]$.
\end{itemize}
\begin{fact}[\cite{KKM+11}, and modification described in \cite{Vid17}]\label{fc:MSE-w}
There exists a constant $0 < c^\E < 1$ such that $\omega^*(\MSE_\alpha) = 1 - c^\E(1-\alpha)$.
\end{fact}
\begin{theorem}\label{thm:mse-parrep}
For $c^\E_\alpha = c^\E(1-\alpha)$ from Fact \ref{fc:MSE-w}, $\delta$ such that $t=(1-c^\E_\alpha+\eta)l \in ((1-c^\E_\alpha)l,l]$, there exists a constant $d^\E>0$ such that
\[ \omega^*(\MSE_\alpha^{t/l}) \leq 2^{-d^\E\eta^3\alpha^2l}.\]
\end{theorem}
We prove a more general version of Theorem~\ref{thm:mse-parrep} in Section \ref{sect:parrep}. Note that it is possible to use the result of \cite{BVY15, BVY17} for $k$-player anchored games to get a version of Theorem~\ref{thm:mse-parrep} with worse parameters; we provide a different proof in order to improve the parameters.

\section{ECD protocol}\label{sect:protocol}
Our ECD protocol, which uses the resources $\clC, \clR, \clR^P$ and $(\clB^1_1\ldots \clB^1_l, \clB^2_1\ldots \clB^2_l)_\eps$, is given in Protocol~\ref{prot:ECD}. To be specific, Protocol~\ref{prot:ECD} describes the steps performed by Alice and honest Bob; we have highlighted the steps that a dishonest Bob need not perform in \textcolor{red}{red}.\footnote{We assume Bob sends $b_{\overline{T}}$ and $b'_{\overline{T}}$ generated in some arbitrary manner to Alice in lines \ref{alg:ini-state} and \ref{alg:b'-send} even if he is dishonest, so as not to give himself away as dishonest, which is why these have not been highlighted in \textcolor{red}{red}.} We have also indicated in \textcolor{dgreen}{green} steps that occur at specific times corresponding to the ideal functionality. The parameters $l, \alpha, \gamma$ in the protocol need to satisfy conditions specified in Section \ref{sec:parameters}. The function $\syn$ used in the protocol is specified by Fact~\ref{fc:info-recon} later (it is basically the syndrome of an appropriate error-correcting code).

\begin{remark}\label{remark:agency}
In the Protocol~\ref{prot:ECD} description, for readability we present the various steps as being performed by the parties Alice and Bob themselves; for instance we say that (amongst various other steps) Alice performs various checks and outputs values 
$O$ and $F$. However, strictly speaking this does not perfectly match the formal descriptions required in the Abstract Cryptography framework, where many of those steps should instead be stated (in the honest case) as being performed by the \emph{protocols} $\prot^{\A}$ and $\prot^{\B}$ (recall that these are converters with inner interfaces attached to the real resources and outer interfaces interacting with the parties Alice and Bob). In that framework, technically we should instead say for instance that it is the \emph{protocol} $\prot^{\A}$ that performs those checks and outputs the values $O$ and $F$ to Alice (who is viewed as a sort of ``external agent'' who only plays the roles of supplying the values $M,D$ and receiving the outputs $O,F$). However, in the Protocol~\ref{prot:ECD} description and the subsequent analysis in this section, it is awkward to constantly describe the various actions as being performed by abstract protocols/converters, and hence for brevity we gloss over the distinction in this section and just describe all these actions as being performed by Alice and Bob themselves. (Though in Section~\ref{sect:comp-sec} where we prove security in the full Abstract Cryptography picture, we will return to being more explicit about this distinction.)
\end{remark}

\begin{algorithm}
\begin{algorithmic}[1]
\caption{ECD protocol}
\label{prot:ECD}
\vspace{0.3cm}
\Algphase{Phase 1: Encryption}
\State Alice and Bob receive $\clB^1_1\ldots \clB^1_l$ and $\clB^2_1\ldots \clB^2_l$ respectively from Eve\;
\State Alice gets $S \subseteq [l]$ obtained by choosing each index independently with probability $(1-\alpha)$, and $T \subseteq S$ (or $\subseteq [l]$ if $|S| < \gamma l$) of size $\gamma l$ uniformly at random, from $\clR$ \alglinelabel{alg:rand-subsets}
\State Alice gets $x_S$, $y_S$ and $y'_{\overline{T}}$ uniformly at random such that $y_i \neq y'_i$ for each $i$, from $\clR$ \alglinelabel{alg:rand-gen}
\State Alice supplies $x_S$ as inputs to her boxes corresponding to $S$ and uniformly random inputs from $\clR^P$ to the rest of her boxes, recording the outputs from the subset $S$ as a string $a_S$ 
\State Alice sends $(T, y_T)$ to Bob using $\clC$ \alglinelabel{alg:test-send}
\State \textcolor{red}{Bob inputs $y_T$ into his boxes corresponding to $T$ and gets output $b_T$}
\State Bob sends $b_T$ to Alice using $\clC$  \alglinelabel{alg:ini-state}
\State Alice tests if $|S| \geq \gamma l$ and $a_i[y_i] = b_i[x_i]$ for at least $(1-\eps)|T|$ many $i$-s in $T$ \alglinelabel{alg:A-test}
\If{the test passes} \alglinelabel{alg:setOvalue}
\State Alice sends $O=\top$ to Bob  \alglinelabel{alg:E-test}
\State \textcolor{dgreen}{At time $t_1$,} Alice \textcolor{red}{and Bob} output $O=\top$ 
\State Alice sets $K^\A=(a_i[y_i])_{i\in S}$ 
\State Alice gets $h \in \mfH(l+1,n)$, $U_1\in \{0,1\}^n$, $U_2\in \{0,1\}^{|\syn(K^\A)|}$ uniformly at random from $\clR$ \alglinelabel{alg:hashotp-gen}
\State \textcolor{dgreen}{At time $t_2$,} Alice selects input $M \in\{0,1\}^n$ \alglinelabel{alg:message}
\State Alice sends $C = (C_1, C_2) = (M\oplus h(K^\A)\oplus U_1, \syn(K^\A) \oplus U_2)$ to Bob using $\clC$ \alglinelabel{alg:C-send}
\Else
\State Alice sends $O=\bot$ to Bob
\State \textcolor{dgreen}{At time $t_1$,} Alice \textcolor{red}{and Bob} output $O=\bot$ and the protocol ends
\EndIf

\Algphase{Phase 2: Certified deletion}
\State \textcolor{dgreen}{At time $t_3$,} Alice selects input $D\in\{0,1\}$ 
\If{$D=0$}
\State Alice sends $0$ to Bob using $\clC$ \alglinelabel{alg:D=0}
\Else
\State Alice sends $(1,y'_{\overline{T}})$ to Bob using $\clC$ \alglinelabel{alg:y'-send}
\State \textcolor{red}{Bob inputs $y'_{\overline{T}}$ into his boxes corresponding to $\overline{T}$ and gets output $b'_{\overline{T}}$ \alglinelabel{alg:Bob-y'-meas}}
\State Bob sends $b'_{\overline{T}}$ to Alice using $\clC$ \alglinelabel{alg:b'-send}
\State Alice tests if $a_i[y'_i] = b'_i[x_i]$ for at least $(1-2\eps)|S\setminus T|$ many $i$-s in $S\setminus T$ \alglinelabel{alg:B-test}
\If{the test passes} \alglinelabel{alg:setFvalue}
\State \textcolor{dgreen}{At time $t_4$,} Alice outputs $F=\text{\cmark}$
\Else
\State \textcolor{dgreen}{At time $t_4$,} Alice outputs $F=\text{\xmark}$
\EndIf
\EndIf

\Algphase{Phase 3: Decryption}
\State $\clR$ reveals $R=(x_S, y_S, y'_{\overline{T}}, S, T, h,u_1,u_2)$ \alglinelabel{alg:R-rev}
\If{$D=0$}
\State \textcolor{red}{Bob inputs $y_{S\cap\overline{T}}$ to his boxes corresponding to $S\cap\overline{T}$, uniformly random inputs to the boxes corresponding to $\overline{S}$, and records the outputs from $S\cap \overline{T}$ as $b_{S\cap\overline{T}}$ \alglinelabel{alg:Bob-y-meas}}
\State \textcolor{red}{Bob sets $K^\B= (b_i[x_i])_{i\in S}$}
\State \textcolor{red}{Bob uses $K^\B$ and $\syn(K^\A)=C_2\oplus u_2$ to compute a guess $\tK^\A$ for $K^\A$ \alglinelabel{alg:err-corr} }
\State \textcolor{dgreen}{At time $t_5$,}  \textcolor{red}{Bob outputs $\tM = C_1\oplus h(\tK^\A)\oplus u_1$ \alglinelabel{alg:outputM} }
\Else
\State \textcolor{dgreen}{At time $t_5$,} \textcolor{red}{Bob outputs $\tM = 0^n$}
\EndIf
\end{algorithmic}
\end{algorithm}
\newpage 

\begin{remark}
Protocol \ref{prot:ECD} can be modified so that instead of step \ref{alg:y'-send}, Alice sends $y'_{\overline{T}}$ to Bob in either step \ref{alg:test-send} or \ref{alg:C-send}. It may be desirable to do this if we want to achieve the alternative ECD task discussed in Section \ref{sect:sec-disc}, where Bob makes the deletion decision himself, so as not to include an unnecessary round of interaction (since that version will not include the communication in steps \ref{alg:D=0} and \ref{alg:y'-send} at all). Because of the OTPs with $u_1$ and $u_2$, as we discuss later, sending $y'_{\overline{T}}$ earlier has no effect on security.
\end{remark}

\begin{remark}\label{remark:altdef}
We explain here how the above Protocol~\ref{prot:ECD} description should be viewed in the context of the alternative security definitions in Section~\ref{sec:altdef}. 

In the encryption phase of that definition, Alice is supposed to produce and store registers $O$, $R$ and $P$. With respect to the Protocol~\ref{prot:ECD} description, these should be viewed as follows: $O$ is simply the value $O$ generated in step~\ref{alg:setOvalue}, $R$ consists of the values $(x_S, y_S, y'_{\overline{T}}, S, T, h,u_1,u_2)$ obtained from the temporarily private randomness source $\clR$ over the course of the encryption phase (which are later revealed in step~\ref{alg:R-rev} for decryption, as we would expect), and $P$ consists of the values $(y'_{\overline{T}},x_S,y_S,a_S)$ (for potential later use in the certified deletion phase). 
As for Bob, he is supposed to end up with a ciphertext state in his register $Q_B$: with respect to the Protocol~\ref{prot:ECD} description, this ciphertext state consists of the register $C$ in step~\ref{alg:C-send} together with the quantum state in his half of the boxes. (If Bob is honest, he holds no other registers in $Q_B$; if he is dishonest then $Q_B$ holds a state which he has been updating arbitrarily over the course of the encryption phase using the messages exchanged.) Eve's side-information register $Q_E$ consists of her original quantum side-information on the boxes, as well as all messages exchanged between Alice and Bob.

In the certified deletion phase of that definition, the register $P$ used by Alice and Bob is as described above, and the register $F$ produced by Alice is simply the flag $F$ generated in step~\ref{alg:setFvalue} of the Protocol~\ref{prot:ECD} description. 

Finally, in the decryption phase of that definition, the decryption key $R$ is again as described above, and similarly for the ciphertext state that Bob uses to decrypt the message. His decrypted value $\tM$ for the message corresponds to step~\ref{alg:outputM} of the Protocol~\ref{prot:ECD} description.
\end{remark}

\subsection{Notation}\label{sec:prot-notn}
We shall introduce some notation that will be used in the rest of the section and the composable security proof. Firstly, note that even though for ease of presentation in the protocol, we have indicated Alice getting $R$ step by step from $\clR$, in reality she could have gotten it all in step \ref{alg:rand-gen} and here we shall consider her having done so. We do not use any registers for the randomness Alice got from $\clR^P$ as this is not relevant in the protocol or the security proof.

Consider the following state shared by Alice, Bob and Eve after step \ref{alg:A-test} of the protocol, when Alice has produced the abort decision $O$ but has not sent it to Bob yet:
\begin{align*}
\vph_{CFORK^\A \tA\tB B_T\tE} & = \state{0^n\text{\xmark}}_{CF}\otimes \sum_{ork^\A}\sfP_{ORK^\A}(ork^\A)\state{ork^\A}_{ORK^\A}\otimes\vph_{\tA\tB B_T\tE|ork^\A}.
\end{align*}
Here the ciphertext register $C=C_1C_2$ and the flag register $F$ --- which are initialized to default values --- are with Alice, as is the randomness $R$ received from $\clR$. The answer $B_T$ Alice got from Bob is with both Alice and Eve, but for the sake of brevity we only explicitly specify the copy with Eve. $\tA, \tB, \tE$ are the quantum registers held by Alice, Bob and Eve. We shall assume $\tB\tE$ includes $TY_T$ that Bob (and Eve) got from Alice, and $\tA\tB$ includes Alice and Bob's copies of $B_T$. Finally, we shall assume $\tB$ contains the register $K^\B$ on which Bob would obtain his raw key, if he were honest. Further states in the protocol are obtained from $\vph$ by passing some registers from Alice to Bob (and Eve) and local operations on the registers possessed by Alice or jointly Bob and Eve.

At times $t_2$ and $t_3$ the message $M=m$ and the deletion decision $D=0/1$ enter the protocol, and we shall specify these parameters when talking about states from these points on --- although the message dependence is only on the $C$ register, so we may drop the $M$ dependence when talking about other registers. We use the following notation to denote states at various times in the protocol conditioned on various events (all the states are conditioned on outputting $\top$ at time $t_1$, though we only mention this in the first one, since the protocol only continues after $t_1$ under this condition):

\nopagebreak
\begin{tabular}{lcp{10.5cm}}
$\rho_{CFORK^\A\tA\tB B_T\tE}$ & : & $\vph_{CFORK^\A\tA\tB B_T\tE}$ conditioned on $O=\top$ \\
$\hBrho_{CFORK^\A\tA\tB B_T\tE}(m,0)$ & : & state after honest Bob's measurement in step \ref{alg:Bob-y-meas} \\
$\Frho_{CFORK^\A \tA\tB B_T\tE}(m,1)$ & : & state at time $t_4$ when Alice has produced the flag $F$ \\
$\sigma_{CFORK^\A \tA\tB B_T\tE}(m,1)$ & : & $\Frho_{CFORK^\A \tA\tB B_T\tE}(m,1)$ conditioned on $F=$ \cmark
\vspace{0.2cm}
\end{tabular}

We shall use $p_\top$ to denote the probability of outputting $\top$ at time $t_1$, which is clearly the probability of $\rho$ within $\vph$.\footnote{By this we mean that $\vph = p_\top\rho + (1-p_\top)\rho'$, where $\rho'$ is the state conditioned on $O=\bot$ instead.} Let $p_{\text{\cmark}|\top}$ denote the probability of Alice outputting \cmark at time $t_4$ conditioned on outputting $\top$ at time $t_1$, for message $M=m$ and $D=1$, i.e., the probability of $\sigma(m,1)$ within $\Frho(m,1)$. This probability is independent of $m$, as we shall argue in Lemma~\ref{lem:Hmin-sigma}.

\subsection{Completeness and correctness}\label{sect:comcorr}

We now prove some lemmas which, in the context of the alternative security definitions in Section~\ref{sec:altdef}, will imply the completeness and correctness properties (as we shall discuss in Section~\ref{sec:altdefproof}). In the context of our composable security analysis (in Section~\ref{sect:comp-sec}), however, these will just be intermediate results in the full composable security proof.

\begin{lemma}\label{lem:honest-prob}
Suppose $\alpha,\gamma < \frac{1}{2}$ and $l \geq \frac{4}{(1-2\gamma)^2}$. If Bob and Eve are honest (so Alice and Bob's boxes are $\eps/2$-noisy, i.e., able to win each instance of $\MS$ with probability $1-\eps/2$) then
\[ p_\top \geq \left(1 -2^{-(1-2\gamma)^2l/8}\right)\left(1-2^{-\eps^2\gamma l/8}\right) \geq 1 -2^{-(1-2\gamma)^2l/8} -2^{-\eps^2\gamma l/8}.\]
Moreover, if Bob is honest and $D=1$ (i.e.~Alice requests a deletion certificate), then 
regardless of Eve, we have
\[ p_{\top}(1-p_{\text{\cmark}|\top}) \leq 2^{-2\eps^2\gamma l}.\]
\end{lemma}
\begin{proof}
Since each element of $[l]$ is included in $S$ independently with probability $(1-\alpha)$, we see from the Chernoff bound that the probability of outputting $\bot$ due to $|S| < \gamma l$ is bounded by
\[ \Pr[|S| < \gamma l] \leq 2^{-\frac{(1-\alpha-\gamma)^2}{2(1-\alpha)}l} \leq 2^{-\frac{(1-2\gamma)^2}{8}l},\]
for the choice of $\alpha$. Moreover, for our choice of $l$ the above quantity is at most $\frac{1}{2}$. Conditioned on $|S| \geq \gamma l$, the behaviour of the boxes on $T$ is independent of $S$. For any $i\in [l]$, let $W_i$ denote the indicator variable for the event $a_i[y_i] = b_i[x_i]$. Since each instance of $\MS$ is won with probability at least $1-\eps/2$ by honest boxes, the probability of aborting due to $a_i[y_i]\neq b_i[y_i]$ in at least $\eps|T|$ boxes in $T$, i.e., $\sum_{i\in T}W_i \leq (1-\eps)|T|$, is
\[ \Pr\left[\sum_{i\in T}W_i < (1-\eps)|T|\middle||S| \geq \gamma l\right] \leq 2^{-\frac{\eps^2|T|}{8}} = 2^{-\frac{\eps^2\gamma l}{8}}.\]
Hence overall,
\[ p_\top \geq \left(1 -2^{-(1-2\gamma)^2l/8}\right)\left(1-2^{-\eps^2\gamma l/8}\right).\]

$p_{\text{\cmark}|\top}$ is independent of $m$ in general (see Lemma~\ref{lem:Hmin-sigma} below), but it is easy to see why this is so for honest Bob: his behaviour in Phase 2 is entirely independent of $m$. 
To lower bound $p_{\text{\cmark}|\top}$, let $W_i$ be defined as before, and let $W'_i$ be the indicator variable for the event that when Bob inputs $y'_i$ into his box and gets output $b'_i$, they satisfy $a_i[y'_i] = b'_i[x_i]$. As the marginal distributions of $y_i$ and $y'_i$ are exactly the same, $W'_i$ and $W_i$ are identically distributed.

Recall that $W_i$ is the same variable regardless of when the inputs are provided and the outputs obtained, so we can consider doing the $y'_{\overline{T}}$ measurement on $\vph$. $p_\top$ is the probability of the event
\begin{equation}\label{eq:event_ptop}
(|S| \geq \gamma l) \land\left(\sum_{i\in T}W_i \geq (1-\eps)|T|\right),
\end{equation}
when all the measurements are done on $\vph$. Let $p_{\text{\cmark}}$ denote the probability of
\begin{equation}\label{eq:event_cmark}
(|S| \geq \gamma l)\land\left(\sum_{i\in T}W_i \geq (1-\eps)|T|\right)\land\left(\sum_{i\in S\setminus T}W'_i\geq (1-2\eps)|S\setminus T|\right).
\end{equation}
Since the distribution of $S$ is independent of the $W_i$-s and $W'_i$-s, from Lemma~\ref{fc:serfling},
\begin{align*}
& \Pr\left[\left(\sum_{i\in T}W_i \geq (1-\eps)|T|\right) \land \left(\sum_{i\in S\setminus T}W'_i < (1-2\eps)|S\setminus T|\right)\middle||S| \geq \gamma l\right]  \\
& = \Pr\left[\left(\sum_{i\in T}W_i \geq (1-\eps)|T|\right) \land \left(\sum_{i\in S\setminus T}W_i < (1-2\eps)|S\setminus T|\right)\middle||S| \geq \gamma l\right] \\
& \leq 2^{-2\eps^2\frac{\gamma l}{|S|}\cdot|S|} = 2^{-2\eps^2\gamma l}.
\end{align*}
This gives us 
\begin{align*}
p_{\text{\cmark}} &= p_\top - \Pr\left[(|S| \geq \gamma l)\land\left(\sum_{i\in T}W_i \geq (1-\eps)|T|\right)\land\left(\sum_{i\in S\setminus T}W'_i < (1-2\eps)|S\setminus T|\right)\right]\\
&\geq p_\top - 2^{-2\eps^2\gamma l} \Pr\left[|S| \geq \gamma l\right]\\
&\geq p_\top - 2^{-2\eps^2\gamma l} 
.
\end{align*}

Now simply observe that $p_{\text{\cmark}} = p_\top p_{\text{\cmark}|\top}$, which gives us the required result. Note that here we required that upon receiving $y'_{\overline{T}}$, Bob produces $b'_{\overline{T}}$ by the same procedure by which he produced $b_T$ upon receiving $y_T$, which is his honest behaviour (even though the boxes themselves are untrusted), but we did not assume that the procedure actually implements anything close to the ideal MS measurements. A dishonest Bob, on the other hand, may produce $b'_{\overline{T}}$ by some different procedure, and hence this bound does not apply to him.
\end{proof}

Further analysis will be done assuming $\alpha, \gamma, l$ satisfy the conditions of Lemma~\ref{lem:honest-prob}, though we shall not state it explicitly in each case.

The correctness of Protocol \ref{prot:ECD}, i.e., the fact that Bob is able to produce the correct message if $D=0$ and he is honest, uses the following fact.
\begin{fact}[\cite{Ren-th}, Lemma 6.3.4]\label{fc:info-recon}
Suppose Alice and Bob respectively hold random variables $K^\A, K^\B \in \{0,1\}^s$. Then for $0 < \delta \leq 1$, there exists a protocol in which Alice communicates a single message $\syn(K^\A)$ of at most $\sfH^\delta_0(K^\A|K^\B) + \log(1/\lEC)$ bits to Bob, after which Bob can produce a guess $\widetilde{K}^\A$ that is equal to $K^\A$ with probability at least $1-(\delta+\lEC)/2$.  
\end{fact}

\begin{lemma}\label{lem:Bob-guess}
There is a choice of $C_2 = \syn(K^\A)$ of length $h_2(2\eps)l + \log(1/\lEC)$ bits, such that $\tK^\A$ produced by honest Bob in step \ref{alg:err-corr} of Phase 3 of the protocol is equal to $K^\A$ with probability at least $1-(2\cdot2^{-2\eps^2\gamma l}/p_\top + \lEC/2)$, where $h_2$ is the binary entropy function.
\end{lemma}
\begin{proof}
Define the random variables $W_i$ as in the proof of Lemma~\ref{lem:honest-prob}.
First observe that $\hBrho$ is the state conditioned on the event~\eqref{eq:event_ptop} after the measurements are done on $\vph$.\footnote{Note that we previously defined $\hBrho$ to be the state produced by the following sequence of steps: perform the measurements on $T$, then condition on a particular event, then perform the measurements on $\overline{T}$. However, this is where we use the condition imposed on honest Bob's boxes that the order of measurements does not change the distribution: because of that condition, we can equivalently consider $\hBrho$ to be the state where all the measurements are performed before the event is conditioned on.} 
Let $\thBrho$ be the state conditioned on the event
\begin{equation}\label{eq:event_thBrho}
\widetilde{\clE} = (|S| \geq \gamma l)\land\left(\sum_{i\in T}W_i \geq (1-\eps)|T|)\right)\land\left(\sum_{i\in S\setminus T}W_i\geq (1-2\eps)|S\setminus T|\right),
\end{equation}
when all the measurements are done on $\vph$.
Note that $\thBrho$ is also equal to $\hBrho$ conditioned on the event $\widetilde{\clE}$ (since the event $\widetilde{\clE}$ is a stricter condition than the event~\eqref{eq:event_ptop}).
We shall now argue that $\hBrho$ and $\thBrho$ are close in trace distance, by showing that the event $\widetilde{\clE}$ occurs with some sufficiently high probability in the state $\hBrho$. 

To do so, we follow exactly the same argument structure as in the proof of Lemma~\ref{lem:honest-prob}, except with the event $\widetilde{\clE}$ in place of the event~\eqref{eq:event_cmark}.
With respect to the state immediately after the measurements are done on $\vph$, let $p_{\widetilde{\clE}}$ denote the probability of the event $\widetilde{\clE}$, and let $p_{\widetilde{\clE}|\top}$ denote the probability of the event $\widetilde{\clE}$ conditioned on the event~\eqref{eq:event_ptop}. (Note that $p_{\widetilde{\clE}|\top}$ is also precisely the probability of the event $\widetilde{\clE}$ with respect to the state $\hBrho$.)
With this notation, we obtain analogous inequalities to the Lemma~\ref{lem:honest-prob} proof:
\begin{align*}
p_{\widetilde{\clE}} &= p_\top - \Pr\left[(|S| \geq \gamma l)\land\left(\sum_{i\in T}W_i \geq (1-\eps)|T|\right)\land\left(\sum_{i\in S\setminus T}W_i < (1-2\eps)|S\setminus T|\right)\right]\\
&\geq p_\top - 2^{-2\eps^2\gamma l} \Pr\left[|S| \geq \gamma l\right]\\
&\geq p_\top - 2^{-2\eps^2\gamma l} 
,
\end{align*}
where the inequality in the second line is again due to Lemma~\ref{fc:serfling}. Since we also again have $p_{\widetilde{\clE}} = p_{\widetilde{\clE}|\top} p_\top$, the above bound gives $p_{\widetilde{\clE}|\top} \geq (p_\top - 2^{-2\eps^2\gamma l})/p_\top = 1 - 2^{-2\eps^2\gamma l}/p_\top$.
In other words, we have shown that the probability of the event \eqref{eq:event_thBrho} in the state $\hBrho$ is at least $1 - 2^{-2\eps^2\gamma l}/p_\top$. Recalling that $\thBrho$ is the state conditioned on that event, from the definition of the $1$-norm distance it immediately follows that
$\Vert \hBrho - \thBrho\Vert_1 \leq 2(1-p_{\widetilde{\clE}|\top}) \leq 2\cdot2^{-2\eps^2\gamma l}/p_\top$.

In $\thBrho$, the $K^\B$ thus obtained differs from $K^\A$ in at most $2\eps|S|$ many indices. The number of $|S|$-bit binary strings that can disagree with $K^\B$ in at most $2\eps|S|$ places is at most $2^{h_2(2\eps)|S|} \leq 2^{h_2(2\eps)l}$. Hence,
\[ \sfH_0(K^\A|K^\B)_{\thBrho} \leq h_2(2\eps)l\]
and this implies that the $(2\cdot 2^{-2\eps^2\gamma l}/p_\top)$-smoothed entropy of $\hBrho$ is at most $h_2(2\eps)l$. Hence by Fact \ref{fc:info-recon}, we get the required result.\footnote{Strictly speaking, in order to actually implement the \cite{Ren-th} protocol (Fact~\ref{fc:info-recon}), it is not sufficient to only have the upper bound $\sfH^\delta_0(K^\A|K^\B) \leq h_2(2\eps)l$. Rather, for each value of $K^\B$ Bob needs to know the set of $K^\A$ values such that $\Pr[K^\A|K^\B]>0$, where the probability is with respect to some distribution in the $\delta$-ball that attains $\sfH_0(K^\A|K^\B) = h_2(2\eps)l$. 
The proof we give here indeed characterizes this set, namely the set of $K^\A$ that differ from $K^\B$ in at most $2\eps l$ indices, so it is possible to apply that protocol. 
}
\end{proof}

\subsection{Secrecy}

We now prove some lemmas which, in the context of the alternative security definitions in Section~\ref{sec:altdef}, will imply the secrecy property (as we shall discuss in Section~\ref{sec:altdefproof}). In the context of our composable security analysis (in Section~\ref{sect:comp-sec}), however, these will just be intermediate results in the full composable security proof.

Specifically, we prove two lower bounds for the (smoothed) min-entropy of $K^\A$ in the states $\rho$ and $\sigma$, conditioned on Bob and Eve's side information and the randomness $R$. These will later allow us to show secrecy of the protocol via the Leftover Hashing Lemma. 

\begin{lemma}\label{lem:Hmin-rho}
If Bob plays honestly and $p_\top \geq 2\cdot2^{-2\eps^2\gamma l}$, then the state $\hBrho_{C_2RK^\A\tB B_T\tE}$ satisfies
\begin{equation*}
\sfH^{\delta_\top}_\infty(K^\A|C_2RB_T\tE)_{\hBrho} \geq d^\E(c^\E_\alpha - 2\eps)^3\alpha^2l - 2\eps^2\gamma l - 2\gamma l - h_2(2\eps)l - \log(1/\lEC),
\end{equation*}
where $c^\E_\alpha, d^\E$ are the constants from Theorem~\ref{thm:mse-parrep}, and $\delta_\top = \frac{2\cdot2^{-2\eps^2\gamma l}}{p_\top}$.
\end{lemma}
\begin{proof}
We follow the proof approach of \cite{Vid17} for the protocol in \cite{JMS17}. First we shall bound $\sfH^{\delta_\top}_\infty(K^\A|SX_SY_S\tE)_{\hBrho}$. Consider the $\MSE^{t/l}$ game with $t= (1-2\eps)l$ being played on the shared state between Alice, Bob and Eve. Let $W_i$ denote the indicator variable for the event $a_i[y_i] = b_i[x_i]$ and $V_i$ denote the indicator variable for the event $a_i[y_i] = c_i$ (Eve's guess for the $i$-th bit). The winning condition for the $i$-th game is $W_i\land V_i = 1$.

By Theorem~\ref{thm:mse-parrep}, the winning probability of $\MSE^{t/l}$ on the original state shared by Alice, Bob and Eve is at most $2^{-d^\E(c^\E_\alpha-2\eps)^3\alpha^2l}$. 
We first consider the state $\thBrho$ defined in the proof of Lemma~\ref{lem:Bob-guess}, i.e.~this original state conditioned on the event~\eqref{eq:event_thBrho}. Denoting the probability of this conditioning event as $p_{\widetilde{\clE}}$, we can bound the winning probability of $\MSE^{t/l}$ on the state $\thBrho$ as
\[ \Pr_{\thBrho}[\text{Win } \MSE^{t/l}] \leq \frac{2^{-d^\E(c^\E_\alpha-2\eps)^3\alpha^2l}}{p_{\widetilde{\clE}}}.\]
By construction, in $\thBrho$ there is always some subset of $S$ with size at least $(1-2\eps)|S|$ on which $W_i=1$ for each $i$. Hence whenever $V_i=1$ for all $i\in S$, $\MSE^{t/l}$ is won. This implies
\[ \Pr_{\thBrho}\left[\sum_{i\in S}V_i = |S|\right] \leq \Pr_{\thBrho}[\text{Win } \MSE^{t/l}].\]
But the probability of $V_i=1$ for all $i\in S$ is the probability that Eve is able to guess $a_i[y_i]$ given $x_iy_i$ for all $i\in S$. Hence from Fact \ref{fc:guess-prob},
\begin{align*}
H_\infty(K^\A|SX_SY_S\tE)_{\thBrho} & = \log\left(\frac{1}{\Pr_{\thBrho}\left[\sum_{i\in S}V_i = |S|\right]}\right) \\
& \geq d^\E(c^\E_\alpha-2\eps)^3\alpha^2l - \log(1/p_{\widetilde{\clE}}) .
\end{align*}
We now relate this to the state of interest $\hBrho$ (i.e.~the state conditioned only on the event $(|S| \geq \gamma l)\land\left(\sum_{i\in T}W_i \geq (1-\eps)|T|\right)$) by recalling that $\Vert\hBrho - \thBrho\Vert_1 \leq \delta_\top$ as shown in the proof of Lemma~\ref{lem:Bob-guess},
and hence the $\delta_\top$-smoothed min-entropy of $\hBrho$ is at least the above value as well. Furthermore, since the Serfling bound in the form of Lemma~\ref{fc:serfling} also implies $p_{\widetilde{\clE}} \geq p_\top-2^{-2\eps^2\gamma l}$ (this addresses a minor issue in the \cite{Vid17} proof, which just replaced $p_{\widetilde{\clE}}$ with $p_\top$ directly), we have
\begin{align*}
H^{\delta_\top}_\infty(K^\A|SX_SY_S\tE)_{\hBrho} 
& \geq d^\E(c^\E_\alpha-2\eps)^3\alpha^2l - \log\frac{1}{p_\top-2^{-2\eps^2\gamma l}}
\\& \geq d^\E(c^\E_\alpha-2\eps)^3\alpha^2l - 2\eps^2\gamma l,
\end{align*}
using the condition $p_\top \geq 2\cdot2^{-2\eps^2\gamma l}$.

Finally, the other parts of $R$ besides $SX_SY_S$ are $T$ (which Eve already has in $\tE$), $Y'_{\overline{T}}$, $H$ and $U_1U_2$. But $K^\A$ is independent of $Y'_{\overline{T}}$ given $SY_S$ (since $K^\A$ is produced by a measurement in Alice's boxes only, which has no relation to Bob's string $Y'_{\overline{T}}$ apart from the values $SY_S$ used to specify which bits of Alice's input to include in $K^\A$), and $H$ and $U_1U_2$ are independent of everything else (since by construction they are drawn uniformly and independently of all other registers here). Hence giving Eve these extra registers in $R$ makes no difference. Lastly, to handle $C_2$ and $B_T$ (which are not independent of $K^\A$), we simply note that $C_2$ is at most $h_2(2\eps)l+\log(1/\lEC)$ bits and $B_T$ is at most $2\gamma l$ bits, hence by Fact \ref{fc:Hinf-ch-rule} we get the desired result.
\end{proof}

For proving the next bound, we shall need the following fact, which is easily proven by a summation relabelling:
\begin{fact}\label{fc:OTP}
Consider a CQ state $\rho_{ZQ}$ where $Z$ is an $s$-bit classical register. If we select an independent uniformly random $U\in\{0,1\}^s$ and generate a register $C=Z \oplus U$, then the resulting global state,
\begin{align}
\sum_{z\in\{0,1\}^s}\sum_{u\in\{0,1\}^s}\frac{1}{2^s}\sfP_Z(z) \state{z\oplus u}_{C} \otimes \state{u}_{U} \otimes \state{z}_{Z} \otimes \rho_{Q|z},
\label{eq:OTPstate}
\end{align}
is equal to
\begin{align}
\sum_{z\in\{0,1\}^s}\sum_{u\in\{0,1\}^s}\frac{1}{2^s}\sfP_Z(z) \state{u}_{C} \otimes \state{z\oplus u}_{U} \otimes \state{z}_{Z} \otimes \rho_{Q|z}.
\label{eq:OTPswapped}
\end{align}
\end{fact}
When applying this fact, we shall take $U$ to correspond to $(U_1,U_2)$ in Protocol~\ref{sect:protocol}, which is basically a one-time pad. Intuitively, Fact~\ref{fc:OTP} expresses a symmetry\footnote{For the purposes of our proof, we technically do not need such an exact symmetry --- it would suffice to have 
functions $f$ and $g$ such that~\eqref{eq:OTPstate} and~\eqref{eq:OTPswapped} are equal when register $C$ in~\eqref{eq:OTPstate} is set to $f(z,u)$ and register $U$ in~\eqref{eq:OTPswapped} is set to $g(z,u)$. However, our proof is easier to describe using the formulation shown. Furthermore, in principle our proofs also hold using a relaxed version of this statement in which~\eqref{eq:OTPstate} and~\eqref{eq:OTPswapped} are only $\lambda_\mathrm{OTP}$-close in $\ell_1$ distance for some value $\lambda_\mathrm{OTP} > 0$, at the cost of increasing our composable security parameter by $O(\lambda_\mathrm{OTP})$.} between the ``ciphertext'' and the ``padding string'' when applying a one-time pad 
--- while we usually think of the ciphertext as taking the value $Z\oplus U$ and the padding string as taking the independent uniform value $U$, this fact implies that we have an exactly equivalent situation by thinking of the ciphertext as taking the value $U$ and the padding string as taking the value $Z\oplus U$. 
We use this to prove the following lemma (for all possible behaviours by Bob and Eve --- note that while dishonest Bob does not have to honestly report the ``raw'' outcomes of measurements on his states, such behaviour can simply be absorbed into the strategy he uses to generate $\tB$ and $B_T$):
\begin{lemma}\label{lem:Hmin-sigma}
The probability $p_{\text{\cmark}|\top}$ is independent of the message $m$. Furthermore, letting $R'$ denote all the registers in $R$ except $U_1$, the state $\sigma_{K^\A CR'\tB B_T\tE}$ satisfies
\[ \sfH_\infty(K^\A|CR'\tB B_T\tE)_{\sigma} \geq d^\B(c^\B_\alpha-\eps)^3\alpha^2(1-\gamma)l - \log(1/p_\top p_{\text{\cmark}|\top}) - \gamma l - h_2(2\eps)l - \log(1/\lEC),\]
where $c^\B_\alpha, d^\B$ are the constants from Theorem~\ref{thm:msb-parrep}.
\end{lemma}
\begin{proof}
Recall that $U$ was initially generated as a uniformly random value independent of all the other registers, and that it is not revealed to Bob and Eve until the final time $t_5$. Also, by Fact~\ref{fc:OTP}, we know that at the point at which $C$ is generated, the global state remains the same if we swap the roles of the registers $C$ and $U$. This means that it is perfectly equivalent to instead consider the following ``virtual'' process: Bob and Eve generate an independent uniformly random value in the register $C$, and this is used to generate a register $U=(M\oplus h(K^\A)\oplus C_1, \syn(K^\A) \oplus C_2)$ which is given to Alice only (until time $t_5$). We stress that this virtual process does not correspond to a physical procedure which is actually performed, but it produces exactly the same state as the original protocol, so it is valid to study it in place of the original protocol.

With this process in mind, it is clear that $p_{\text{\cmark}|\top}$ is independent of $m$, since the only register that depends on $m$ at that point is always with Alice and not acted upon. We shall now prove
\[ \sfH_\infty(K^\A|C R'' \tB B_T\tE)_{\sigma} \geq d^\B(c^\B_\alpha-\eps)^3\alpha^2(1-\gamma)l - \log(1/p_\top p_{\text{\cmark}|\top}) - \gamma l,\]
where $R''$ denotes all the registers in $R$ except $U_1 U_2$.
From there the desired bound would follow by subtracting the number of bits in $U_2$, via Fact~\ref{fc:Hinf-ch-rule}.

Under the virtual process, the register $C$ is locally generated from the joint system of Bob and Eve, without access to any of Alice's registers. Hence we can proceed as in the proof of Lemma~\ref{lem:Hmin-rho}, this time by considering the game $\MSB^{t/(1-\gamma)l}_\alpha$ for $t=(1-\eps)(1-\gamma)l$ on the set $\overline{T}=[l]\setminus T$ between Alice and the joint system of Bob and Eve. In this case, however, an important difference is that we allow the output distribution to depend on which subset $T$ is supplied first (this information is implicitly included in $R''$). This works because Alice's accept condition at this point is directly the condition that enough rounds are won on the entire set, instead of a small test subset. In particular, this means that for each input order we can directly study the corresponding state conditioned on accepting, without relating it to some ``nearby'' state that is independent of the input order (this was the only part of the Lemma~\ref{lem:Hmin-rho} proof that required the condition on honest Bob's boxes, in order to apply the Serfling bound). 

However, even apart from the $|S| \geq \gamma l$ conditioning (whose probability is included in $p_\top$), the input distribution in $\overline{T}$ is not quite right for $\MSB^{t/(1-\gamma)l}$. 
To see this, let us recap the distribution in the actual protocol of the sets $S$ and $T$ conditioned on $|S| \geq \gamma l$, according to step~\ref{alg:rand-subsets}:
\begin{itemize}
\item Let $\sfP_{ST}$ denote the distribution of the variables $S$ and $T$ in the actual protocol, where $S \subseteq [l]$ is generated by choosing each index independently with probability $(1-\alpha)$.
Let $\clE_1$ denote the event that $|S|\geq \gamma l$ when $S$ is drawn according to $\sfP_S$, 
and $\sfP_{ST|\clE_1}$ is the joint distribution of $S$ and $T$ conditioned on this event (in which case $T$ is a uniformly random subset of $S$).
\end{itemize}
To relate this to the game $\MSB^{t/(1-\gamma)l}_\alpha$ (for arbitrary $t$), we would like to consider the game to be played on $\overline{T}$, and consider $S \cap \overline{T}$ to specify the subset of $\overline{T}$ on which Alice and Bob get nontrivial inputs. 
However, since in $\sfP_{ST|\clE_1}$ the subset $S$ was drawn first and $T$ was then drawn as a subset of it, this somewhat affects the distribution of $S \cap \overline{T}$ within $\overline{T}$ (in that $S \cap \overline{T}$ is not distributed according to choosing each index in $\overline{T}$ independently with probability $(1-\alpha)$). To obtain sets $ST$ with a distribution suitable for $\MSB^{t/(1-\gamma)l}_\alpha$ in the preceding sense, we could instead consider the following distribution:
\begin{itemize}
\item Let $\sfQ_{ST}$ denote the distribution obtained as follows: choose $T$ as a uniformly random subset of $[l]$ with size $\gamma l$; independently generate $S \subseteq [l]$ by choosing each index independently with probability $(1-\alpha)$. Since $S$ and $T$ are chosen independently, it can be seen that $\sfQ_{ST}$ indeed gives the right input distribution for $\MSB^{t/(1-\gamma)l}_\alpha$ played on $\overline{T}$ with Alice and Bob getting nontrivial inputs on $S \cap \overline{T}$.
\end{itemize}
Now let $\clE_2$ be the event that when $S$ and $T$ are drawn according to $\sfQ_{ST}$, $T$ is a subset of $S$; $\sfQ_{ST|\clE_2}$ is the distribution conditioned on this event.
Since $S$ and $T$ are independent in $\sfQ_{ST}$, the distribution $\sfQ_{ST|\clE_2}$ can be equivalently described as first choosing $S$ followed by $T$ according to their respective distributions, and then conditioning on $T \subseteq S$ --- but from this, we see that in fact $\sfQ_{ST|\clE_2}$ is just the same distribution as $\sfP_{ST|\clE_1}$ (note that $T \subseteq S$ implicitly includes the condition $|S|\geq\gamma l$, since $T$ always has size $\gamma l$).

Since conditioning on an event with probability $q$ cannot increase probabilities by more than a factor of $1/q$, we see that when studying the protocol we can compensate for using the actual distribution $\sfP_{ST|\clE_1}$ (i.e.~$\sfQ_{ST|\clE_2}$) instead of the ``correct game distribution'' $\sfQ_{ST}$ by rescaling all upper bounds on probabilities by $1/\Pr_{\sfQ}[T\subseteq S]$. From the definition of $\sfQ_{ST}$, we can easily compute that $\Pr_{\sfQ}[T\subseteq S] = (1-\alpha)^{\gamma l}$, hence we shall include this factor when bounding the winning probability of $\MSB_\alpha^{t/(1-\gamma)l}$ under the actual protocol's input distribution.

With this, we return to the topic of bounding the min-entropy.
The state $\sigma$ is conditioned on the first round of $\MSB_\alpha^{t/(1-\gamma)l}$ winning, as well as the initial conditioning of outputting $\top$. The probability of winning both the first and second rounds  on an unconditioned state with an unconditioned input distribution is at most $2^{-d^\B(c^\B_\alpha - \eps)^3\alpha^2(1-\gamma)l}$. Hence,
\[ \sfH_\infty(K^\A|R''\tB B_T\tE)_{\sigma} = \log\left(\frac{1}{\Pr_\sigma\left[\text{Win } \MSB^{t/(1-\gamma)l}_\alpha\right]}\right) \geq \log\left(\frac{(1-\alpha)^{\gamma l}p_\top p_{\text{\cmark}|\top}}{2^{-d^\B(c^\B_\alpha - \eps)^3\alpha^2(1-\gamma)l}}\right),\]
which yields the desired bound (noting that $\log(1-\alpha) > -1$ since we took the condition $\alpha<\frac{1}{2}$ in Lemma~\ref{lem:honest-prob} to be satisfied). We get the min-entropy instead of smoothed min-entropy here (and we do not need an explicit lower bound on $p_\top p_{\text{\cmark}|\top}$, unlike Lemma~\ref{lem:Hmin-rho}) because Alice checks the condition on the entire $\overline{T}$ instead of a test subset, so we do not need to consider a ``nearby'' state in place of the actual conditional state. 
\end{proof}

\subsection{Parameter choices}\label{sec:parameters}

Take any values of $\alpha\in(0,\frac{1}{2})$ and $\eps\in(0,1)$ satisfying 
\begin{align}
\min\{
d^\E(c^\E_\alpha - 2\eps)^3\alpha^2 \;,\;
d^\B(c^\B_\alpha-\eps)^3\alpha^2
\}
> h_2(2\eps).
\label{eq:alphaepsbnd}
\end{align}
To construct this more explicitly, we could focus on a fixed choice of $\alpha$ (say, $\alpha=0.4$), in which case there clearly exists $\eps_0$ such that $\min\{
d^\E(c^\E_\alpha - 2\eps_0)^3\alpha^2 \;,\;d^\B(c^\B_\alpha-\eps_0)^3\alpha^2
\}
> h_2(2\eps_0)$ (by noting the behaviour of both sides as $\eps_0\to0$). This value is then a valid choice of $\eps_0$ in Theorem~\ref{thm:main-achieve}, since~\eqref{eq:alphaepsbnd} would then be satisfied for any $\eps \in (0,\eps_0]$ (using that fixed choice of $\alpha$).

Now take some $\lcom,\lCI,\lEC \in (0,1]$ and some desired message length $n$, and choose $l$ large enough such that when setting 
\begin{align}\label{eq:gamma-choice}
\gamma=
\frac{1}{\eps^2l}\max \left\{
8\log\frac{2}{\lcom} \;,\; 
\frac{1}{2}\log\frac{8}{\lCI} \;,\;
\frac{1}{2}\log\frac{8}{\lEC}\right\},
\end{align}
the following conditions are satisfied: firstly, $\gamma<\frac{1}{2}$, secondly, 
\begin{align}
2^{-(1-2\gamma)^2l/8} \leq \frac{\lcom}{2},
\label{eq:combound}
\end{align}
and lastly,
\begin{align}
\begin{aligned}
n &\leq d^\E(c^\E_\alpha - 2\eps)^3\alpha^2l - 2\eps^2\gamma l - 2\gamma l - h_2(2\eps)l - \log(1/\lEC) - 2\log(2/\lCI), \\
n &\leq d^\B(c^\B_\alpha-\eps)^3\alpha^2(1-\gamma)l - \gamma l - h_2(2\eps)l - \log(1/\lEC) - 2\log(2/\lCI).
\end{aligned}
\label{eq:nbound}
\end{align}
The conditions $\gamma<\frac{1}{2}$ and \eqref{eq:combound} ensure 
that the conditions on $\gamma, l$ for Lemma~\ref{lem:honest-prob} are satisfied. Since the choice of $\gamma$ in~\eqref{eq:gamma-choice} satisfies $\gamma\to0$ as $l\to\infty$, these conditions can always be satisfied by taking sufficiently large $l$. Furthermore, given that~\eqref{eq:alphaepsbnd} holds, for any $n$ the conditions~\eqref{eq:nbound} will be satisfied at sufficiently large $l$, because all the $\gamma l$ terms are independent of $l$ for this choice of $\gamma$. 

For these parameter choices, together with a choice of syndrome satisfying Lemma~\ref{lem:Bob-guess} for the specified $\lEC$ value, the described ECD protocol satisfies the following security properties (which hold independently of Conjecture~\ref{conj:DI}). They are qualitatively similar to the notions of completeness, correctness and secrecy laid out in the alternative definition in Section~\ref{sec:altdef}: we discuss in Section~\ref{sec:altdefproof} how they imply that the properties in that alternative definition are indeed satisfied, whereas in Section~\ref{sect:comp-sec} we use these results to prove composable security.

\begin{lemma}\label{lem:com}
Given parameter choices satisfying~\eqref{eq:alphaepsbnd}--\eqref{eq:nbound}, if Bob and Eve are honest, then
\begin{align*}
1-p_\top \leq \lcom.
\end{align*}
Also, if Bob is honest and $D=1$ (i.e.~Alice requests a deletion certificate), then
\begin{align*}
p_{\top}(1-p_{\text{\cmark}|\top}) 
\leq \lcom.
\end{align*}
\end{lemma}
\begin{proof}
This follows immediately from Lemma~\ref{lem:honest-prob}, recalling that we chose parameters such that $2^{-(1-2\gamma)^2l/8} \leq \lcom/2$ and $\gamma \geq 8\log(2/\lcom)/(\eps^2l) \geq \log(1/\lcom)/(2\eps^2l)$.
\end{proof}

\begin{lemma}\label{lem:corr}
Given parameter choices satisfying~\eqref{eq:alphaepsbnd}--\eqref{eq:nbound}, if Bob is honest, then 
for any specific message value $M=\mval$ we have 
\begin{align*}
p_\top 
\Pr[\tM\neq m | O=\top] 
\leq \lEC.
\end{align*}
\end{lemma}
\begin{proof}
Observe that conditioned on $O=\top$ in step~\ref{alg:E-test} of the protocol (so that Bob gets the value $C_1 = m\oplus h(K^\A)\oplus U_1$ from Alice, where $m$ is the specific message value we discuss here), honest Bob's final value $\tM = C_1\oplus h(\tK^\A)\oplus U_1$ will be equal to $m$ whenever his guess $\tK^\A$ matches Alice's value $K^\A$. 
Hence conditioned on $O=\top$, the probability that his final value $\tM$ differs from $m$ 
(i.e.~$
\Pr[\tM\neq m | O=\top] 
$) 
is at most the probability that his guess was wrong, i.e.~$\tK^\A \neq K^\A$.
Recalling that we chose parameters such that $\gamma\geq\log(8/\lEC)/(2\eps^2l)$, Lemma~\ref{lem:Bob-guess} tells us for honest Bob, the probability that $\tK^\A \neq K^\A$ is at most
\[
\frac{2\cdot2^{-2\eps^2\gamma l}}{p_\top} + \frac{\lEC}{2}
\leq \frac{\lEC}{2p_\top} + \frac{\lEC}{2}
\leq \frac{\lEC}{p_\top},
\]
which gives the desired result.
\end{proof}

\begin{lemma}\label{lem:ct-indist}
Given parameter choices satisfying~\eqref{eq:alphaepsbnd}--\eqref{eq:nbound}, we have for any specific message value $M=\mval$:
\begin{align*}
p_\top p_{\text{\cmark}|\top} \norm{\sigma_{CR\tB B_T\tE}(\mval,1) - \sigma_{CR\tB B_T\tE}(0^n,1)}_1 \leq 2\lCI,
\end{align*}
and if Bob plays honestly,
\begin{align*}
p_\top \norm{\hBrho_{CRB_T\tE}(\mval,0) - \hBrho_{CRB_T\tE}(0^n,0)}_1 \leq 2\lCI.
\end{align*}
\end{lemma}

\newcommand{\regPA}{S}
\begin{proof}
We first prove the expression for $\hBrho$, which is under the assumption that Bob plays honestly. Observe that if $p_\top < 2\cdot2^{-2\eps^2\gamma l}$, then we are already done since we chose $\gamma\geq\log(8/\lCI)/(2\eps^2l)$. Hence we can take $p_\top \geq 2\cdot2^{-2\eps^2\gamma l}$, in which case we can put together the bound in Lemma~\ref{lem:Hmin-rho} with the first of the bounds on $n$ in~\eqref{eq:nbound} to get
\begin{align*}
\frac{1}{2}\left(\sfH^{\delta_\top}_\infty(K^\A|C_2RB_T\tE)_{\hBrho} - n\right) \geq \log(2/\lCI) .
\end{align*}
Let $\regPA$ be a register storing the value of the hash $h(K^\A)$. Recalling that $\gamma\geq\log(8/\lCI)/(2\eps^2l)$, the Leftover Hashing Lemma then implies 
\begin{align*}
\norm{\hBrho_{\regPA C_2 R B_T\tE} - \frac{\Id_{\regPA}}{2^n}\otimes\hBrho_{C_2 R B_T\tE}}_1 \leq 2^{- \log(2/\lCI)} + 2 \frac{2\cdot2^{-2\eps^2\gamma l}}{p_\top} 
\leq \frac{\lCI}{2} + \frac{\lCI}{2 p_\top}
\leq \frac{\lCI}{p_\top}.
\end{align*}
The state on these registers is independent of the value of $M$. Now for any message value $m$, let $\clE^{m}$ denote the map that generates the ciphertext register $C_1 = m \oplus s$ by reading $s$ off the register $\regPA$ and then tracing it out. By the properties of the one-time pad, we know that
\begin{align*}
\clE^{\mval}\!\left(\frac{\Id_{\regPA}}{2^n}\otimes\hBrho_{C_2 R B_T\tE}\right) = \clE^{0^n}\!\left(\frac{\Id_{\regPA}}{2^n}\otimes\hBrho_{C_2 R B_T\tE}\right).
\end{align*}
This yields the desired result:
\begin{align*}
&\norm{\hBrho_{C R B_T\tE}(\mval,0) - \hBrho_{C R B_T\tE}(0^n,1)}_1 \\
=& \norm{\clE^{\mval}\!\left(\hBrho_{\regPA C_2 R B_T\tE}\right) - \clE^{0^n}\!\left(\hBrho_{\regPA C_2 R B_T\tE}\right)}_1 \\
\leq& \norm{\clE^{\mval}\!\left(\hBrho_{\regPA C_2 R B_T\tE}\right) - \clE^{\mval}\!\left(\frac{\Id_{\regPA}}{2^n}\otimes\hBrho_{C_2 R B_T\tE}\right)}_1 + \norm{\clE^{0^n}\!\left(\frac{\Id_{\regPA}}{2^n}\otimes\hBrho_{C_2 R B_T\tE}\right) - \clE^{0^n}\!\left(\hBrho_{\regPA C_2 R B_T\tE}\right)}_1 \\
\leq& \frac{2\lCI}{p_\top},
\end{align*}
using Fact~\ref{fc:chan-l1} in the last line.

For $\sigma$, it is again easier to analyze the situation by using Fact~\ref{fc:OTP} to switch to the virtual process of $C$ being a uniformly random value and $U$ being set to $U = (M\oplus h(K^\A)\oplus C_1, \syn(K^\A) \oplus C_2)$. We then follow a similar argument as above: by Lemma~\ref{lem:Hmin-sigma} and the second bound on $n$ in~\eqref{eq:nbound}, we have
\begin{align*}
\frac{1}{2}\left(\sfH_\infty(K^\A|CR'\tB B_T\tE)_{\sigma} - n\right) \geq \log(2/\lCI) - \frac{1}{2}\log(1/p_\top p_{\text{\cmark}|\top}) \geq \log(2/\lCI) - \log(1/p_\top p_{\text{\cmark}|\top}),
\end{align*}
where $R'$ denotes all the registers in $R$ except $U_1$. Defining $\hat{\clE}^{m}$ the same way as $\clE^{m}$ above, except with the output register being $U_1$ instead of $C_1$, we follow the same line of reasoning and obtain 
\begin{align*}
\norm{\sigma_{CR\tB B_T\tE}(\mval,1) - \sigma_{CR\tB B_T\tE}(0^n,1)}_1
= \norm{\hat{\clE}^{\mval}\!\left(\sigma_{CR'\tB B_T\tE}\right) - \hat{\clE}^{0^n}\!\left(\sigma_{CR'\tB B_T\tE}\right)}_1 
\leq \frac{2\lCI}{p_\top p_{\text{\cmark}|\top}}.
\end{align*}
(In fact, a tighter bound of $\lCI/(p_\top p_{\text{\cmark}|\top})$ holds here since the min-entropy bound for $\sigma$ is not smoothed, but we shall not track this detail.)
\end{proof}

\subsection{Security under the alternative definition}
\label{sec:altdefproof}

From the above properties, it is not difficult to show that the required properties for the alternative security definition in Section~\ref{sec:altdef} are satisfied. Firstly, the completeness requirements~\eqref{eq:defncomBE}--\eqref{eq:defncomB} are precisely the properties shown in Lemma~\ref{lem:com}. Next, the correctness requirement~\eqref{eq:defncorr} is precisely the property shown in Lemma~\ref{lem:corr}. 

As for the secrecy requirements~\eqref{eq:CI-BEbeforedec}--\eqref{eq:CI-E}, we first observe that~\eqref{eq:CI-BEbeforedec} trivially holds over the required duration (i.e.~until the decryption key is released), because the only point in the encryption phase that potentially depends on the message is the value $M\oplus h(K^\A)\oplus U_1$ in step~\ref{alg:C-send}, but by Fact~\ref{fc:OTP} this register is in fact also independent of the message (until $U_1$ is released) since $U_1$ serves as a one-time-pad that is kept secret until being released in the decryption key $R$.
As for the requirement~\eqref{eq:certdel-BE}, it also trivially holds before the decryption key $R$ is released, for the same reason; whereas after that point, it is ensured by the first bound in the Lemma~\ref{lem:ct-indist} statement
(together with the fact that any further processing of Bob and/or Eve's registers cannot increase that trace distance, by the data-processing inequality). 
Finally, an analogous argument holds for the requirement~\eqref{eq:CI-E}: it trivially holds before the decryption key $R$ is released, and afterwards it is ensured by the second bound in the Lemma~\ref{lem:ct-indist} statement.

\section{Composable security proof}\label{sect:comp-sec}
In this section, we prove our main security result, which implies Theorem~\ref{thm:main-achieve}. The argument essentially only depends on Fact~\ref{fc:OTP} and Lemmas~\ref{lem:com}--\ref{lem:ct-indist}, without requiring the details of the analysis leading up to those lemmas.
\begin{theorem}
Assuming Conjecture \ref{conj:DI}, there exists a universal constant $\eps_0 \in (0,1)$ such that for any $\eps \in(0,\eps_0]$, $\lcom,\lCI,\lEC \in (0,1]$ and $n \in \mathbb{N}$, there exist parameter choices for Protocol~\ref{prot:ECD} such that it 
constructs the $\mathrm{ECD}_n$ functionality from the resources $\clR$, $\clC$ and $(\clB^1_1\ldots \clB^1_l, \clB^2_1\ldots \clB^2_l)_\eps$, within distance $\lambda=2\lcom + \lCI + \lEC$.
\end{theorem}
As noted in the previous section, using the value of $\eps_0$ specified there allows us to choose parameters such that \eqref{eq:alphaepsbnd}--\eqref{eq:nbound} are satisfied, in which case Lemmas~\ref{lem:com}--\ref{lem:ct-indist} hold and we can use them in our subsequent proof.
To prove composable security according to Definition~\ref{def:comp}, we need to consider the four possible combinations of honest/dishonest Bob and Eve's behaviours, and for each case bound the distinguishing probability between the real functionality with the honest parties performing the honest protocol, versus the ideal functionality with some simulator attached to the dishonest parties' interfaces. We shall construct appropriate simulators and argue that for a distinguisher interacting with either scenario, the states held by the distinguisher in the two scenarios differ in $\ell_1$ distance by at most
$2\lambda = 4\lcom + 2\lCI + 2\lEC$
at all times. This implies the distinguishing advantage is bounded by $\lambda$ via Fact~\ref{fc:chan-l1}, since the process of the distinguisher producing a value on the guess register $G$ can be viewed as a channel applied to the states it holds. 

Note that it suffices to consider only the points where output registers are released to the distinguisher, since by Fact~\ref{fc:chan-l1}, any operations the distinguisher performs between these points cannot increase the $\ell_1$ distance. Furthermore, we observe that for classical inputs, it is not necessary to bound the distinguishability for all possible input distributions that the distinguisher could supply --- it suffices to find a bound that holds for all specific values that could be supplied as input, since by convexity of the $\ell_1$ norm, the same bound would hold when considering arbitrary distributions over those input values. In particular, for the subsequent arguments we shall assume the distinguisher supplies a specific value $\mval$ for the input $M$, and we shall split the analysis into different cases for the two possible values for the input $D$.

\begin{remark}
In the following proofs, we shall construct simulators by explicitly using Fact~\ref{fc:OTP}, but an alternative approach appears possible, which we sketch out here. First, observe that the use of the one-time pad $U$ in Protocol~\ref{prot:ECD} is in fact a composably secure realization of a functionality we could call a \emph{trusted-sender channel with delay}, which is defined in exactly the same way as the \emph{channel with delay} in~\cite{VPdR19}, except that only the recipient is potentially dishonest.\footnote{Proving this would be fairly simple: just follow the argument in the composable security proof for the one-time pad~\cite{PR14}, except with appropriate changes in timing.} If we now view Protocol~\ref{prot:ECD} as sending the value $(M\oplus h(K^\A), \syn(K^\A))$ through a composably secure implementation of a trusted-sender channel with delay, we can safely assume that the $C$ register gives no information to the dishonest parties about Alice's outputs until the final step, which may be a helpful perspective to keep in mind when considering the proofs below. Essentially, our approach below has the simulator in the composable security proof for the trusted-sender channel with delay ``built into'' the argument directly, by repeated use of Fact~\ref{fc:OTP}.
\end{remark}

Since in this section we are proving results in the Abstract Cryptography framework, we shall now explicitly refer to various actions in the Protocol~\ref{prot:ECD} description as being performed by the protocols $\prot^{\A}$ and $\prot^{\B}$ rather than Alice and Bob (recall the discussion in Remark~\ref{remark:agency}).

\subsection{Dishonest Bob and Eve}
\label{sec:compBE}

The distinguisher's task in this case is to distinguish $\fR{\A}{\B\E}$ and $\fI{\B\E}$. 
As mentioned in the introduction, intuitively speaking we shall choose the simulator $\simu^{\B\E}$ here to run the honest protocol internally with a ``virtual'' simulated instance of the protocol $\prot^{\A}$ (and a ``dummy'' version of the protocol's inputs on Alice's interfaces, which we shall soon describe) --- the construction here is in fact rather similar to the QKD analysis in~\cite{PR14}; the example described there may be instructive in helping to understand our construction. 

To give a broad overview (the detailed description will subsequently follow; see also Figure~\ref{fig:simdBdE}), the simulator begins by following the actions of $\prot^{\A}$ in Protocol~\ref{prot:ECD}: it accepts the boxes from Eve, generates $S$ and $T$ the same way as step~\ref{alg:rand-subsets}, generates inputs the same way as in step~\ref{alg:rand-gen} and supplies them to the boxes, etc --- in particular, on step~\ref{alg:ini-state} it accepts some value of $b_T$ supplied at Bob's interface\footnote{For the purpose of the distinguishability argument, the value supplied here is chosen by the distinguisher; our full security analysis below will cover any possible such choice when trying to distinguish $\fR{\A}{\B\E}$ and $\fI{\B\E}$.}, then performs the checks in step~\ref{alg:A-test}, and so on.
A small difference occurs in step~\ref{alg:E-test}, where an output value $O$ is produced by $\prot^{\A}$ that would normally be sent out to Alice's interface on $\prot^{\A}$; here however the simulator is instead supposed to be connected to $\fideal_{\B\E}$ on that interface, and hence we shall instead make the simulator supply the value $O$ as the input values $O^\B$ and $O^\E$ on the relevant interfaces in the ideal functionality (this is again a very similar construction to the QKD analysis in~\cite{PR14}).
Upon reaching step~\ref{alg:message} we encounter a notable obstacle: the simulator does not have access to the true message value $\mval$ supplied to Alice's interface on $\fI{\B\E}$ (since the simulator is only attached to Bob and Eve's interfaces in $\fideal_{\B\E}$). Hence what it does at that step is that it instead acts as though the message had a dummy value $0^n$, and releases a corresponding ciphertext $C_1 = 0^n \oplus u_1$ (where $u_1$ is the uniformly random value drawn in step~\ref{alg:hashotp-gen}). In our security proof below, we will be showing that despite this substitution, the distinguisher still cannot easily distinguish $\fR{\A}{\B\E}$ and $\fI{\B\E}$. Proceeding onwards in a similar fashion, the last notable point is that if in the final steps the values of $D$ and $F$ are such that ideal functionality releases the true message value $\mval$, then the simulator will use $\mval$ to set the value on the register $U_1$ to $\mval \oplus h(k^\A) \oplus u_1$ instead, before releasing it as part of $R$\footnote{Note that the simulator is not constrained to use the actual resources of the ECD protocol. In particular, it does not have to use a temporarily private randomness source, which is why in some of the cases we describe, the value of $R$ the simulator reveals at the outer interface does not describe the randomness used by the simulated $\prot^{\A}$.} --- the idea here is that by setting $U_1$ to this value, the simulator makes it ``retroactively'' consistent with having encrypted the true message value $\mval$ in the previously released ciphertext $C_1$ using $U_1$ as a one-time-pad, despite the fact that $C_1$ was originally generated as a dummy value $C_1 = 0^n \oplus u_1$. At the end of this section (see Remark~\ref{remark:sim}), we briefly discuss some informal intuition of some concepts captured by this simulator construction.

\begin{figure}[!h]
\centering
\begin{tikzpicture}[scale=0.9]
\draw[rounded corners] (-1,-2.9) rectangle node {\small $\fideal_{\B\E}$} (0.6,3.5);
\draw[thick,->] (-0.5,4.2) -- node[xshift=-0.4cm] {\small $O^\E$} (-0.5,3.5);
\draw[thick,<-] (0.1,4.2) -- node[xshift=0.3cm] {\small $D$}  (0.1,3.5);
\draw[thick,<-] (0.6,2.8) -- node[yshift=0.3cm] {\small $O^\B$} (1.7,2.8);
\draw[thick,<-] (-3.3,1.8) -- node[xshift=-0.1cm,yshift=0.3cm] {\small $O=O^\B\land O^\E$} node[xshift=-1.3cm] {\color{dgreen} $t_1$} (-1,1.8);
\draw[thick,->] (0.6,1.8) -- node[yshift=0.3cm] {\small $O$} (1.7,1.8);
\draw[thick,->] (-3.3,0.8) -- node[yshift=0.3cm] {\small $M \in \{0,1\}^n$} node[xshift=-1.3cm] {\color{dgreen} $t_2$} (-1,0.8);
\draw[thick,->] (-3.3,-0.2) -- node[yshift=0.3cm] {\small $D \in \{0,1\}$} node[xshift=-1.3cm] {\color{dgreen} $t_3$} (-1,-0.2);
\draw[thick,->] (0.6,-0.2) -- node[xshift=2cm,yshift=0.3cm] {\small $(D,y'_{\overline{T}})$} (8.4,-0.2);
\draw[thick,<-] (-3.3,-1.2) -- node[yshift=0.3cm] {\small $F \in \{\text{\xmark,\cmark}\}$} node[xshift=-1.3cm] {\color{dgreen} $t_4$} (-1,-1.2);
\draw[thick,<-] (0.6,-1.2) -- node[yshift=0.3cm] {\small $F$} (1.7,-1.2);
\draw[thick,->] (0.6,-2.2) -- node[xshift=-3cm,yshift=0.3cm] {\small $\tM$} node[xshift=2cm,yshift=0.3cm] {\small Simulated $R$} (8.4,-2.2);

\draw[dashed] (1.7,4.2) -- (-1.6,4.2);
\draw[dashed] (-1.6,4.2) -- (-1.6,9.1);
\draw[dashed] (-1.6,9.1) -- (5,9.1);
\draw[dashed] (5,9.1) -- (5,-3.5);
\draw[dashed] (5,-3.5) -- (1.7,-3.5);
\draw[dashed] (1.7,-3.5) -- (1.7,4.2);

\draw[rounded corners] (2.5,6.4) rectangle node {\small $\clB^1_1\ldots \clB^1_l$} (4.3,7.9);
\draw[rounded corners] (5.9,6.4) rectangle node {\small $\clB^2_1\ldots \clB^2_l$} (7.7,7.9);
\draw[decorate, decoration=snake] (4.3,7.2) -- (5.9,7.2);

\draw[rounded corners] (-1.2,7.9) rectangle node {\small $\clR$} (-0.4,8.7);
\draw[thick,->] (-0.4,8.3) -- node[xshift=0.6cm,yshift=0.3cm] {\small $\;\;\, R=(x_S,y_S,y'_{\overline{T}},S,T,h,u_1,u_2)$} (3.4,8.3) -- (3.4,7.9);
\draw[thick,->] (2.1,8.3) -- (2.1,5.5) -- node[xshift=1.3cm,yshift=0.3cm] {\small $(T,y_T)$} (8.4,5.5);

\draw[thick] (8.4,4.6) -- node[xshift=2.5cm,yshift=0.3cm] {\small $b_T$} node [xshift=-0.6cm,yshift=0.3cm] {\small Step \ref{alg:E-test}} (-0.5,4.6) -- (-0.5,4.2);
\draw[thick] (3.3,4.6) -- (3.3,2.8) -- (1.7,2.8);

\draw[thick,->] (3.3,0.8) -- node[xshift=0.7cm,yshift=0.3cm] {\small 
$(u_1,\syn(k^\A)\oplus u_2)$
} (8.4,0.8);
\draw[thick] (8.4,-1.2) -- node[xshift=1.6cm,yshift=0.3cm] {\small $b'_{\overline{T}}$} node[xshift=-1.5cm,yshift=0.3cm] {\small Step \ref{alg:B-test}} (1.7,-1.2);
\end{tikzpicture}
\caption{Schematic for the case of dishonest Bob and Eve, in which we require a simulator $\simu^{\B\E}$ (depicted as the dashed region; refer to the main text for more description of the simulator's actions) acting on the ideal functionality $\fideal_{\B\E}$. As before in the honest functionality, the $F$ input and output is provided only if $D=1$, and the simulator only sends $y'_{\overline{T}}$ if $D=1$ as well. The version of $R$ the simulator releases has $U_1$ set to $\widetilde{m}\oplus h(k^\A)\oplus u_1$.}
\label{fig:simdBdE}
\end{figure}

With the broad picture in mind, we now give the full description of the simulator's actions, with a schematic depiction in Figure~\ref{fig:simdBdE}. Furthermore, after each step in the description, we derive bounds on the distinguishability of the real and ideal functionalities up to that point. 

\begin{itemize}

\item The simulator accepts the input states from the outer interface corresponding to Eve, and performs an internal simulated instance of $\prot^{\A}$ from the ECD protocol (producing a value $(T,y_T)$ which it outputs at the outer interfaces for Bob and/or Eve, and accepting a value $b_T$ supplied at the outer interface for Bob, which it uses in its simulated instance of $\prot^{\A}$), until step~\ref{alg:E-test}. The inner interface of the simulator then feeds the output $O$ of that step as the values $O^\B$ and $O^\E$ to the ideal functionality, which releases the same value $O$ to the distinguisher (and also the simulator, though the simulator does not need it).
\inset{
$\fR{\A}{\B\E}$ and $\fI{\B\E}$ are perfectly indistinguishable throughout this process, since no message value has been chosen yet, and hence the states produced by $\fR{\A}{\B\E}$ and $\fI{\B\E}$ are identical.
}

\item If $O=\perp$, the simulator stops here, apart from releasing its register $R$ at the end.
Otherwise, it continues on with its simulated instance of $\prot^{\A}$ from the ECD protocol, except that at the step where $C_1$ is to be generated, it instead prepares $C_1$ by generating an independent and uniformly random $u_1$ and setting $C_1 = u_1$. Furthermore, the simulator does not initialize a register $U_1$ yet --- this is valid because after generating $C_1$, the register $U_1$ is not needed at any point in the ECD protocol until the last step. The simulator then proceeds until it receives $D$ from the ideal functionality at the inner interface.
\inset{
By Fact~\ref{fc:OTP}, it is easily seen that the states produced by $\fR{\A}{\B\E}$ and $\fI{\B\E}$ remain perfectly indistinguishable throughout these steps: we can equivalently consider the virtual process where $\fR{\A}{\B\E}$ initializes the register $C_1$ with the independent uniform value $u_1$, exactly as $\fI{\B\E}$ did. (The distinguisher does not yet have access to $U_1$, the only register which differs between $\fR{\A}{\B\E}$ and $\fI{\B\E}$ under this virtual process.)
}

\item If $D=0$, the simulator receives the message $\mval$ at the inner interface and sets $U_1=\mval \oplus h(k^\A) \oplus u_1 $, then it outputs the register $R$ at the outer interfaces. 
\inset{
Through this process, the distinguisher only receives $D$ followed by $R$. Since it already knows $D$, the former is trivial, and we only need to bound the distinguishability after receiving $R$. 
At this point, the state produced by $\fR{\A}{\B\E}$ is such that $U_1$ was initialized with the independent uniform value $u_1$, and $C_1$ with the value $\mval \oplus h(k^\A) \oplus u_1$. In comparison, the state produced by $\fI{\B\E}$ is such that $C_1$ was initialized with the independent uniform value $u_1$ and $U_1$ was initialized with the value $\mval \oplus h(k^\A) \oplus u_1$.
Applying Fact~\ref{fc:OTP}, the situations for $\fR{\A}{\B\E}$ and $\fI{\B\E}$ are hence exactly equivalent.
}

\item If $D=1$, the simulator releases $y'_{\overline{T}}$ at the outer interfaces, then receives an input $b'_{\overline{T}}$. Using this value, it runs step~\ref{alg:B-test}, and feeds the output $F$ of that step to the ideal functionality. Depending on the value of $F$, it performs one of the following actions:
\begin{itemize}
\item If $F=\text{\xmark}$, the simulator does the same as in the $D=0$ case: it receives the message $\mval$ at the inner interface and sets $U_1 = \mval \oplus h(k^\A) \oplus u_1 $, then it outputs the register $R$ at the outer interfaces.
\item If $F=\text{\cmark}$, the simulator sets $U_1 = 0^n \oplus h(k^\A) \oplus u_1 $, then it outputs the register $R$ at the outer interfaces.
\end{itemize}
\inset{
Through this process, the distinguisher receives $(D,y'_{\overline{T}})$, supplies an input $b'_{\overline{T}}$, then receives $F$ followed by $R$.
By Fact~\ref{fc:OTP}, it is again easily seen that the states produced by $\fR{\A}{\B\E}$ and $\fI{\B\E}$ remain perfectly indistinguishable up until $R$ is released, because as long as the distinguisher does not have access to $R$ (and hence $U_1$), we can consider the virtual process where both $\fR{\A}{\B\E}$ and $\fI{\B\E}$ initialized $C_1$ with the independent uniform value $u_1$.

After $R$ is released, we note that the conditional states for the $O=\perp$ component are perfectly indistinguishable, because in that component 
all the registers are independent of the message (possibly by being set to ``blank'' values).
Also, the conditional states for $F=\text{\xmark}$ are perfectly indistinguishable, by the same argument as in the $D=0$ case above. 
As for the conditional states for $F=\text{\cmark}$, the states produced by $\fR{\A}{\B\E}$ and $\fI{\B\E}$ are $\sigma_{CR\tB B_T\tE}(\mval,1)$ and $\sigma_{CR\tB B_T\tE}(0^n,1)$ respectively --- the former holds by definition, while the latter can be understood by noting the simulator set the values $C_1=u_1$ and $U_1=0^n \oplus h(k^\A) \oplus u_1$, but by Fact~\ref{fc:OTP} the values on $C_1$ and $U_1$ can be swapped, which would then result in the state $\sigma_{CR\tB B_T\tE}(0^n,1)$. 

Overall, the distinguisher's states produced by $\fR{\A}{\B\E}$ and $\fI{\B\E}$ at this point are respectively of the form
\begin{align*}
&(1-p_\top) \state{\perp}_O \otimes \omega_{CR\tB B_T\tE} \\ 
&\qquad + p_\top \state{\top}_O \otimes \left((1-p_{\text{\cmark}|\top}) \state{\text{\xmark}}_F \otimes \psi_{CR\tB B_T\tE} + p_{\text{\cmark}|\top} \state{\text{\cmark}}_F \otimes \sigma_{CR\tB B_T\tE}(\mval,1)\right),\\
&(1-p_\top) \state{\perp}_O \otimes \omega_{CR\tB B_T\tE} \\
&\qquad + p_\top \state{\top}_O \otimes  \left((1-p_{\text{\cmark}|\top}) \state{\text{\xmark}}_F \otimes \psi_{CR\tB B_T\tE} + p_{\text{\cmark}|\top} \state{\text{\cmark}}_F \otimes \sigma_{CR\tB B_T\tE}(0^n,1)\right),
\end{align*}
where $\omega$ and $\psi$ are appropriate conditional states 
for $O=\perp$ and $F=\text{\xmark}$ 
(as argued above, these states are the same in the two scenarios), and $p_{\text{\cmark}|\top}$ is the same in both scenarios (by Lemma~\ref{lem:Hmin-sigma}). The only components that differ in the two scenarios are the $\sigma$ terms, hence by Lemma~\ref{lem:ct-indist} we see that the $\ell_1$ distance between the states is bounded by $2\lCI$.

}

\end{itemize}

\begin{remark}\label{remark:sim}
Informally, one piece of intuition captured by the above simulator construction is that we are showing there exists a process that can (given access to Bob and Eve's interfaces in $\fideal_{\B\E}$) closely reproduce the real behaviour $\fR{\A}{\B\E}$, \emph{without} having access to Alice's actual message $\mval$ until/unless the ideal resource $\fideal_{\B\E}$ reveals it in the final steps.
This in particular serves to formalize the notion that the real behaviour $\fR{\A}{\B\E}$ does not ``leak unwanted information about $\mval$'' (though we highlight that as discussed in Section~\ref{sec:operational}, satisfying the full Definition~\ref{def:comp} condition automatically guards against a much broader class of possible ``flaws''). This form of reasoning also lies behind other simulator-based arguments outside of Abstract Cryptography, for instance in the security analysis of protocols for zero-knowledge proofs.
\end{remark}

\subsection{Dishonest Bob and honest Eve}

Since Eve has no inputs to the protocol after the initial step, the argument for this case is essentially the same as the preceding section, just with $\tE$ traced out.

\subsection{Honest Bob and dishonest Eve}

\begin{itemize}

\item The simulator accepts the input states from the outer interface corresponding to Eve, and performs internal simulated instances of $\prot^{\A}$ and $\prot^{\B}$ from the ECD protocol until step~\ref{alg:E-test}. The inner interface of the simulator then feeds the output $O$ of that step as the value $O^\E$ to the ideal functionality, which releases the same value $O$ to the distinguisher (and also the simulator, though the simulator does not need it).
\inset{
$\fR{\A\B}{\E}$ and $\fI{\E}$ are perfectly indistinguishable throughout this process, since no message value has been chosen yet, and hence the states produced by $\fR{\A\B}{\E}$ and $\fI{\E}$ are identical.
}

\item  If $O=\perp$, the simulator stops here, apart from releasing its register $R$ at the end.
Otherwise, it continues on with 
its simulated instance of $\prot^{\A}$ from the ECD protocol, except that at the step where $C_1$ is to be generated, it instead prepares $C_1$ by generating an independent and uniformly random $u_1$ and setting $C_1 = u_1$. 
Furthermore, the simulator does not initialize a register $U_1$ yet. The simulator then proceeds until it receives $D$ from the ideal functionality at the inner interface.
\inset{
By Fact~\ref{fc:OTP}, it is easily seen that the states produced by $\fR{\A\B}{\E}$ and $\fI{\E}$ remain perfectly indistinguishable throughout these steps: we can equivalently consider the virtual process where $\fR{\A\B}{\E}$ initializes the register $C_1$ with the independent uniform value $u_1$, exactly as $\fI{\E}$ did. (The distinguisher does not yet have access to $U_1$, the only register which differs between $\fR{\A\B}{\E}$ and $\fI{\E}$ under this virtual process.)
}

\item If $D=0$, 
the simulator continues on with its simulated instances of $\prot^{\A}$ and $\prot^{\B}$ until those protocols are finished, upon which the simulator sets $U_1=0^n \oplus h(k^\A) \oplus u_1 $ and releases $R$ at the outer interface.
\inset{
Through this process, the distinguisher receives $D$ followed by $R\tM$. (There is no $F$ output for $D=0$.) 
The distinguisher supplies no inputs, so we can suppose without loss of generality that it applies no operations on its systems through this process, and we only need to bound the distinguishability after $R\tM$ is released.

We note that the conditional states for the $O=\perp$ component at this point are perfectly indistinguishable, because in that component 
all the registers are independent of the message (possibly by being set to ``blank'' values).
For the $O=\top$ component, the conditional states produced by $\fR{\A\B}{\E}$ and $\fI{\E}$ are $\hBrho_{\tM CRB_T\tE}(\mval,0)$ and $\hBrho_{\tM CRB_T\tE}(0^n,0)$ respectively, where the latter can be understood by again using Fact~\ref{fc:OTP} to swap the values on the $C_1$ and $U_1$ registers. 

Overall,
the distance between the states from $\fR{\A\B}{\E}$ and $\fI{\E}$ at this point is
\begin{align*}
&p_\top \norm{\hBrho_{\tM CRB_T\tE}(\mval,0) - \state{\mval}_{\tM} \otimes \hBrho_{CRB_T\tE}(0^n,0)}_1 \\
\leq& p_\top \norm{\hBrho_{\tM CRB_T\tE}(\mval,0) - \state{\mval}_{\tM} \otimes \hBrho_{CRB_T\tE}(\mval,0)}_1 \\
&\quad + p_\top \norm{\state{\mval}_{\tM} \otimes \hBrho_{CRB_T\tE}(\mval,0) - \state{\mval}_{\tM} \otimes \hBrho_{CRB_T\tE}(0^n,0)}_1.
\end{align*}
By Lemma~\ref{lem:ct-indist}, the second term is bounded by $2\lCI$. As for the first term, we have 
\begin{align*}
&p_\top\norm{\hBrho_{\tM CRB_T\tE}(\mval,0) - \state{\mval}_{\tM} \otimes \hBrho_{CRB_T\tE}(\mval,0)}_1\\
\leq& p_\top\sum_{\widetilde{m}} \Pr[\tM=\widetilde{m} | O=\top] \norm{\state{\widetilde{m}}_{\tM} \otimes \hBrho_{CRB_T\tE|\tM=\widetilde{m}}(\mval,0) - \state{\mval}_{\tM} \otimes \hBrho_{CRB_T\tE|\tM=\widetilde{m}}(\mval,0)}_1\\
\leq& p_\top\sum_{\widetilde{m}\neq \mval} 2\Pr[\tM=\widetilde{m} | O=\top] \leq {2\lEC},
\end{align*}
applying Lemma~\ref{lem:corr} in the last line. Adding the two bounds, we arrive at a final bound of $2\lCI + 2\lEC$.
}

\item If $D=1$, 
the simulator continues on with its simulated instances of $\prot^{\A}$ and $\prot^{\B}$: it releases $y'_{\overline{T}}$ and $b'_{\overline{T}}$, then sets $U_1=0^n \oplus h(k^\A) \oplus u_1 $ and finally releases $R$ at the outer interface.
\inset{
Through this process, the distinguisher receives $(D,y'_{\overline{T}})$ (at $t_3$), then $b'_{\overline{T}}$ followed by $F$ (at $t_4$), and finally $R\tM$ (at $t_5$).
The distinguisher supplies no inputs, so we can suppose without loss of generality that it applies no operations on its systems through this process, and we only need to bound the distinguishability after $R\tM$ is released. Also, $D$ is trivial since the distinguisher chose it, and so is $\tM$ since here it is always set to $0^n$, so we shall ignore these registers.

We note that the conditional states for the $O=\perp$ component at this point are perfectly indistinguishable, because in that component 
all the registers are independent of the message (possibly by being set to ``blank'' values).
Also, for the $F=\text{\cmark}$ component the conditional states produced by $\fR{\A}{\B\E}$ and $\fI{\B\E}$ are $\sigma_{RCY'_{\overline{T}} B'_{\overline{T}} B_T\tE}(\mval,1)$ and $\sigma_{RCY'_{\overline{T}} B'_{\overline{T}} B_T\tE}(0^n,1)$ respectively, where the latter can be understood by again using Fact~\ref{fc:OTP} to swap the values on the $C_1$ and $U_1$ registers. 

Overall, the states produced by $\fR{\A\B}{\E}$ and $\fI{\E}$ at this point are respectively of the form
\begin{align*}
&(1-p_\top) \state{\perp}_O \otimes \omega_{FRCY'_{\overline{T}} B'_{\overline{T}} B_T\tE} + p_\top \Frho_{FRCY'_{\overline{T}} B'_{\overline{T}} B_T\tE}(\mval,1),\\
&(1-p_\top) \state{\perp}_O \otimes \omega_{FRCY'_{\overline{T}} B'_{\overline{T}} B_T\tE} + p_\top \state{\text{\cmark}}_F \otimes \Frho_{RCY'_{\overline{T}} B'_{\overline{T}} B_T\tE}(0^n,1),
\end{align*}
where $\omega$ is an appropriate conditional state (as argued above, it is the same in both scenarios), and 
\begin{align*}
\Frho_{FRCY'_{\overline{T}} B'_{\overline{T}} B_T\tE}(\mval,1) & = (1-p_{\text{\cmark}|\top}) \state{\text{\xmark}}_F \otimes \psi_{RCY'_{\overline{T}} B'_{\overline{T}} B_T\tE}(\mval,1) \\
& \quad + p_{\text{\cmark}|\top} \state{\text{\cmark}}_F \otimes \sigma_{RCY'_{\overline{T}} B'_{\overline{T}} B_T\tE}(\mval,1),\\
\state{\text{\cmark}}_F \otimes \Frho_{RCY'_{\overline{T}} B'_{\overline{T}} B_T\tE}(0^n,1) & = \state{\text{\cmark}}_F \otimes \left( (1-p_{\text{\cmark}|\top}) \psi_{RCY'_{\overline{T}} B'_{\overline{T}} B_T\tE}(0^n,1) \right. \\
& \quad \left. + p_{\text{\cmark}|\top} \sigma_{RCY'_{\overline{T}} B'_{\overline{T}} B_T\tE}(0^n,1) \right).
\end{align*}
where $\psi$ are appropriate conditional states, and $p_{\text{\cmark}|\top}$ is the same in both scenarios (by Lemma~\ref{lem:Hmin-sigma}). The $\ell_1$ distance between the two expressions is bounded by
\begin{align*}
2p_\top (1-p_{\text{\cmark}|\top}) + p_\top p_{\text{\cmark}|\top} \norm{\sigma_{RCY'_{\overline{T}} B'_{\overline{T}} B_T\tE}(\mval,1) - \sigma_{RCY'_{\overline{T}} B'_{\overline{T}} B_T\tE}(0^n,1)}_1 \leq 2\lcom + 2\lCI,
\end{align*}
where we have applied Lemmas~\ref{lem:com} and~\ref{lem:ct-indist} (for the latter we use the fact that $\tB$ can contain a copy of $Y'_{\overline{T}} B'_{\overline{T}}$, and apply Fact~\ref{fc:chan-l1}).
}

\end{itemize}

\subsection{Honest Bob and Eve}

In this case there are no dishonest parties, so the simulator is trivial and our task is simply to bound the distinguishability between $\fR{\A\B\E}{}$ and $\fideal$.

\begin{itemize}

\item We first consider the situation up until $D$ is supplied.
\inset{
Through this process, the distinguisher releases $O$, then supplies $M$ and $D$. 
When $O$ is released, the states produced by $\fR{\A\B\E}{}$ and $\fideal$ are $(1-p_\top) \state{\perp}_O + p_\top \state{\top}_O$ and $\state{\top}_O$ respectively, where $p_\top$ is computed with respect to the honest behaviour in the ECD protocol. Then Lemma~\ref{lem:com} implies the $\ell_1$ distance between them is bounded by
\begin{align*}
2(1-p_\top) \leq 2\lcom.
\end{align*}

After that, $\fR{\A\B\E}{}$ and $\fideal$ do not release any outputs during the steps described here, hence the distance between the states cannot increase. 
}

\item If $D=0$:
\inset{
The distinguisher receives $D$ followed by $\tM$. (There is no $F$ output for $D=0$.) $D$ is trivial since the distinguisher chose it, so we only need to bound the distinguishability after $\tM$ is released. The states produced by $\fR{\A\B\E}{}$ and $\fideal$ at this point are respectively (filling in the register $\tM$ with a ``blank value'' $\phi$ in the case where $O=\perp$ for $\fR{\A\B\E}{}$):
\begin{gather*}
(1-p_\top) \state{\perp}_O \otimes \state{\phi}_{\tM} + p_\top \state{\top}_O \otimes \sum_{\widetilde{m}} \Pr[\tM=\widetilde{m} | O=\top] \state{\widetilde{m}}_{\tM}, \\
\state{\top}_O \otimes \state{\mval}_{\tM},
\end{gather*}
and the $\ell_1$ distance between them is upper bounded by
\begin{align*}
&2(1-p_\top) + p_\top \norm{\state{\top}_O \otimes \sum_{\widetilde{m}} \Pr[\tM=\widetilde{m} | O=\top] \state{\widetilde{m}}_{\tM} - \state{\top}_O \otimes \state{\mval}_{\tM}}_1 \\
\leq& 2(1-p_\top) + p_\top \sum_{\widetilde{m}\neq\mval} 2 \Pr[\tM=\widetilde{m} | O=\top] \leq 2\lcom + 2\lEC,
\end{align*}
applying Lemmas~\ref{lem:com} and \ref{lem:corr} in the last line.
}

\item If $D=1$:
\inset{
The distinguisher receives $D$, followed by $F$ and $\tM$. $D$ is trivial, and no inputs occur between $F$ and $\tM$, so we only need to bound the distinguishability after $\tM$ is released. The states produced by $\fR{\A\B\E}{}$ and $\fideal$ at this point are respectively (filling in the registers $\tM F$ with a ``blank value'' $\phi$ in the case where $O=\perp$ for $\fR{\A\B\E}{}$):
\begin{gather*}
(1-p_\top) \state{\perp}_O \otimes \state{\phi}_F \otimes \state{\phi}_{\tM} + p_\top \state{\top}_O \otimes \Frho_F \otimes \state{0^n}_{\tM}, \\
\state{\top}_O \otimes \state{\text{\cmark}}_F \otimes \state{0^n}_{\tM},
\end{gather*}
where $\Frho_F=(1-p_{\text{\cmark}|\top})\state{\text{\xmark}}_F + p_{\text{\cmark}|\top} \state{\text{\cmark}}_F$. The $\ell_1$ distance between them is upper bounded by 
\begin{align*}
2(1-p_\top) + 2 p_\top (1-p_{\text{\cmark}|\top}) \leq 4\lcom,
\end{align*}
applying Lemma~\ref{lem:com}.
}

\end{itemize}

\section{Parallel repetition theorems}\label{sect:parrep}
\subsection{Parallel repetition theorem for 2-round 2-player product-anchored game}
\begin{definition}
A 2-round 2-player non-local game is called a \emph{product-anchored game} with \emph{anchoring probability} $\alpha$ iff
\begin{itemize}
\item Alice and Bob get $(x,y) \in \clX\times\clY$ from a product distribution as their first round inputs.
\item Alice and Bob produce $(a,b) \in \clA\times\clB$ as their first round outputs.
\item Bob gets $z=\perp$ with probability $\alpha$ and $z=(x,y')$ with probability $1-\alpha$, as his second round input, such that the distribution of $(x,y)$ conditioned on $z=\perp$ is the same as the marginal distribution of $(x,y)$. (Alice has no input.)
\item Bob produces $b'$ as his second round output. (Alice has no output.)
\item Alice and Bob win the game iff $\sfV(x,y,a,b)$ and $\sfV'(x,y,z,a,b,b')$ are both satisfied.
\end{itemize}
\end{definition}

\begin{theorem}\label{thm:parrep}
Let $G$ be a 2-round 2-player non-local product-anchored game satisfying the conditions above with parameter $\alpha$. Then for $\delta >0$ and $t= (\omega^*(G)+\eta)l$,
\begin{align*}
\omega^*(G^l) & = \left(1-(1-\omega^*(G))^3\right)^{\Omega\left(\frac{\alpha^2l}{\log(|\clA|\cdot|\clB|\cdot|\clB'|)}\right)} \\
\omega^*(G^{t/l}) & = \left(1-\eta^3\right)^{\Omega\left(\frac{\alpha^2l}{\log(|\clA|\cdot|\clB|\cdot|\clB'|)}\right)}.
\end{align*}
\end{theorem}

We shall use the following results in order to prove the theorem.
\begin{fact}[\cite{Hol09}]\label{fc:hol-cond}
Let $\sfP_{TU_1\ldots U_lV} = \sfP_T\sfP_{U_1|T}\sfP_{U_2|T}\ldots\sfP_{U_l|T}\sfP_{V|TU_1\ldots U_l}$ be a probability distribution over $\clT\times\clU^l\times\clV$, and let $\clE$ be any event. Then,
\[ \sum_{i=1}^l\Vert\sfP_{TU_iV|\clE}-\sfP_{TV|\clE}\sfP_{U_i|T}\Vert_1 \leq \sqrt{l\left(\log(|\clV|) + \log\left(\frac{1}{\Pr[\clE]}\right)\right)}.\] \end{fact}
\begin{fact}[\cite{BVY15}, Lemma 16]\label{fc:anchor-dist}
Suppose $TVW$ are random variables 
such that for some $w^*$, 
we have $\sfP_{VW}(v,w^*)=\alpha\cdot\sfP_V(v)$ for all $v$. Then,
\[ \left\Vert \sfP_{TVW} - \sfP_{VW}\sfP_{T|V,w^*}\right\Vert_1 \leq \frac{2}{\alpha}\left\Vert \sfP_{TVW}-\sfP_{VW}\sfP_{T|V}\right\Vert_1.\]
\end{fact}
Using the above fact, we prove the following lemma that we shall use later.
\begin{lemma}\label{lem:anchor-t*}
Suppose $\sfP_{ST}$ and $\sfP_{S'T'R'}$ are distributions such that 
for some $t^*$, we have for some $t^*$, $\sfP_{ST}(s,t^*) = \alpha\cdot\sfP_S(s)$ for all $s$. If $\norm{\sfP_{ST}-\sfP_{S'T'}}_1 \leq \alpha$, then,
\begin{enumerate}[(i)]
\item $ \Vert\sfP_{S'R'|t^*} - \sfP_{S'R'}\Vert_1 \leq \dfrac{2}{\alpha}\Vert\sfP_{S'T'R'} - \sfP_{S'R'}\sfP_{T|S}\Vert_1 + \dfrac{5}{\alpha}\norm{\sfP_{S'T'} - \sfP_{ST}}_1$;
\vspace{-1cm}
\item $\displaystyle \begin{aligned}
 & \\[0.4cm]
\norm{\sfP_{S'T'R'} - \sfP_{ST}\sfP_{R'|t^*}}_1 & \leq \frac{2}{\alpha}\Big(\norm{\sfP_{S'T'R'} - \sfP_{T'R'}\sfP_{S|T}}_1 + \norm{\sfP_{S'T'R'} - \sfP_{S'R'}\sfP_{T|S}}_1\Big)  \\
& \quad + \frac{7}{\alpha}\norm{\sfP_{S'T'} - \sfP_{ST}}_1.
\end{aligned}$
\end{enumerate}
\end{lemma}
\begin{proof}
Note that
\[
\Vert \sfP_{S|t^*} - \sfP_{S'|t^*}\Vert_1 \leq \frac{2}{\alpha}\norm{\sfP_{S'T'}-\sfP_{ST}}_1
\]
by Fact \ref{fc:cond-prob}. Let $\sfP_{STR''}$ denote the distribution $\sfP_{ST}\sfP_{R'|S'T'}$, i.e., $\sfP_{STR''}(s,t,r) = \sfP_{ST}(st)\sfP_{R'|S'=s,T'=t}(r)$.
\begin{align*}
\Vert \sfP_{S'R'} - \sfP_{SR''}\Vert_1 & = \sum_{s,r}\left| \sfP_{S'}(s)\sum_t\sfP_{T'|s}(t)\sfP_{R'|st}(r) - \sfP_S(s)\sum_t\sfP_{T|s}(t)\sfP_{R'|st}(r)\right| \\
 & \leq \sum_{s,t,r}\left|\sfP_{S'}(s)\sfP_{T'|s}(t) - \sfP_S(s)\sfP_{T|s}(t)\right|\sfP_{R'|st}(r) \\
 & = \Vert \sfP_{S'T'} - \sfP_{ST}\Vert_1.
\end{align*}
Similarly,
\begin{align*}
\norm{\sfP_{SR''} - \sfP_{SR''|t^*}}_1 & \leq \left\Vert \sfP_{ST}\sfP_{R'|S'T'} - \sfP_{ST}\sfP_{R'|S',t^*}\right\Vert_1 \\
 & \leq \frac{2}{\alpha}\left\Vert \sfP_{ST}\sfP_{R'|S'T'}-\sfP_{ST}\sfP_{R'|S'}\right\Vert_1 \\
 & \leq \frac{2}{\alpha}\left(\left\Vert \sfP_{ST}\sfP_{R'|S'T'}-\sfP_{T|S}\sfP_{S'R'}\right\Vert_1 + \norm{\sfP_{S'T'} - \sfP_{ST}}_1\right)
\end{align*}
where in the second inequality we have used Fact \ref{fc:anchor-dist}.
Combining all this, and using Fact \ref{fc:cond-prob},
\begin{align*}
\Vert \sfP_{S'R'|t^*} - \sfP_{S'R'}\Vert_1 & \leq \Vert \sfP_{S'R'|t^*} - \sfP_{SR''|t^*}\Vert_1 + \Vert \sfP_{SR''|t^*} - \sfP_{SR''}\Vert_1 + \Vert \sfP_{SR''} - \sfP_{S'R'}\Vert_1 \\
 & \leq \frac{2}{\alpha}\norm{\sfP_{S'T'R'}-\sfP_{STR''}}_1 + \frac{2}{\alpha}\left(\left\Vert \sfP_{ST}\sfP_{R'|S'T'}-\sfP_{T|S}\sfP_{S'R'}\right\Vert_1 + \norm{\sfP_{S'T'} - \sfP_{ST}}_1\right) \\
 & \quad + \norm{\sfP_{S'T'}-\sfP_{ST}}_1 \\
 & = \frac{2}{\alpha}\norm{(\sfP_{S'T'}-\sfP_{ST})\sfP_{R'|S'T'}}_1 + \frac{2}{\alpha}\Vert\sfP_{S'T'R'} - \sfP_{S'R'}\sfP_{T|S}\Vert_1 + \frac{3}{\alpha}\norm{\sfP_{S'T'} - \sfP_{ST}}_1 \\
 & = \frac{2}{\alpha}\norm{\sfP_{S'T'}-\sfP_{ST}}_1 + \frac{2}{\alpha}\Vert\sfP_{S'T'R'} - \sfP_{S'R'}\sfP_{T|S}\Vert_1 + \frac{3}{\alpha}\norm{\sfP_{S'T'} - \sfP_{ST}}_1 \\
 & = \frac{2}{\alpha}\Vert\sfP_{S'T'R'} - \sfP_{S'R'}\sfP_{T|S}\Vert_1 + \frac{5}{\alpha}\norm{\sfP_{S'T'} - \sfP_{ST}}_1.
\end{align*}
This proves item (i).

We have $\sfP_{T'}(t^*) \geq \sfP_T(t^*) - \frac{1}{2}\norm{\sfP_{S'T'}-\sfP_{ST}}_1 \geq \alpha/2$. Therefore we have,
\[
\norm{\sfP_{S'T'R'} - \sfP_{S|T}\sfP_{T'R'}}_1 \geq \sfP_{T'}(t^*)\norm{\sfP_{S'R'|t^*} - \sfP_{S|t^*}\sfP_{R'|t^*}}_1 \geq \frac{\alpha}{2}\norm{\sfP_{S'R'|t^*} - \sfP_{S|t^*}\sfP_{R'|t^*}}_1.
\]
Using this we get,
\begin{align}
\norm{\sfP_S\sfP_{R'|S',t^*} - \sfP_S\sfP_{R'|t^*}}_1 & \leq \norm{\sfP_{S'R'|t^*} - \sfP_{S|t^*}\sfP_{R'|t^*}}_1 + \norm{(\sfP_{S'|t^*} - \sfP_{S|t^*})\sfP_{R'|S',t^*}}_1 \nonumber \\
 & \leq \frac{2}{\alpha}\norm{\sfP_{S'T'R'} - \sfP_{S|T}\sfP_{T'R'}}_1 + \norm{\sfP_{S'|t^*} - \sfP_{S|t^*}}_1 \nonumber \\
 & \leq \frac{2}{\alpha}\norm{\sfP_{S'T'R'} - \sfP_{S|T}\sfP_{T'R'}}_1 + \frac{2}{\alpha}\norm{\sfP_{S'T'} - \sfP_{ST}}_1, \label{eq:anchor-SR}
\end{align}
where we have used Fact \ref{fc:cond-prob} in the last inequality.

Next, note that we can apply Fact \ref{fc:anchor-dist} to get a bound 
\[
\norm{\sfP_{STR''} - \sfP_{ST}\sfP_{R''|S,t^*}}_1 \leq \frac{2}{\alpha}\norm{\sfP_{STR''} - \sfP_{ST}\sfP_{R''|S}}_1,
\]
because the marginal distribution on the first two variables is $\sfP_{ST}$ in both terms, which satisfies $\sfP_{ST}(s,t^*) = \sfP_{T}(t^*)\sfP_S(s) = \alpha\cdot\sfP_S(s)$ for all $s$. But since by definition we have $\sfP_{STR''} = \sfP_{ST}\sfP_{R'|S'T'}$, we can rewrite the bound as follows:
\[
\norm{\sfP_{ST}\sfP_{R'|S'T'} - \sfP_{ST}\sfP_{R'|S',t^*}}_1 \leq \frac{2}{\alpha}\norm{\sfP_{ST}\sfP_{R'|S'T'}-\sfP_{ST}\sfP_{R'|S'}}_1 .
\]

Using this we get,
\begin{align}
\norm{\sfP_{S'T'R'} - \sfP_{ST}\sfP_{R'|S',t^*}}_1 & \leq \norm{\sfP_{ST}\sfP_{R'|S'T'} - \sfP_{ST}\sfP_{R'|S',t^*}}_1 + \norm{\sfP_{S'T'R'} - \sfP_{ST}\sfP_{R'|S'T'}}_1 \nonumber \\
 & \leq \frac{2}{\alpha}\norm{\sfP_{ST}\sfP_{R'|S'T'}-\sfP_{ST}\sfP_{R'|S'}}_1 + \norm{\sfP_{S'T'} - \sfP_{ST}}_1 \nonumber \\
 & \leq \frac{2}{\alpha}\Big(\norm{\sfP_{S'T'R'} - \sfP_{S'R'}\sfP_{T|S}}_1 + \norm{\sfP_{ST} - \sfP_{S'T'}}_1 + \norm{\sfP_{S'} - \sfP_S}_1\Big) \nonumber \\
 & \quad + \norm{\sfP_{S'T'} - \sfP_{ST}}_1 \nonumber \\
 & \leq \frac{2}{\alpha}\norm{\sfP_{S'T'R'} - \sfP_{S'R'}\sfP_{T|S}}_1 + \frac{5}{\alpha}\norm{\sfP_{S'T'} - \sfP_{ST}}_1. \label{eq:anchor-Rt*}
\end{align}
Therefore, using \eqref{eq:anchor-SR} and \eqref{eq:anchor-Rt*},
\begin{align*}
\norm{\sfP_{S'T'R'} - \sfP_{ST}\sfP_{R'|t^*}}_1 & \leq \norm{\sfP_{S'T'R'} - \sfP_{ST}\sfP_{R'|S',t^*}}_1 + \norm{\sfP_S\sfP_{R'|S',t^*} - \sfP_S\sfP_{R'|t^*}}_1 \\
 & \leq \frac{2}{\alpha}\Big(\norm{\sfP_{S'T'R'} - \sfP_{T'R'}\sfP_{S|T}}_1 + \norm{\sfP_{S'T'R'} - \sfP_{S'R'}\sfP_{T|S}}_1\Big)  \\
& \quad + \frac{7}{\alpha}\norm{\sfP_{S'T'} - \sfP_{ST}}_1.
\end{align*}
This proves item (ii).
\end{proof}

Finally, we shall use the following two facts.
\begin{fact}[\cite{JPY14}, Lemma III.1]\label{fc:jpy-cond}
Suppose $\rho$ and $\sigma$ are CQ states satisfying $\rho = \delta\sigma + (1-\delta)\sigma'$ for some other state $\sigma'$. Suppose $Z$ is a classical register of size $|\clZ|$ in $\rho$ and $\sigma$ such that the distribution on $Z$ in $\sigma$ is $\sfP_Z$, then
\[ \bbE_{\sfP_Z}\sfD(\sigma_z\Vert\rho) \leq \log(1/\delta) + \log|\clZ|.\]
\end{fact}
\begin{fact}[Quantum Raz's Lemma, \cite{BVY15}]\label{fc:qraz}
Let $\rho_{XY}$ and $\sigma_{XY}$ be two CQ states with $X = X_1\ldots X_l$ being classical, and $\sigma$ being product across all registers. Then,
\[ \sum_{i=1}^l\sfI(X_i:Y)_\rho \leq \sfD(\rho_{XY}\Vert \sigma_{XY}).\]
\end{fact}

\begin{proof}[Proof of Theorem~\ref{thm:parrep}]
Consider a strategy $\clS$ for $l$ copies of $G$ (it may correspond to $G^l$ or $G^{t/l}$ --- it doesn't really matter): before the game starts, Alice and Bob share an entangled state on registers $A\tA B\tB B'\tB' E^\A E^\B$. Here $A, B, B'$ will be the registers in which the outputs are measured in the computational basis, and $\tA, \tB, \tB'$ are registers onto which the contents of $A,B, B'$ are copied --- we can always assume the outputs are copied since they are classical. Alice and Bob apply unitaries based on their first round inputs $XY$ to their respective halves of this entangled state and measure in the computational basis to obtain their first round outputs. We define the following pure state to represent the inputs, outputs and other registers in the protocol at this stage:
\begin{align*}
& \ket{\rho}_{X\tX Y\tY Z\tZ A\tA B\tB B'\tB'E^\A E^\B} \\
= & \sum_{x,y,z}\sqrt{\sfP_{XYZ}(x,y,z)}\ket{xx}_{X\tX}\ket{yy}_{Y\tY}\ket{zz}_{Z\tZ}\sum_{a,b}\sqrt{\sfP_{AB|xy}(ab)}\ket{aa}_{A\tA}\ket{bb}_{B\tB}\ket{\rho}_{B'\tB'E^\A E^\B|xyab}
\end{align*}
where we have used $Z$ to denote Bob's second round input, which is either $\perp$ or $(x,y')$. We have included the $Z\tZ$ registers in this state even though Bob has not received the $z$ input yet; the state in the entangled registers has no dependence on $z$ however. Here $\sfP_{AB|xy}(a,b)$ is the probability of Alice and Bob obtaining outputs $(a,b)$ on inputs $(x,y)$ in the first round.

In the actual protocol, the $AB$ registers are measured on $\ket{\rho}$, and the subsequent unitary Bob applies on the $B'\tB'E^\B$ registers can depend on his first round output, as well as both his inputs. We represent the state of the protocol at this state by:
\begin{align*}
& \ket{\sigma}_{X\tX Y\tY Z\tZ A\tA B\tB B'\tB'E^\A E^\B} \\
= & \sum_{x,y,z}\sqrt{\sfP_{XYZ}(x,y,z)}\ket{xx}_{X\tX}\ket{yy}_{Y\tY}\ket{zz}_{Z\tZ}\sum_{a,b}\sqrt{\sfP_{AB|xy}(ab)}\ket{aa}_{A\tA}\ket{bb}_{B\tB}\otimes \\
& \quad\quad \sum_{b'}\sqrt{\sfP_{B'|xyzab}(b')}\ket{b'b'}_{B'\tB'}\ket{\sigma}_{E^\A E^\B|xyzabb'}.
\end{align*}
Note that $\ket{\sigma}$ is related to $\ket{\rho}$ by a unitary on the $B'\tB'E^\B$ registers that is controlled on the registers $YZB$, which is why the marginal distribution of $AB$ is the same in $\ket{\rho}$ and $\ket{\sigma}$. Note that even though no operations explicitly dependent on $x$ are done in the second round, the distribution of $B'$ obtained in $\sigma$ depends on $x$, because the state $\ket{\rho}_{B'\tB'E^\A E^\B|xyab}$ depended on $x$.

Let $\omega^*(G) = 1 - \eps$. To prove the theorem, we shall use the following lemma (whose proof is given later).
\begin{lemma}\label{lem:parrep3-ind}
For $i\in[l]$, let $T_i=\sfV(A_iB_i,X_iY_i)\land\sfV'(A_iB_iB'_i,X_iY_iZ_i)$ in a strategy $\clS$ for $l$ copies of $G$ (here $X,Y$ are the first round inputs, $A,B$ first round outputs, and $\sfV$ the first round predicate; $Z$ is the second round input, $B'$ the second round output, and $\sfV'$ the second round predicate). If $\clE_C$ is the event $\prod_{i\in C}T_i=1$ for $C\subseteq [l]$, then
\[ \bbE_{i\in\oC}\Pr[T_i=1|\clE_C] \leq 1-\eps + \frac{191\sqrt{\delta_C}}{\alpha} \]
where
\[ \delta_C = \frac{|C|\cdot\log(|\clA|\cdot|\clB|\cdot|\clB'|) + \log(1/\Pr[\clE_C])}{l}.\] 
\end{lemma}

The theorem follows from the above lemma using standard arguments as in e.g.~\cite{Rao10} --- we shall reproduce here the proof given in that work, with some additional elaborations. We only need to prove the upper bound on $\omega^*(G^{t/l})$; the one for $\omega^*(G^l)$ then follows immediately by setting $\eta = 1-\omega^*(G^l)$. 

Consider any strategy for playing $l$ parallel instances of $G$. We begin by proving the following claim: for any $\gamma\in(0,1)$, if we set $n$ to be a value satisfying
\begin{equation}\label{eq:ndefncond}
2^{-\gamma^2l+n\cdot\log(|\clA|\cdot|\clB|\cdot|\clB'|)} = \left(1-\eps+\frac{191\gamma}{\alpha}\right)^n
\end{equation}
and pick a uniformly random subset $C \subseteq [l]$ of size $n$, the expected probability of winning all the instances on the chosen subset satisfies
\begin{equation}\label{eq:thres-bound}
\bbE_{C} \Pr[\clE_C] = \sum_{C\subseteq [l]:|C|=n}\frac{1}{{l \choose n}}\Pr[\clE_C] \leq 2\left(1-\eps+\frac{191\gamma}{\alpha}\right)^n.
\end{equation}
Note that there is indeed a (unique) value of $n$ satisfying the required condition, specifically
\[ n = \frac{\gamma^2l}{\log\left(\frac{|\clA|\cdot|\clB|\cdot|\clB'|}{1-\eps+\frac{191\gamma}{\alpha}}\right)}. \]

\newcommand{\rvI}{\mathbf{I}}
\newcommand{\rvC}{\mathbf{C}}
\newcommand{\rvV}{\mathbf{V}}
\newcommand{\evW}{\mathcal{W}}
\newcommand{\evL}{\mathcal{L}}
\newcommand{\evH}{\mathcal{H}}

To prove this claim, let $\rvI_1, \rvI_2, \dots \rvI_n$ be a sequence of random variables obtained by drawing elements from $[l]$ uniformly at random without replacement (in this argument, random variables will be denoted in bold, while specific values they can take will be denoted without bolding). 
For each $j\in[n]$, we define a random variable $\rvC_j = (\rvI_1, \rvI_2, \dots \rvI_j)$, i.e.~its value is the tuple formed by the first $j$ elements of the sequence. 
In a minor abuse of terminology, we will sometimes call the tuple $\rvC_j$ a subset of $[l]$, and we define $\evW_j$ to be the event that all instances in $\rvC_j$ are won.
With this interpretation, observe that $\rvC_n$ is a uniformly random subset of size $n$, and hence our goal is just to prove a bound on $\Pr[\evW_n]$ (since this would be equal to $\bbE_{C} \Pr[\clE_C]$ in~\eqref{eq:thres-bound}).
Next, for each $j\in[n]$ we define $\evL_j$ to be an event on $\rvC_j$, as follows: it is the event that $\rvC_j$ takes a value $C_j$ such that $\Pr[\clE_{C_j}] \leq \left(1-\eps+{191\gamma}/{\alpha}\right)^n$, i.e.~the probability of winning all games on that particular subset $C_j$ is ``low''.\footnote{A slightly different perspective that may help in understanding the definition of this event is that for each $j\in[n]$, we can define an indicator variable $\rvV^{\evL}_j$ as follows: $\rvV^{\evL}_j$ is a function of $\rvC_j$, taking value $1$ if $\rvC_j$ takes a value $C_j$ such that $\Pr[\clE_{C_j}] \leq \left(1-\eps+{191\gamma}/{\alpha}\right)^n$, and taking value $0$ otherwise. (Note that this is a well-defined function of $C_j$ because we have fixed a particular strategy for playing the game, and hence $\Pr[\clE_{C_j}]$ is a function of $C_j$ only.) The event $\evL_j$ is then exactly the event $\rvV^{\evL}_j=1$.}
Also, let $\evH_j$ be the complementary event that $\rvC_j$ takes a value $C_j$ such that $\Pr[\clE_{C_j}] > \left(1-\eps+{191\gamma}/{\alpha}\right)^n$, i.e.~the winning probability is ``high''.

With these events, we can write $\Pr[\evW_n] = \Pr[\evW_n \land \evL_n] + \Pr[\evW_n \land \evH_n]$ and bound the individual terms. The first term is simply upper bounded by $\left(1-\eps+{191\gamma}/{\alpha}\right)^n$ due to the definition of $\evL_n$. To bound the second term, we argue as follows for each $j\in[n]$. Observe that if we consider any fixed value $C_j$ for the random variable $\rvC_j$, this also fixes a value $C_{j-1}$ for the ``preceding'' tuple $\rvC_{j-1}$, and furthermore, winning all instances of the game on $C_j$ implies winning all instances on $C_{j-1}$.
With this, the event $\evW_j$ always implies the event $\evW_{j-1}$, and so we can write $\evW_j = \evW_j \land \evW_{j-1}$. 
It also tells us that for each value $C_j$ we have $\Pr[\clE_{C_j}] \leq \Pr[\clE_{C_{j-1}}]$, from which we see that the event $\clH_{j}$ always implies the event $\clH_{j-1}$, and we can write $\evH_j = \evH_j \land \evH_{j-1}$. Hence for each $j\in[n]$ we have
\[
\Pr[\evW_j \land \evH_j] = \Pr[\evW_j \land \evW_{j-1} \land \evH_j \land \evH_{j-1}] = \Pr[\evW_j \land \evH_j | \evW_{j-1} \land \evH_{j-1}] \Pr[\evW_{j-1} \land \evH_{j-1}],
\]
where for ease of notation we introduce trivial ``always-true'' events $\evW_{0}$ and $\evH_{0}$. 
Applying this relation repeatedly, we arrive at
\[
\Pr[\evW_n \land \evH_n] = \prod_{j=1}^n \Pr[\evW_j \land \evH_j | \evW_{j-1} \land \evH_{j-1}] \Pr[\evW_{j-1} \land \evH_{j-1}].
\]

Hence it suffices to bound $\Pr[\evW_j \land \evH_j | \evW_{j-1} \land \evH_{j-1}]$ for each $j$, which we shall do by upper bounding it with $\Pr[\evW_j | \evW_{j-1} \land \evH_{j-1}]$ and applying Lemma~\ref{lem:parrep3-ind}. In more detail: note that conditioned on $\evH_{j-1}$, by definition $\rvC_{j-1}$ takes a value $C_{j-1}$ satisfying $\Pr[\clE_{C_{j-1}}] > \left(1-\eps+{191\gamma}/{\alpha}\right)^n$, and this is the same as $\log(1/\Pr[\clE_{C_{j-1}}]) < \gamma^2l-n\cdot\log(|\clA|\cdot|\clB|\cdot|\clB'|)$ by the defining property~\eqref{eq:ndefncond} for $n$. Substituting this into the bound in Lemma~\ref{lem:parrep3-ind} gives (since $|C_{j-1}|\leq n$)
\[
1-\eps + \frac{191\sqrt{\delta_{C_{j-1}}}}{\alpha} \leq
1-\eps + \frac{191}{\alpha}\sqrt{\frac{n \log(|\clA|\cdot|\clB|\cdot|\clB'|) + \log(1/\Pr[\clE_{C_{j-1}}])}{l}} \leq
1-\eps + \frac{191\gamma}{\alpha}.
\]
With this, we can apply Lemma~\ref{lem:parrep3-ind} to write $\Pr[\evW_j | \evW_{j-1} \land \evH_{j-1}] \leq
1-\eps + {191\gamma}/{\alpha}$  (recalling that $\rvC_{j}$ can be viewed as being generated from $\rvC_{j-1}$ by drawing a uniformly random $i\notin \rvC_{j-1}$ and appending it).
Since this bound holds for every $j\in[n]$, we get $\Pr[\evW_n \land \evH_n] \leq \left(1-\eps+{191\gamma}/{\alpha}\right)^n$, which yields the claimed bound~\eqref{eq:thres-bound}.

Finally, to use this to prove the bound on $\omega^*(G^{t/l})$ for $t=(1-\eps+\eta)l$, we set $\gamma=\frac{\alpha\eta}{764}$, which gives $n< (1-\eps+\eta)l$. If $t=(1-\eps+\eta)l$ games are won, we can pick a random subset of size $n$ out of the $(1-\eps+\eta)l$ won games, and say that the probability of winning $(1-\eps+\eta)l$ games is upper bounded by the probability of winning on this random subset. Therefore we have,
\begin{equation}\label{eq:thres-bound2}
\omega^*(G^{t/l}) \leq \sum_{C\subseteq [(1-\eps+\eta)l]:|C|=n}\frac{1}{{(1-\eps+\eta)l \choose n}}\Pr[\clE_C] \leq 2\left(1-\eps+\frac{\eta}{4}\right)^n\cdot\frac{{l \choose n}}{{(1-\eps +\eta)l \choose n}},
\end{equation}
where the last inequality follows from the bound~\eqref{eq:thres-bound} (note that the binomial-coefficient terms are independent of the summations and could hence be factored out). We can simplify the second factor in the above expression as
\[ \frac{{l \choose n}}{{(1-\eps +\eta)l \choose n}} \leq \left(\frac{l}{(1-\eps+\eta)l-n}\right)^n \leq \left(\frac{1}{1-\eps+\frac{3\eta}{4}}\right)^n,\]
where in the last inequality we have used the expression for $n$. Putting this into \eqref{eq:thres-bound2} we get,
\[ \omega^*(G^{t/l}) \leq 2\left(\frac{1-\eps+\frac{\eta}{4}}{1-\eps+\frac{3\eta}{4}}\right)^n = 2\left(1-\frac{\frac{\eta}{2}}{1-\eps+\frac{3\eta}{4}}\right)^n \leq 2\left(1-\frac{\eta}{2}\right)^n,\]
which proves the theorem after substituting the value of $n$.
\end{proof}

To prove Lemma~\ref{lem:parrep3-ind}, we shall first define the correlation-breaking random variables $D_iG_i$ for each $i\in[l]$ as follows: $D_i$ is a uniformly random bit, and $G_i$ takes value $X_iY_i$ or $Z_i$ respectively depending on whether $D_i$ is 0 or 1. With this, $XYZ$ are independent conditioned on $DG$: to see this, first note that the distribution on $XYZDG$ is independent across instances, so it suffices to prove that for each $i$, $X_iY_iZ_i$ are independent conditioned on $D_i G_i$.
Observe that conditioned on $D_i=0$, the value of $X_iY_i$ is fixed by $G_i$, so $X_iY_i$ are trivially independent conditioned on $G_i$. Whereas conditioned on $D_i=1$, the value of $Z_i$ is fixed by $G_i$, so it suffices to show that $X_iY_i$ is independent conditioned on $Z_i$. This is indeed true because conditioned on $Z_i = \perp$, the distribution of $X_iY_i$ is equal to their marginal distribution, which is product; whereas conditioned on any other value of $Z_i$, the value of $X_i$ is fixed by $Z_i$, so $X_i$ is trivially independent of $Y_i$.

Let $\ket{\rho}_{dg}$ and $\ket{\sigma}_{dg}$ denote the states $\ket{\rho}$ and $\ket{\sigma}$ conditioned on $DG=dg$, which simply means that the distribution of $XYZ$ used is conditioned on $dg$.

Conditioned on $DG=dg$, we define the state $\ket{\vph}_{dg}$, which is $\ket{\sigma}_{dg}$ conditioned on success in $C$, i.e., the event $\clE_C$ as defined in Lemma~\ref{lem:parrep3-ind}:
\begin{align*}
& \ket{\vph}_{X\tX Y\tY Z\tZ A\tA B\tB B'\tB'E^\A E^\B|dg} \\
= & \frac{1}{\sqrt{\gamma_{dg}}}\sum_{x,y,z}\sqrt{\sfP_{XYZ|dg}(xyz)}\ket{xx}_{X\tX}\ket{yy}_{Y\tY}\ket{zz}_{Z\tZ}\sum_{\substack{a,b,b': \\ (x_C,y_C,z_C,a_C,b_C,b'_C) \\ \text{win } G^{|C|}}}\sqrt{\sfP_{ABB'|xyz}(ab)}\ket{aa}_{A\tA}\ket{bb}_{B\tB}\otimes \\
& \quad\quad \ket{b'b'}_{B'\tB'}\ket{\sigma}_{E^\A E^\B|xyzabb'}
\end{align*}
where $\gamma_{dg}$ is the probability of winning in $C$ conditioned on $DG=dg$, in $\clS$; $\gamma_{dg}$ averaged over $dg$ is then $\Pr[\clE_C]$. It is easy to see that $\sfP_{XYZABB'|\clE_C,dg}$ is the distribution on the registers $XYZABB'$ in $\ket{\vph}_{dg}$, and $\bbE_{\sfP_{DG|\clE_C}}\sfP_{XYZABB'|\clE_C,dg}$ is $\sfP_{XYZABB'|\clE_C}$.

In the remainder of the proof, we shall use the following notation.
For $i\in [l]$, we shall use $\ket{\vph}_{x_iy_iz_id_{-i}g_{-i}}$ (where $d_{-i}$ stands for $d_1\ldots d_{i-1}d_{i+1}\ldots d_l$ and $g_{-i}$ is defined similarly), $\ket{\vph}_{x_id_{-i}g_{-i}}$, $\ket{\vph}_{y_id_{-i}g_{-i}}$, $\ket{\vph}_{x_iy_id_{-i}g_{-i}}$ to refer to $\ket{\vph}$ with values of $X_iY_iZ_iD_{-i}G_{-i}$, $X_iD_{-i}G_{-i}$, $Y_iD_{-i}G_{-i}$ and $X_iY_iD_{-i}G_{-i}$ respectively conditioned on. We shall use similar notation for other variables and subsets of $[l]$ as well. $\ket{\vph}_{\perp,d_{-i}g_{-i}}$ will be used to refer to $\ket{\vph}$ conditioned on $Z_i=\perp,D_{-i}G_{-i}=d_{-i}g_{-i}$.

We shall use the following lemma, whose proof we give later, to prove Lemma~\ref{lem:parrep3-ind}. In the statement of this lemma, and in the proofs henceforth, we shall refer to $\delta_C$ and $\clE_C$ as just $\delta$ and $\clE$ for brevity.
\begin{lemma}\label{lem:parrep3-conds}
Let $\delta$ be $\delta_C$ as defined in Lemma~\ref{lem:parrep3-ind}. If $\delta \leq \frac{\alpha^2}{8}$, then using $R_i = X_CY_CZ_CA_CB_CB'_C D_{-i}G_{-i}$, the following conditions hold:
\begin{enumerate}[(i)]
\item $\bbE_{i\in\oC}\norm{\sfP_{X_iY_iZ_iR_i|\clE} - \sfP_{X_iY_iZ_i}\sfP_{R_i|\clE,\perp}}_1 \leq \dfrac{11\sqrt{2\delta}}{\alpha}$;
\item For each $i$, there exist unitaries $\{U^i_{x_ir_i}\}_{x_ir_i}$ acting on $X_{\oC}\tX_{\oC}E^\A A_{\oC}\tA_{\oC}$, and $\{V^i_{y_ir_i}\}_{y_ir_i}$ acting on $Y_{\oC}\tY_{\oC}E^\B B_{\oC}\tB_{\oC}B'_{\oC}\tB'_{\oC}$ such that
\[ \bbE_{i\in\oC}\bbE_{\sfP_{X_iY_iR_i|\clE}}\norm{U^i_{x_ir_i}\otimes V^i_{y_ir_i}\state{\vph}_{\perp r_i}(U^i_{x_ir_i})^\dagger\otimes(V^i_{y_ir_i})^\dagger - \state{\vph}_{x_iy_i\perp r_i}}_1 \leq \frac{58\sqrt{2\delta}}{\alpha}; \]
\item $\bbE_{i\in\oC}\norm{\sfP_{X_iY_iZ_iR_i|\clE}\left(\sfP_{A_iB_i|\clE,X_iY_iZ_iR_i} - \sfP_{A_iB_i|\clE,X_iY_i,\perp ,R_i}\right)}_1 \leq \dfrac{30\sqrt{2\delta}}{\alpha}$;
\item There exist unitaries $\{W^i_{y_iz_ir_i}\}_{y_iz_ir_i}$ acting on the registers $Y_{\oC}\tY_{\oC}Z_{\oC}\tZ_{\oC}E^\B B'_{\oC}\tB'_{\oC}$ such that
\[ \bbE_{i\in\oC}\bbE_{\sfP_{X_iY_iZ_iA_iB_iR_i|\clE}}\norm{\Id\otimes W^i_{y_iz_ir_i}\state{\vph}_{x_iy_i\perp a_ib_ir_i}\Id\otimes (W^i_{y_iz_ir_i})^\dagger - \state{\vph}_{x_iy_iz_ia_ib_ir_i}}_1 \leq \frac{90\sqrt{2\delta}}{\alpha}. \]
\end{enumerate}
\end{lemma}

\begin{proof}[Proof of Lemma~\ref{lem:parrep3-ind}]
Note that if $\delta \geq \alpha^2/8$, the upper bound in Lemma~\ref{lem:parrep3-ind} is trivial. So we shall only prove the lemma in the case that $\delta\leq \alpha^2/8$ using Lemma~\ref{lem:parrep3-conds}. Using Lemma~\ref{lem:parrep3-conds}, we give a strategy $\clS'$ for a single copy of $G$ as follows:
\begin{itemize}
\item Alice and Bob share $\log|\oC|$ uniform bits, for each $i\in\oC$, $\sfP_{R_i|\clE,\perp }$ as randomness, and for each $R_i=r_i$, the state $\ket{\vph}_{\perp r_i}$ as entanglement, with Alice holding registers $X_{\oC}\tX_{\oC}E^\A A_{\oC}\tA_{\oC}$ and Bob holding registers $Y_{\oC}\tY_{\oC}Z_{\oC}\tZ_{\oC}E^\B B_{\oC}\tB_{\oC}B'_{\oC}\tB'_{\oC}$.
\item Alice and Bob use their shared randomness to sample $i\in\oC$ uniformly, and in the first round, apply $U^i_{x_ir_i}$, $V^i_{y_ir_i}$ on their parts of the shared entangled state according to their shared randomness from $\sfP_{R_i|\clE,\perp }$ and their first round inputs.
\item Alice and Bob measure the $A_i,B_i$ registers of the resulting state to give their first round outputs.
\item Bob applies $W^i_{y_iz_ir_i}$ to his half of the shared entangled state after the first round according to his second round input and the shared randomness.
\item Bob measures the $B'_i$ register of the resulting state to give his second round output.
\end{itemize}
We shall first do the analysis assuming Alice and Bob have the distribution $\sfP_{X_iY_iZ_iR_i|\clE}$ exactly. Let $\sfP_{\hA_i\hB_i|X_iY_iZ_iR_i}$ denote the conditional distribution Alice and Bob get after the first round (note that $\hA_i\hB_i$ are actually independent of $Z_i$ given $X_iY_i$, but we are still writing $Z_i$ in the conditioning), and $\sfP_{\hB'_i|X_iY_iZ_iR_i\hA_i\hB_i}$ denote their conditional distribution after the second round. Since $\sfP_{\hA_i\hB_i|X_iY_iZ_iR_i}$ is obtained by measuring the $A_iB_i$ registers of the state $U^i_{x_ir_i}\otimes V^i_{y_ir_i}\ket{\vph}_{\perp r_i}$, and $\sfP_{A_iB_i|\clE,X_iY_i,\perp,R_i}$ is obtained by measuring the same registers of $\ket{\vph}_{x_iy_i\perp r_i}$, from item (ii) of Lemma~\ref{lem:parrep3-conds} and Fact \ref{fc:chan-l1} we have,
\[ \bbE_{i\in\oC}\norm{\sfP_{X_iY_iZ_iR_i|\clE}\left(\sfP_{\hA_i\hB_i|X_iY_iZ_iR_i} - \sfP_{A_iB_i|\clE,X_iY_i,\perp,R_i}\right)}_1 \leq \frac{58\sqrt{2\delta}}{\alpha}. \]
Combining this with item (iii) of the lemma we have,
\[ \bbE_{i\in\oC}\norm{\sfP_{X_iY_iZ_iR_i|\clE}\left(\sfP_{\hA_i\hB_i|X_iY_iZ_iR_i} - \sfP_{A_iB_i|\clE,X_iY_iZ_iR_i}\right)}_1 \leq \frac{58\sqrt{2\delta}}{\alpha} + \frac{30\sqrt{2\delta}}{\alpha}.\]
By similar reasoning, we have from item (iv),
\[ \bbE_{i\in\oC}\norm{\sfP_{X_iY_iZ_iR_iA_iB_i}\left(\sfP_{\hB'_i|X_iY_iZ_iR_i\hA_i\hB_i} - \sfP_{B'_i|\clE,X_iY_iZ_iR_iA_iB_i}\right)}_1 \leq \frac{90\sqrt{2\delta}}{\alpha}.\]
Combining these with item (i), then we overall have,
\[ \bbE_{i\in\oC}\norm{\sfP_{X_iY_iZ_i}\sfP_{R_i|\clE,\perp }\sfP_{\hA_i\hB_i\hB'_i|X_iY_iZ_iR_i} - \sfP_{X_iY_iZ_iR_iA_iB_iB'_i|\clE}}_1 \leq \frac{11\sqrt{2\delta}}{\alpha} + \frac{58\sqrt{2\delta}}{\alpha} + \frac{120\sqrt{2\delta}}{\alpha} \leq \frac{382\sqrt{\delta}}{\alpha}. \]
Since $\Pr[T_i=1|\clE]$ is the probability of that the distribution $\sfP_{X_iY_iZ_iR_iA_iB_iB'_i|\clE}$ wins a single copy of the game, if $\bbE_{i\in\oC}\Pr[T_i=1|\clE] > 1-\eps+\frac{191\sqrt{\delta}}{\alpha}$, then the winning probability of our constructed strategy is more than
\[ 1-\eps + \frac{191\sqrt{\delta}}{\alpha} - \frac{1}{2}\bbE_{i\in\oC}\norm{\sfP_{X_iY_iZ_i}\sfP_{R_i|\clE,\perp }\sfP_{\hA_i\hB_i\hB'_i|X_iY_iZ_iR_i} - \sfP_{X_iY_iZ_iR_iA_iB_iB'_i|\clE}}_1 \geq \omega^*(G),\]
which is a contradiction. Therefore we must have $\bbE_{i\in\oC}\Pr[T_i=1|\clE] \leq 1-\eps+\frac{191\sqrt{\delta}}{\alpha}$.
\end{proof}

\begin{proof}[Proof of Lemma~\ref{lem:parrep3-conds}]
\textbf{Closeness of distributions.} Applying Fact \ref{fc:hol-cond} with $T,V$ being trivial and $U_i = X_iY_iZ_i$ we get,
\begin{equation}
\bbE_{i \in \oC }\Vert\sfP_{X_iY_iZ_i|\clE} - \sfP_{X_iY_iZ_i}\Vert_1 \leq \frac{1}{l-|C|}\sqrt{(l-|C|)\cdot \log(1/\Pr[\clE])} \leq \sqrt{2\delta}, \label{eq:XYZ}
\end{equation}
recalling we are taking $\delta$ to be the value $\delta_C$ defined in Lemma~\ref{lem:parrep3-ind}. In particular, the last line of the above equation is obtained by recalling that we have required $\delta \leq \alpha^2/8$, which implies $|C| \leq l/2$. 

Also, applying Fact \ref{fc:hol-cond} again with $U_i$ the same, $T=X_CY_CZ_CDG$ and $V=A_CB_CB'_C$, and using $R_i = X_CY_CZ_CA_CB_CB'_CD_{-i}G_{-i}$, we get
\begin{align}
\sqrt{2\delta} & \geq \frac{1}{l-|C|}\sqrt{(l-|C|)(\log(1/\Pr[\clE]) + |C|\cdot\log(|\clA|\cdot|\clB|\cdot|\clB'|)} \nonumber \\
& \geq \bbE_{i\in\oC}\norm{\sfP_{X_iY_iZ_iX_CY_CZ_CDGA_CB_CB'_C|\clE} - \sfP_{X_CY_CZ_CA_CB_CB'_CDG|\clE}\sfP_{X_iY_iZ_i|X_CY_CZ_CDG}}_1 \nonumber \\
& = \bbE_{i\in\oC }\Vert\sfP_{X_iY_iZ_iD_iG_iR_i|\clE} - \sfP_{D_iG_iR_i|\clE}\sfP_{X_iY_iZ_i|D_iG_i}\Vert_1 \nonumber \\
 & = \frac{1}{2}\bbE_{i\in\oC }\left(\Vert\sfP_{X_iY_iZ_iR_i|\clE} - \sfP_{X_iY_iR_i|\clE}\sfP_{Z_i|X_iY_i}\Vert_1 + \Vert\sfP_{X_iY_iZ_iR_i|\clE} - \sfP_{Z_iR_i|\clE}\sfP_{X_iY_i|Z_i}\Vert_1\right), \label{eq:XYZR-1}
\end{align}
where the last line is obtained by conditioning on values $D_i=0$ and $D_i=1$.

Now, the bound~\eqref{eq:XYZ} allows us to apply item (ii) of Lemma~\ref{lem:anchor-t*}, with $X_iY_i=S$, $Z_i=T$, $R_i=R$, and the corresponding variables conditioned on $\clE$ being the primed variables in the lemma statement. This gives
\begin{align*}
\bbE_{i\in\oC}\norm{\sfP_{X_iY_iZ_iR_i|\clE} - \sfP_{X_iY_iZ_i}\sfP_{R_i|\clE,\perp }}_1  & \leq \frac{2}{\alpha}\bbE_{i\in\oC }\big(\Vert\sfP_{X_iY_iZ_iR_i|\clE} - \sfP_{X_iY_iR_i|\clE}\sfP_{Z_i|X_iY_i}\Vert_1 \nonumber \\
& \quad + \Vert\sfP_{X_iY_iZ_iR_i|\clE} - \sfP_{Z_iR_i|\clE}\sfP_{X_iY_i|Z_i}\Vert_1\big) + \frac{7}{\alpha}\bbE_{i \in \oC }\Vert\sfP_{X_iY_iZ_i|\clE} - \sfP_{X_iY_iZ_i}\Vert_1.
\end{align*}
Applying \eqref{eq:XYZ} and \eqref{eq:XYZR-1} to the terms on the right-hand side yields item (i) of the lemma.

\vspace{0.5cm}
\textbf{Existence of unitaries $U^i_{x_ir_i}$ and $V^i_{y_ir_i}$.}  We first note
\begin{align*}
& \bbE_{\sfP_{X_CY_CZ_CA_CB_CB'_CDG|\clE}}\sfD\left(\vph_{X_{\oC }Y_{\oC }\tY_{\oC }Z_{\oC }\tZ_{\oC } B_{\oC }\tB_{\oC } B'_{\oC }\tB'_{\oC }E^\B|x_Cy_Cz_Ca_Cb_Cb'_Cdg}\middle\Vert\sigma_{X_{\oC }Y_{\oC }\tY_{\oC }Z_{\oC }\tZ_{\oC } B_{\oC }\tB_{\oC } B'_{\oC }\tB'_{\oC }E^\B|x_Cy_Cz_Cdg}\right) \\
\leq & \bbE_{\sfP_{A_CB_CB'_CDG|\clE}}\sfD\left(\vph_{XY\tY Z\tZ B_{\oC }\tB_{\oC } B'_{\oC }\tB'_{\oC }E^\B|a_Cb_Cb'_Cdg}\middle\Vert\sigma_{XY\tY Z\tZ B_{\oC }\tB_{\oC } B'_{\oC }\tB'_{\oC }E^\B|dg}\right) \\
\leq & \bbE_{\sfP_{DG|\clE}}(\log(1/\gamma_{dg}) + \log(|\clA|^{|C|}\cdot|\clB|^{|C|}\cdot|\clB'|^{|C|})) \\
\leq & \log\left(1/\bbE_{\sfP_{DG|\clE}}\gamma_{dg}\right) + |C|\cdot\log(|\clA|\cdot|\clB|\cdot|\clB'|) \\
= & \log(1/\Pr[\clE]) + |C|\cdot\log(|\clA|\cdot|\clB|\cdot|\clB'|) = \delta l,
\end{align*}
where the first inequality is from \eqref{eq:CQ-S-ch}, and to get the second inequality we have used Fact \ref{fc:jpy-cond} on the states $\vph_{dg}$  and $\sigma_{dg}$ (with $z$ being $a_Cb_Cb'_C$).

Note that $\sigma_{X_{\oC }Y_{\oC }\tY_{\oC }Z_{\oC }\tZ_{\oC } E^\B B_{\oC }\tB_{\oC } B'_{\oC }\tB'_{\oC }|x_Cy_Cz_Cdg}$ is product across $X_{\oC }$ and the rest of the registers, since $dg$ is being conditioned on. Hence using Quantum Raz's Lemma,
\begin{align}
\delta l & \geq \bbE_{\sfP_{X_CY_CZ_CA_CB_CB'_CDG|\clE}}\sfD\left(\vph_{X_{\oC }Y_{\oC }\tY_{\oC }Z_{\oC }\tZ_{\oC } B_{\oC }\tB_{\oC } B'_{\oC }\tB'_{\oC }E^\B|x_Cy_Cz_Ca_Cb_Cb'_Cdg}\middle\Vert\sigma_{X_{\oC }Y_{\oC }\tY_{\oC }Z_{\oC }\tZ_{\oC } B_{\oC }\tB_{\oC } B'_{\oC }\tB'_{\oC }E^\B|x_Cy_Cz_Cdg}\right) \nonumber \\
& \geq \sum_{i\in\oC }\sfI(X_i:Y_{\oC }\tY_{\oC }Z_{\oC }\tZ_{\oC } B_{\oC }\tB_{\oC } B'_{\oC }\tB'_{\oC }E^\B|X_CY_CZ_CA_CB_CB'_CDG)_\vph \nonumber \\
& = \sum_{i\in\oC}\sfI(X_i:Y_{\oC }\tY_{\oC }Z_{\oC }\tZ_{\oC } B_{\oC }\tB_{\oC } B'_{\oC }\tB'_{\oC }E^\B|D_iG_iR_i)_\vph \nonumber \\
& = \sum_{i\in\oC} \bbE_{\sfP_{D_iG_iR_i|\clE}}\sfI(X_i:Y_{\oC }\tY_{\oC }Z_{\oC }\tZ_{\oC } B_{\oC }\tB_{\oC } B'_{\oC }\tB'_{\oC }E^\B)_{\vph_{d_ig_ir_i}}\nonumber \\
& \geq \frac{l}{2}\bbE_{i\in\oC }\bbE_{\sfP_{D_iG_iR_i|\clE}}\sfI(X_i:Y_{\oC }\tY_{\oC }Z_{\oC }\tZ_{\oC } B_{\oC }\tB_{\oC } B'_{\oC }\tB'_{\oC }E^\B)_{\vph_{d_ig_ir_i}} \nonumber \\
& \geq \frac{l}{2}\cdot\frac{1}{2}\bbE_{i\in\oC }\bbE_{\sfP_{Z_iR_i|\clE}}\sfI(X_i:Y_{\oC }\tY_{\oC }Z_{\oC }\tZ_{\oC } B_{\oC }\tB_{\oC } B'_{\oC }\tB'_{\oC }E^\B)_{\vph_{g_ir_i|D_i=1}} \nonumber \\
& = \frac{l}{2}\cdot\frac{1}{2}\bbE_{i\in\oC }\bbE_{\sfP_{Z_iR_i|\clE}}\sfI(X_i:Y_{\oC }\tY_{\oC }Z_{\oC }\tZ_{\oC } B_{\oC }\tB_{\oC } B'_{\oC }\tB'_{\oC }E^\B)_{\vph_{z_ir_i}} \nonumber \\
& = \frac{l}{4}\bbE_{i\in\oC }\bbE_{\sfP_{X_iZ_iR_i|\clE}} \sfD\left(\vph_{Y_{\oC }\tY_{\oC }Z_{\oC }\tZ_{\oC } E^\B B_{\oC }\tB_{\oC } B'_{\oC }\tB'_{\oC }|x_iz_ir_i}\middle\Vert \vph_{Y_{\oC }\tY_{\oC }Z_{\oC }\tZ_{\oC } E^\B B_{\oC }\tB_{\oC } B'_{\oC }\tB'_{\oC }|z_ir_i}\right) \nonumber \\
& \geq \frac{l}{4}\bbE_{i\in\oC }\bbE_{\sfP_{X_iZ_iR_i|\clE}} \sfB\left(\vph_{Y_{\oC }\tY_{\oC }Z_{\oC }\tZ_{\oC } E^\B B_{\oC }\tB_{\oC } B'_{\oC }\tB'_{\oC }|x_iz_ir_i}, \vph_{Y_{\oC }\tY_{\oC }Z_{\oC }\tZ_{\oC } E^\B B_{\oC }\tB_{\oC } B'_{\oC }\tB'_{\oC }|z_ir_i}\right)^2 \nonumber
\end{align}
where we have used Pinsker's inequality in the final step. Using Jensen's inequality on the above we then have,
\begin{align*}
2\sqrt{\delta} & \geq \bbE_{i\in\oC }\bbE_{\sfP_{X_iZ_iR_i|\clE}} \sfB\left(\vph_{Y_{\oC }\tY_{\oC }Z_{\oC }\tZ_{\oC } E^\B B_{\oC }\tB_{\oC } B'_{\oC }\tB'_{\oC }|x_iz_ir_i}, \vph_{Y_{\oC }\tY_{\oC }Z_{\oC }\tZ_{\oC } E^\B B_{\oC }\tB_{\oC } B'_{\oC }\tB'_{\oC }|z_ir_i}\right).
\end{align*}
Since $\sfB(\cdot,\cdot)$ always lies between 0 and 1, changing the distribution over which its expectation value is taken can only increase the expectation value by at most the $\ell_1$-distance between the distributions. This lets us write
\begin{align*}
& \bbE_{i\in\oC }\bbE_{\sfP_{Z_i}\sfP_{X_iR_i|\clE,Z_i}} \sfB\left(\vph_{Y_{\oC }\tY_{\oC }Z_{\oC }\tZ_{\oC } E^\B B_{\oC }\tB_{\oC } B'_{\oC }\tB'_{\oC }|x_iz_ir_i}, \vph_{Y_{\oC }\tY_{\oC }Z_{\oC }\tZ_{\oC } E^\B B_{\oC }\tB_{\oC } B'_{\oC }\tB'_{\oC }|z_ir_i}\right) \\
\leq & \bbE_{i\in\oC }\bbE_{\sfP_{X_iZ_iR_i|\clE}} \sfB\left(\vph_{Y_{\oC }\tY_{\oC }Z_{\oC }\tZ_{\oC } E^\B B_{\oC }\tB_{\oC } B'_{\oC }\tB'_{\oC }|x_iz_ir_i}, \vph_{Y_{\oC }\tY_{\oC }Z_{\oC }\tZ_{\oC } E^\B B_{\oC }\tB_{\oC } B'_{\oC }\tB'_{\oC }|z_ir_i}\right) \\
& \quad + \bbE_{i\in\oC}\norm{\sfP_{X_iZ_iR_i|\clE}-\sfP_{Z_i}\sfP_{X_iR_i|\clE,Z_i}}_1 \\
\leq & 2\sqrt{\delta} + \sqrt{2\delta},
\end{align*}
where to get the last step we use the fact that \eqref{eq:XYZ} implies
\begin{align*}
\bbE_{i\in\oC}\norm{\sfP_{X_iY_iZ_iR_i|\clE} - \sfP_{Z_i}\sfP_{X_iY_iR_i|\clE,Z_i}}_1 & = \bbE_{i\in\oC}\norm{\sfP_{X_iY_iR_i|\clE,Z_i}(\sfP_{Z_i|\clE} - \sfP_{Z_i})}_1 \\
& = \bbE_{i\in\oC}\norm{\sfP_{Z_i|\clE} - \sfP_{Z_i}}_1 \\
 & \leq \sqrt{2\delta}.
\end{align*}
Finally, since $\sfP_{Z_i}(\perp)=\alpha$, we have,
\[ \bbE_{i\in\oC}\bbE_{\sfP_{X_iR_i|\clE,\perp}}\sfB\left(\vph_{Y_{\oC }\tY_{\oC }Z_{\oC }\tZ_{\oC } E^\B B_{\oC }\tB_{\oC } B'_{\oC }\tB'_{\oC }|x_i\perp r_i}, \vph_{Y_{\oC }\tY_{\oC }Z_{\oC }\tZ_{\oC } E^\B B_{\oC }\tB_{\oC } B'_{\oC }\tB'_{\oC }|\perp r_i}\right) \leq \frac{2\sqrt{\delta} + \sqrt{2\delta}}{\alpha} \leq \frac{4\sqrt{\delta}}{\alpha}.\]

By Uhlmann's theorem, there exist unitaries $\{U^i_{x_ir_i}\}_{x_ir_i}$ acting on the registers $X_{\oC}\tX_{\oC}E^\A A_{\oC}\tA_{\oC}$ such that
\begin{align*}
& \bbE_{i\in\oC}\bbE_{\sfP_{X_iR_i|\clE,\perp}}\sfB\left(U^i_{x_ir_i}\otimes\Id\state{\vph}_{\perp r_i}(U^i_{x_ir_i})^\dagger\otimes\Id, \state{\vph}_{x_i\perp r_i}\right) \\
= & \bbE_{i\in\oC }\bbE_{\sfP_{X_iR_i|\clE,\perp}} \sfB\left(\vph_{Y_{\oC }\tY_{\oC }Z_{\oC }\tZ_{\oC } E^\B B_{\oC }\tB_{\oC } B'_{\oC }\tB'_{\oC }|x_i\perp r_i}, \vph_{Y_{\oC }\tY_{\oC }Z_{\oC }\tZ_{\oC } E^\B B_{\oC }\tB_{\oC } B'_{\oC }\tB'_{\oC }|\perp r_i}\right) \\
\leq & \frac{4\sqrt{\delta}}{\alpha}.
\end{align*}
Now using the Fuchs-van de Graaf inequality we get,
\begin{equation}\label{eq:X-unitary}
\bbE_{i\in\oC }\bbE_{\sfP_{X_iR_i|\clE,\perp}}\norm{U^i_{x_ir_i}\otimes\Id\state{\vph}_{\perp r_i}(U^i_{x_ir_i})^\dagger\otimes\Id - \state{\vph}_{x_i\perp r_i}}_1 \leq \frac{8\sqrt{2\delta}}{\alpha}.
\end{equation}

Similarly, $\sigma_{Y_{\oC}X_{\oC}\tX_{\oC}E^\A A_{\oC}\tA_{\oC}|x_Cy_Cz_Cdg}$ is product across $Y_{\oC}$ and the rest of the registers. Hence by using the same analysis as above, we get that there exist unitaries $\{V^i_{y_ir_i}\}_{y_ir_i}$ acting only on the registers $Y_{\oC }\tY_{\oC }Z_{\oC }\tZ_{\oC } E^\B B_{\oC }\tB_{\oC } B'_{\oC }\tB'_{\oC }$ such that,
\begin{equation}\label{eq:Y-unitary}
\bbE_{i\in\oC }\bbE_{\sfP_{Y_iR_i|\clE,\perp}}\norm{\Id\otimes V^i_{y_ir_i}\state{\vph}_{\perp r_i}\Id\otimes(V^i_{y_ir_i})^\dagger - \state{\vph}_{y_i\perp r_i}}_1 \leq \frac{8\sqrt{2\delta}}{\alpha}.
\end{equation}
Now, if $\clO_{X_i}$ is the channel that measures the $X_i$ register and records the outcome, then
\begin{align*}
\clO_{X_i}\left(\Id\otimes V^i_{y_ir_i}\state{\vph}_{\perp r_i}\Id\otimes(V^i_{y_ir_i})^\dagger\right) & = \bbE_{\sfP_{X_i|\clE,\perp r_i}}\state{x_i}\otimes\left(\Id\otimes V^i_{y_ir_i}\state{\vph}_{x_i\perp r_i}\Id\otimes(V^i_{y_ir_i})^\dagger\right) \\
\clO_{X_i}\left(\state{\vph}_{y_i\perp r_i}\right) & = \bbE_{\sfP_{X_i|\clE,y_i\perp r_i}}\state{x_i}\otimes\state{\vph}_{x_iy_i\perp r_i}.
\end{align*}
Therefore, applying Fact \ref{fc:chan-l1} to \eqref{eq:Y-unitary} with the $\clO_{X_i}$ channel we get,
\begin{align*}
& \bbE_{i\in\oC }\bbE_{\sfP_{Y_iR_i|\clE,\perp}}\norm{\bbE_{\sfP_{X_i|\clE,\perp r_i}}\state{x_i}\otimes\left(\Id\otimes V^i_{y_ir_i}\state{\vph}_{x_i\perp r_i}\Id\otimes(V^i_{y_ir_i})^\dagger\right) - \bbE_{\sfP_{X_i|\clE,y_i\perp r_i}}\state{x_i}\otimes\state{\vph}_{x_iy_i\perp r_i}}_1 \\
\leq & \frac{8\sqrt{2\delta}}{\alpha}.
\end{align*}
From this we have,
\begin{align}
& \bbE_{i\in\oC }\bbE_{\sfP_{X_iY_iR_i|\clE,\perp}}\norm{\Id\otimes V^i_{y_ir_i}\state{\vph}_{x_i\perp r_i}\Id\otimes(V^i_{y_ir_i})^\dagger - \state{\vph}_{x_iy_i\perp r_i}}_1 \nonumber \\
= & \bbE_{i\in\oC }\bbE_{\sfP_{Y_iR_i|\clE,\perp}}\norm{\bbE_{\sfP_{X_i|\clE,y_i\perp r_i}}\state{x_i}\otimes\left(\Id\otimes V^i_{y_ir_i}\state{\vph}_{x_i\perp r_i}\Id\otimes(V^i_{y_ir_i})^\dagger - \state{\vph}_{x_iy_i\perp r_i}\right)}_1 \nonumber \\
\leq & \bbE_{i\in\oC }\bbE_{\sfP_{Y_iR_i|\clE,\perp}}\norm{\bbE_{\sfP_{X_i|\clE,\perp r_i}}\state{x_i}\otimes\left(\Id\otimes V^i_{y_ir_i}\state{\vph}_{x_i\perp r_i}\Id\otimes(V^i_{y_ir_i})^\dagger\right) - \bbE_{\sfP_{X_i|\clE,y_i\perp r_i}}\state{x_i}\otimes\state{\vph}_{x_iy_i\perp r_i}}_1 \nonumber \\
& \quad + \bbE_{i\in\oC }\bbE_{\sfP_{Y_iR_i|\clE,\perp}}\norm{\left(\bbE_{\sfP_{X_i|\clE,y_i\perp r_i}}\state{x_i} - \bbE_{\sfP_{X_i|\clE,\perp r_i}}\state{x_i}\right)\otimes\left(\Id\otimes V^i_{y_ir_i}\state{\vph}_{x_i\perp r_i}\Id\otimes(V^i_{y_ir_i})^\dagger\right)}_1 \nonumber \\
\leq & \frac{8\sqrt{2\delta}}{\alpha} + 2\bbE_{i\in\oC}\norm{\sfP_{X_iY_iR_i|\clE,\perp } - \sfP_{R_i|\clE,\perp }\sfP_{X_i|\clE,\perp ,R_i}\sfP_{Y_i|\clE,\perp ,R_i}}_1. \label{eq:Y-unitary-2}
\end{align}
Combining equations \eqref{eq:X-unitary} and \eqref{eq:Y-unitary-2} we then get,
\begin{align}
& \bbE_{i\in\oC}\bbE_{\sfP_{X_iY_iR_i|\clE,\perp }}\norm{U^i_{x_ir_i}\otimes V^i_{y_ir_i}\state{\vph}_{\perp r_i}(U^i_{x_ir_i})^\dagger\otimes(V^i_{y_ir_i})^\dagger - \state{\vph}_{x_iy_i\perp r_i}}_1 \nonumber \\
\leq & \bbE_{i\in\oC}\bbE_{\sfP_{X_iY_iR_i|\clE,\perp }}\norm{\Id\otimes V^i_{y_ir_i}\left(U^i_{x_ir_i}\otimes\Id\state{\vph}_{\perp r_i}(U^i_{x_ir_i})^\dagger\otimes\Id - \state{\vph}_{x_i\perp r_i}\right)\Id\otimes (V^i_{y_ir_i})^\dagger}_1  \nonumber \\
& \quad + \bbE_{i\in\oC}\bbE_{\sfP_{X_iY_iR_i|\clE,\perp }}\norm{\Id\otimes V^i_{y_ir_i}\state{\vph}_{x_i\perp r_i}\Id\otimes(V^i_{y_ir_i})^\dagger - \state{\vph}_{x_iy_i\perp r_i}}_1 \nonumber \\
= & \bbE_{i\in\oC}\bbE_{\sfP_{X_iY_iR_i|\clE,\perp }}\norm{U^i_{x_ir_i}\otimes\Id\state{\vph}_{\perp r_i}(U^i_{x_ir_i})^\dagger\otimes\Id - \state{\vph}_{x_i\perp r_i}}_1 \nonumber \\
& \quad + \bbE_{i\in\oC }\bbE_{\sfP_{X_iY_iR_i|\clE,\perp}}\norm{\Id\otimes V^i_{y_ir_i}\state{\vph}_{x_i\perp r_i}\Id\otimes(V^i_{y_ir_i})^\dagger - \state{\vph}_{x_iy_i\perp r_i}}_1 \nonumber \\
\leq & \frac{16\sqrt{2\delta}}{\alpha} + 2\bbE_{i\in\oC}\norm{\sfP_{X_iY_iR_i|\clE,\perp } - \sfP_{R_i|\clE,\perp }\sfP_{X_i|\clE,\perp ,R_i}\sfP_{Y_i|\clE,\perp ,R_i}}_1. \label{eq:XY-UV-1}
\end{align}
Now note that we have an upper bound of $2\sqrt{2\delta}$ on $\bbE_{i\in\oC}\norm{\sfP_{X_iY_iZ_iR_i|\clE} - \sfP_{X_iY_i|Z_i}\sfP_{Z_iR_i|\clE}}_1$. Thus, by Fact \ref{fc:cond-prob},
\begin{align*}
\bbE_{i\in\oC}\norm{\sfP_{X_iY_iR_i|\clE,\perp } - \sfP_{X_iY_i|\perp }\sfP_{R_i|\clE,\perp }}_1 & \leq \frac{2}{\sfP_{Z_i}(\perp )}\bbE_{i\in\oC}\norm{\sfP_{X_iY_iZ_iR_i|\clE} - \sfP_{X_iY_i|Z_i}\sfP_{Z_iR_i|\clE}}_1 \\
 & \leq \frac{4\sqrt{2\delta}}{\alpha}.
\end{align*}
Using this we get, and the fact that $\sfP_{X_iY_i|\perp}=\sfP_{X_iY_i}=\sfP_{X_i}\sfP_{Y_i}=\sfP_{X_i|\perp}\sfP_{Y_i|\perp}$, we get,
\begin{align*}
& \bbE_{i\in\oC}\left\Vert\sfP_{X_iY_iR_i|\clE,\perp } - \sfP_{R_i|\clE,\perp }\sfP_{X_i|\clE,\perp ,R_i}\sfP_{Y_i|\clE,\perp ,R_i}\right\Vert_1 \nonumber \\
\leq & \bbE_{i\in\oC}\left(\left\Vert\sfP_{X_iY_iR_i|\clE,\perp } - \sfP_{X_iY_i|\perp }\sfP_{R_i|\clE,\perp }\right\Vert_1 + \left\Vert\sfP_{X_i|\perp }\sfP_{Y_i|\perp }\sfP_{R_i|\clE,\perp } - \sfP_{Y_iR_i|\clE,\perp }\sfP_{X_i|\clE,\perp ,R_i}\right\Vert_1\right) \nonumber \\
\leq & \frac{4\sqrt{2\delta}}{\alpha} + \bbE_{i\in\oC}\left(\left\Vert(\sfP_{X_i|\perp }\sfP_{R_i|\clE,\perp } - \sfP_{X_iR_i|\clE,\perp })\sfP_{Y_i|\perp }\right\Vert_1 + \left\Vert\sfP_{X_i|\clE,\perp ,R_i}(\sfP_{Y_i|\perp }\sfP_{R_i|\clE,\perp } - \sfP_{Y_iR_i|\clE,\perp })\right\Vert_1\right) \nonumber \\
\leq & \frac{4\sqrt{2\delta}}{\alpha} + \bbE_{i\in\oC}\left(\left\Vert\sfP_{X_i|\perp }\sfP_{R_i|\clE,\perp } - \sfP_{X_iR_i|\clE,\perp }\right\Vert_1 + \left\Vert\sfP_{Y_i|\perp }\sfP_{R_i|\clE,\perp } - \sfP_{Y_iR_i|\clE,\perp }\right\Vert_1\right) \nonumber \\
\leq & \frac{4\sqrt{2\delta}}{\alpha} + 2\bbE_{i\in\oC} \left\Vert \sfP_{X_iY_i|\perp }\sfP_{R_i|\clE,\perp } - \sfP_{X_iY_iR_i|\clE,\perp } \right\Vert_1 \leq \frac{12\sqrt{2\delta}}{\alpha}.
\end{align*}
In the fourth line of the above calculation, we have upper bounded $\left\Vert\sfP_{X_i|\perp }\sfP_{R_i|\clE,\perp } - \sfP_{X_iR_i|\clE,\perp }\right\Vert_1$ by $\left\Vert \sfP_{X_iY_i|\perp }\sfP_{R_i|\clE,\perp } - \sfP_{X_iY_iR_i|\clE,\perp } \right\Vert_1$ since we get the distributions $\sfP_{X_i|\perp}\sfP_{R_i|\clE,\perp}$ and $\sfP_{X_iR_i|\clE,\perp}$ respectively by tracing out the $Y_i$ registers of $\sfP_{X_iY_i|\perp }\sfP_{R_i|\clE,\perp }$ and $\sfP_{X_iY_iR_i|\clE,\perp }$; also, we have upper bounded $\left\Vert\sfP_{Y_i|\perp }\sfP_{R_i|\clE,\perp } - \sfP_{Y_iR_i|\clE,\perp }\right\Vert_1$ by similar reasoning. Putting the above bound in \eqref{eq:XY-UV-1} we get,
\[
\bbE_{i\in\oC}\bbE_{\sfP_{X_iY_iR_i|\clE,\perp }}\norm{U^i_{x_ir_i}\otimes V^i_{y_ir_i}\state{\vph}_{\perp r_i}(U^i_{x_ir_i})^\dagger\otimes(V^i_{y_ir_i})^\dagger - \state{\vph}_{x_iy_i\perp r_i}}_1 \leq \frac{16\sqrt{2\delta}}{\alpha} + \frac{24\sqrt{2\delta}}{\alpha}.
\]
Now, the quantity we actually want to bound is almost the same as the above expression, except we have to take the expectation over the distribution $\sfP_{X_iY_iR_i|\clE}$ rather than $\sfP_{X_iY_iR_i|\clE,\perp }$. Since the $\ell_1$-distance term always lies between 0 and 2, changing the distribution over which the expectation is taken can only increase the expectation value by at most 2 times the distance between the two distributions. Thus we can write
\begin{align}
& \bbE_{i\in\oC}\bbE_{\sfP_{X_iY_iR_i|\clE}}\norm{U^i_{x_ir_i}\otimes V^i_{y_ir_i}\otimes\Id\state{\vph}_{\perp r_i}(U^i_{x_ir_i})^\dagger\otimes(V^i_{y_ir_i})^\dagger\otimes\Id - \state{\vph}_{x_iy_i\perp r_i}}_1 \nonumber \\
\leq & \bbE_{i\in\oC}\bbE_{\sfP_{X_iY_iR_i|\clE,\perp }}\norm{U^i_{x_ir_i}\otimes V^i_{y_ir_i}\state{\vph}_{\perp r_i}(U^i_{x_ir_i})^\dagger\otimes(V^i_{y_ir_i})^\dagger - \state{\vph}_{x_iy_i\perp r_i}}_1 \nonumber \\
& \quad + 2\bbE_{i\in\oC}\norm{\sfP_{X_iY_iR_i|\clE} - \sfP_{X_iY_iR_i|\clE,\perp} }_1 \nonumber \\
\leq & 
\frac{40\sqrt{2\delta}}{\alpha} + \frac{4}{\alpha} \bbE_{i\in\oC} \norm{ \sfP_{X_iY_iZ_iR_i|\clE} - \sfP_{X_iY_iR_i|\clE} \sfP_{Z_i|X_iY_i} }_1 + 
\frac{10}{\alpha} \bbE_{i\in\oC} \norm{ \sfP_{X_iY_iZ_i|\clE} - \sfP_{X_iY_iZ_i} }_1 \nonumber \\
\leq & \frac{40\sqrt{2\delta}}{\alpha} + \frac{18\sqrt{2\delta}}{\alpha} = \frac{58\sqrt{2\delta}}{\alpha}, \label{eq:XY-UV-2}
\end{align}
where to get the third line we used item (i) of Lemma~\ref{lem:anchor-t*} to bound $\norm{\sfP_{X_iY_iR_i|\clE} - \sfP_{X_iY_iR_i|\clE,\perp }}_1$, and in the next line 
we have used \eqref{eq:XYZR-1} and \eqref{eq:XYZ} to bound the trace distances. This proves item (ii) of the lemma.

\vspace{0.5cm}
\textbf{Existence of unitaries $W^i_{y_iz_ir_i}$.} Note that $\sigma_{Z_{\oC}X_{\oC}\tX_{\oC}A_{\oC}\tA_{\oC}B_{\oC}\tB_{\oC}E^\A|dg}$ is product across $Z_{\oC}$ and the rest of the registers, since they were product in $\rho$, and $\sigma$ is obtained from $\rho$ by a unitary that acts on the other registers, only using $Z_{\oC}$ as a control register (this is also true if we include $Y_{\oC}\tY_{\oC}$ along with $X_{\oC}\tX_{\oC}A_{\oC}\tA_{\oC}B_{\oC}\tB_{\oC}E^\A$, but we don't need to do this). Therefore, by the same analysis as in the case of $X_{\oC}$ and $Y_{\oC}$,
\begin{align*}
2\delta & \geq \bbE_{i\in\oC}\bbE_{\sfP_{Z_iD_iG_iR_i|\clE}}\sfD\left(\vph_{X_{\oC}\tX_{\oC}A_{\oC}\tA_{\oC}B_{\oC}\tB_{\oC}E^\A|z_id_ig_ir_i}\middle\Vert\vph_{X_{\oC}\tX_{\oC}A_{\oC}\tA_{\oC}B_{\oC}\tB_{\oC}E^\A|d_ig_ir_i}\right) \\
 & \geq \frac{1}{2}\bbE_{i\in\oC}\bbE_{\sfP_{X_iY_iZ_iR_i|\clE}}\sfD\left(\vph_{X_{\oC}\tX_{\oC}A_{\oC}\tA_{\oC}B_{\oC}\tB_{\oC}E^\A|x_iy_iz_ir_i}\middle\Vert\vph_{X_{\oC}\tX_{\oC}A_{\oC}\tA_{\oC}B_{\oC}\tB_{\oC}E^\A|x_iy_ir_i}\right) \\
 & \geq \frac{1}{2}\bbE_{i\in\oC}\bbE_{\sfP_{X_iY_iZ_iR_i|\clE}}\sfB\left(\vph_{X_{\oC}\tX_{\oC}A_{\oC}\tA_{\oC}B_{\oC}\tB_{\oC}E^\A|x_iy_iz_ir_i},\vph_{X_{\oC}\tX_{\oC}A_{\oC}\tA_{\oC}B_{\oC}\tB_{\oC}E^\A|x_iy_ir_i}\right)^2.
\end{align*}
Using Jensen's inequality on the last inequality, we get,
\begin{align}
\bbE_{i\in\oC}\bbE_{\sfP_{X_iY_iZ_iR_i|\clE}}\sfB\left(\vph_{X_{\oC}\tX_{\oC}A_{\oC}\tA_{\oC}B_{\oC}\tB_{\oC}E^\A|x_iy_iz_ir_i},\vph_{X_{\oC}\tX_{\oC}A_{\oC}\tA_{\oC}B_{\oC}\tB_{\oC}E^\A|x_iy_ir_i}\right) & \leq 2\sqrt{\delta}. \label{eq:XY-margin}
\end{align}
Moreover, shifting the expectation from $\sfP_{X_iY_iZ_iR_i|\clE}$ to $\sfP_{Z_i}\sfP_{X_iY_iR_i|\clE,Z_i}$ and conditioning on $Z_i=\perp$ (which happens with probability $\alpha$ under $\sfP_{Z_i}$), like we did when showing the existence of $U^i_{x_ir_i}, V^i_{y_ir_i}$, we get,
\begin{align}
 \bbE_{i\in\oC}\bbE_{\sfP_{X_iY_iR_i|\clE,\perp }}\sfB\left(\vph_{X_{\oC}\tX_{\oC}A_{\oC}\tA_{\oC}B_{\oC}\tB_{\oC}E^\A|x_iy_i\perp r_i},\vph_{X_{\oC}\tX_{\oC}A_{\oC}\tA_{\oC}B_{\oC}\tB_{\oC}E^\A|x_iy_ir_i}\right) & \leq \frac{4\sqrt{\delta}}{\alpha}. \label{eq:XY-perp}
\end{align}

Now using the triangle inequality on \eqref{eq:XY-margin} and \eqref{eq:XY-perp}, we get,
\begin{align}
& \bbE_{i\in\oC}\bbE_{\sfP_{X_iY_iZ_iR_i|\clE}}\sfB\left(\vph_{X_{\oC}\tX_{\oC}A_{\oC}\tA_{\oC}B_{\oC}\tB_{\oC}E^\A|x_iy_iz_ir_i},\vph_{X_{\oC}\tX_{\oC}A_{\oC}\tA_{\oC}B_{\oC}\tB_{\oC}E^\A|x_iy_i\perp r_i}\right) \nonumber \\
\leq & \bbE_{i\in\oC}\bbE_{\sfP_{X_iY_iZ_iR_i|\clE}}\sfB\left(\vph_{X_{\oC}\tX_{\oC}A_{\oC}\tA_{\oC}B_{\oC}\tB_{\oC}E^\A|x_iy_iz_ir_i},\vph_{X_{\oC}\tX_{\oC}A_{\oC}\tA_{\oC}B_{\oC}\tB_{\oC}E^\A|x_iy_ir_i}\right) \nonumber \\
& \quad + \bbE_{i\in\oC}\bbE_{\sfP_{X_iY_iR_i|\clE,\perp }}\sfB\left(\vph_{X_{\oC}\tX_{\oC}A_{\oC}\tA_{\oC}B_{\oC}\tB_{\oC}E^\A|x_iy_i\perp r_i},\vph_{X_{\oC}\tX_{\oC}A_{\oC}\tA_{\oC}B_{\oC}\tB_{\oC}E^\A|x_iy_ir_i}\right) \nonumber \\
& \quad + \bbE_{i\in\oC}\norm{(\sfP_{X_iY_iR_i|\clE}-\sfP_{X_iY_iR_i|\clE,\perp})\sfP_{Z_i|\clE,X_iY_iR_i}}_1 \nonumber \\
\leq & 2\sqrt{\delta} + \frac{4\sqrt{\delta}}{\alpha} + \frac{2}{\alpha}\bbE_{i\in\oC}\norm{\sfP_{X_iY_iR_iZ_i|\clE} - \sfP_{X_iY_iR_i|\clE}\sfP_{Z_i|X_iY_i}}_1 + \frac{5}{\alpha}\bbE_{i\in\oC}\norm{\sfP_{X_iY_iZ_i|\clE}-\sfP_{X_iY_iZ_i}}_1 \nonumber \\
\leq & 2\sqrt{\delta} + \frac{4\sqrt{\delta}}{\alpha} + \frac{4\sqrt{2\delta}}{\alpha} + \frac{5\sqrt{2\delta}}{\alpha} \nonumber \\
\leq & \frac{15\sqrt{\delta}}{\alpha}. \label{eq:B-Z-perp}
\end{align}
In the third line of the above calculation, we have noted that $\norm{(\sfP_{X_iY_iR_i|\clE}-\sfP_{X_iY_iR_i|\clE,\perp})\sfP_{Z_i|\clE,X_iY_iR_i}}_1 = \norm{\sfP_{X_iY_iR_i|\clE}-\sfP_{X_iY_iR_i|\clE,\perp}}_1$, and we have used item (i) of Lemma~\ref{lem:anchor-t*} (with $X_iY_i=S, Z_i=T, R_i=R$, and the conditioned variables being the corresponding primed variables) to bound the latter. In the fourth line, we have used \eqref{eq:XYZR-1} and \eqref{eq:XYZ} to bound the trace distances.

Applying the Fuchs-van de Graaf inequality on \eqref{eq:B-Z-perp} and tracing out registers besides $A_iB_i$ gives us
\[ \bbE_{i\in\oC}\bbE_{\sfP_{X_iY_iZ_iR_i|\clE}}\norm{\vph_{A_iB_i|x_iy_iz_ir_i} - \vph_{A_iB_i|x_iy_i\perp r_i}}_1 \leq \frac{30\sqrt{2\delta}}{\alpha}.\]
Since the $A_iB_i$ registers are classical, the trace distance in the above expression can be interpreted as the distance between the distributions $\sfP_{A_iB_i|\clE, x_iy_iz_ir_i}$ and $\sfP_{A_iB_i|\clE,x_iy_i\perp r_i}$. Therefore we have,
\begin{align*}
\bbE_{i\in\oC}\norm{\sfP_{X_iY_iZ_iR_i|\clE}(\sfP_{A_iB_i|\clE,X_iY_iZ_iR_i} - \sfP_{A_iB_i|\clE,X_iY_i,\perp,R_i})}_1 & = \bbE_{i\in\oC}\bbE_{\sfP_{X_iY_iZ_iR_i|\clE}}\norm{\sfP_{A_iB_i|\clE,x_iy_iz_ir_i} - \sfP_{A_iB_i|\clE,x_iy_i\perp r_i}}_1 \\
& \leq \frac{30\sqrt{2\delta}}{\alpha},
\end{align*}
which is item (iii) of the lemma.

To get item (iv) of the lemma, we use the Fuchs-van de Graaf inequality on \eqref{eq:B-Z-perp} but don't trace out any registers, and then we use Uhlmann's theorem as before, which gives us unitaries $\{W^i_{x_iy_iz_ir_i}\}_{x_iy_iz_ir_i}$ acting on $Y_{\oC}\tY_{\oC}Z_{\oC}\tZ_{\oC}E^\B B'_{\oC}\tB'_{\oC}$ such that
\begin{equation}\label{eq:Z-unitary-1}
\bbE_{i\in\oC}\bbE_{\sfP_{X_iY_iZ_iR_i|\clE}}\norm{\Id\otimes W^i_{x_iy_iz_ir_i}\state{\vph}_{x_iy_i,\perp r_i}\Id\otimes(W^i_{x_iy_iz_ir_i})^\dagger - \state{\vph}_{x_iy_iz_ir_i}}_1 \leq \frac{30\sqrt{2\delta}}{\alpha}.
\end{equation}
Since $z_i$ is either $\perp $, in which case the unitary $W^i_{x_iy_iz_ir_i}$ can just be identity, or $z_i$ contains $x_i$, $W^i_{x_iy_iz_ir_i}$ is in fact just $W^i_{y_iz_ir_i}$.

Let $\clO_{A_iB_i}$ be the channel that measures the $A_iB_i$ registers and records the outcomes. This clearly commutes with $W^i_{y_iz_ir_i}$. Therefore,
\begin{align*}
& \clO_{A_iB_i}\left(\Id\otimes W^i_{y_iz_ir_i}\state{\vph}_{x_iy_i\perp r_i}\Id\otimes(W^i_{y_iz_ir_i})^\dagger\right) \\
= & \bbE_{\sfP_{A_iB_i|\clE,x_iy_i\perp r_i}}\state{a_ib_i}\otimes\left(\Id\otimes W^i_{y_iz_ir_i}\state{\vph}_{x_iy_i\perp a_ib_ir_i}\Id\otimes(W^i_{y_iz_ir_i})^\dagger\right) \\
& \clO_{A_iB_i}\left(\state{\vph}_{x_iy_iz_ir_i}\right) \\
= & \bbE_{\sfP_{A_iB_i|x_iy_iz_ir_i}}\state{a_ib_i}\otimes\state{\vph}_{x_iy_iz_ia_ib_ir_i}
\end{align*}
Using this and Fact \ref{fc:chan-l1} on \eqref{eq:Z-unitary-1} along with item (iii) we get,
\begin{align*}
& \bbE_{i\in\oC}\bbE_{\sfP_{X_iY_iZ_iA_iB_iR_i|\clE}}\norm{\Id\otimes W^i_{y_iz_ir_i}\state{\vph}_{x_iy_i\perp a_ib_ir_i}\Id\otimes(W^i_{y_iz_ir_i})^\dagger - \state{\vph}_{x_iy_iz_ia_ib_ir_i}}_1 \\
\leq & \frac{30\sqrt{2\delta}}{\alpha} + 2\bbE_{i\in\oC}\norm{\sfP_{X_iY_iZ_iR_i|\clE}\left(\sfP_{A_iB_i|\clE,X_iY_iZ_iR_i} - \sfP_{A_iB_i|\clE,X_iY_i,\perp ,R_i}\right)}_1 \\
\leq & \frac{90\sqrt{2\delta}}{\alpha}
\end{align*} 
This proves item (iv).
\end{proof}

\subsection{Parallel repetition theorem for 1-round 3-player product-anchored game}
We call a 1-round 3-player product-anchored iff Alice and Bob's marginal input distribution is a product distribution, and Eve's input takes value either $z=(x,y)$ or $z=\perp$ such that $p(x,y,\perp) = \alpha\cdot p(x,y)$. Note that this is also a special type of anchoring on Charlie's side, but this definition will be sufficient for our purposes.
\begin{theorem}\label{thm:3p-parrep}
Let $G$ be a 1-round 3-player non-local product-anchored game with parameter $\alpha$. Then for $\delta>0$ and $t = (\omega^*(G)+\eta)l$,
\begin{align*}
\omega^*(G^l) & = \left(1-(1-\omega^*(G))^3\right)^{\Omega\left(\frac{\alpha^2l}{\log(|\clA|\cdot|\clB|\cdot|\clC|)}\right)} \\
\omega^*(G^{t/l}) & = \left(1-\eta^3\right)^{\Omega\left(\frac{\alpha^2l}{\log(|\clA|\cdot|\clB|\cdot|\clC|)}\right)}.
\end{align*}
\end{theorem}
\begin{proof}[Proof sketch]
Defining the correlation-breaking variables $D_iG_i$ the same way as in the 2-player 2-round case, we have that conditioned on $DG=dg$, Alice's inputs are in product with Bob and Eve's systems in the state of a strategy for $l$ copies of $G$; the analogous statements hold for Bob and Eve's inputs as well. As in the case of the 2-round game, we condition on the success event $\clE$ on a subset $C$, and define $\ket{\vph}$ to be the state of the protocol conditioned on $\clE$. Defining $R_i$ and the quantity $\delta$ the same way, the lemma analogous to Lemma~\ref{lem:parrep3-conds} in this case is the following.
\begin{lemma}\label{lem:parrep2-conds}
If $\delta = O(\eps^2\alpha^2)$, the following conditions hold:
\begin{enumerate}[(i)]
\item $\bbE_{i\in\oC}\norm{\sfP_{X_iY_iZ_iR_i|\clE} - \sfP_{X_iY_iZ_i}\sfP_{R_i|\clE,\perp}}_1 \leq \frac{9\sqrt{2\delta}}{\alpha}$;
\item For each $i\in\oC$, there exist unitaries ${U^i_{x_ir_i}}_{x_ir_i}, \{V^i_{y_ir_i}\}_{y_ir_i}, \{W^i_{z_ir_i}\}_{z_ir_i}$ acting on Alice, Bob and Charlie's systems such that
\begin{align*}
& \bbE_{i\in\oC}\bbE_{\sfP_{X_iY_iZ_iR_i|\clE}}\norm{U^i_{x_ir_i}\otimes V^i_{y_ir_i}\otimes W^i_{z_ir_i}\state{\vph}_{\perp r_i}(U^i_{x_ir_i})^\dagger\otimes(V^i_{y_ir_i})^\dagger\otimes(W^i_{z_ir_i})^\dagger - \state{\vph}_{x_iy_iz_ir_i}}_1 \\
\leq & \frac{88\sqrt{2\delta}}{\alpha}.
\end{align*}
\end{enumerate}
\end{lemma}
The proof of item (i) in Lemma~\ref{lem:parrep2-conds} is exactly the same as in Lemma~\ref{lem:parrep3-conds}. The existence of unitaries $U^i_{x_ir_i}$ and $V^i_{y_ir_i}$ in item (ii) is shown in exactly the same way. The existence of Eve's unitaries $W^i_{z_ir_i}$ is shown in a way similar to the existence of the second round unitaries in Lemma~\ref{lem:parrep3-conds}, and using the fact that Eve's inputs are in product with Alice and Bob's systems in the original state, conditioned on $dg$.

Using Lemma~\ref{lem:parrep2-conds}, we can give a strategy for a single copy of $G$ where Alice, Bob and Eve share $\sfP_{R_i|\clE,\perp}$ as randomness and $\ket{\vph}_{\perp r_i}$ as entanglement, and apply the unitaries $U^i_{x_ir_i}, V^i_{y_ir_i}$ and $W^i_{z_ir_i}$ on receiving inputs $(x_i,y_i,z_i)$. 
\end{proof}

\section*{Acknowledgements}

We thank Anne Broadbent for discussions on the result in~\cite{BI19}, as well as
L{\'{\i}}dia del Rio, Christopher Portmann, Renato Renner and Vilasini Venkatesh for discussions on composable security. We also thank Serge Fehr and anonymous reviewers on an earlier version of this manuscript for pointing out the additional required conditions on Bob's boxes.

This work was prepared while S.~K.~was at the Centre for Quantum Technologies (National University of Singapore), supported by the National Research Foundation, including under NRF RF Award No. NRF-NRFF2013-13, the Prime Minister's Office, Singapore; the Ministry of Education, Singapore, under the Research Centres of Excellence program and by Grant No. MOE2012-T3-1-009; and in part by the NRF2017-NRF-ANR004 VanQuTe Grant.
E.~Y.-Z.~T.~was at the Institute for Theoretical Physics (ETH Z\"{u}rich), supported by the Swiss National Science Foundation (SNSF) grant number 20QT21\_187724 via the National Center for Competence in Research for Quantum Science and Technology (QSIT), the Air Force Office of Scientific Research (AFOSR) via grant FA9550-19-1-0202, and the QuantERA project eDICT.

\printbibliography

\end{document}